\date{}
\documentclass[a4paper,twoside,11pt,leqno]{article}
\usepackage{amsmath,amssymb,amsopn,amsthm,mathrsfs,array,hhline,epsfig,slant,
  mathtools,scalerel}
\usepackage{changebar}
\usepackage{textcomp}    
\usepackage{anyfontsize,t1enc}    
\usepackage{setspace}
\usepackage{mathdots}
\usepackage{authblk}
\usepackage[monochrome]{color}
\usepackage[
dvips,
a4paper,
top=60pt,
bottom=55pt, 
left=80pt, 
right=80pt,
footskip=25pt,
includefoot,
heightrounded, 
]{geometry} 
\usepackage[matha,mathb,mathx]{mathabx}
\usepackage[all]{xy}
\usepackage[normalem]{ulem}
\usepackage{snotez}
\definecolor{colorA}{rgb}{0.0, 0.5, 0.9}
\definecolor{colorB}{rgb}{0.0, 0.42, 0.24}
\definecolor{colorC}{rgb}{0.0, 0.0, 0.9}
\definecolor{colorD}{cmyk}{ 0.0, 0.62, 0.75, 0.0 }
\definecolor{orange}{cmyk}{0.00, 0.35, 1.00, 0.00}
\definecolor{deepgreen}{cmyk}{1.00, 0.00, 1.00, 0.20}
\makeatletter
\def\thm@space@setup{\thm@preskip=5pt
\thm@postskip=5pt}
\makeatother
\makeatletter
\renewenvironment{proof}[1][\proofname]%
{\par\pushQED{\qed}%
 \normalfont \partopsep=\z@skip \topsep=\z@skip
  \trivlist
\item[\hskip\labelsep
 \itshape #1\@addpunct{.}]\ignorespaces}%
{\popQED\endtrivlist\@endpefalse}
\makeatother
\usepackage{titlesec}
\titlespacing*{\section}{0pt}{5pt}{7pt}
\titlespacing*{\subsection}{0pt}{8pt}{5pt}

\numberwithin{equation}{section}
\newtheorem{definition}[equation]{Definition}

\newtheorem{theorem}[equation]{Theorem}
\newtheorem{proposition}[equation]{Proposition}
\newtheorem{prop-def}[equation]{Proposition with Definition}
\newtheorem{corollary}[equation]{Corollary}
\newtheorem{conjecture}[equation]{Conjecture}
\newtheorem{remark}[equation]{Remark}
\newtheorem{lemma}[equation]{Lemma}
\newtheorem{example}[equation]{Example}

\newtheorem{working-hypothesis}[equation]{Working Hypothesis}
\newcommand{\tp}[1]{\prescript{t}{}{\kern1pt\!#1}} 
\newcommand{\rev}[1]{\prescript{\mathrm{rev}}{}{\hskip-0.5pt#1}} 

\newcommand{\strutf}[2]{\vrule height #1 depth #2 width 0pt}
\newcommand{\wt}{\mathrm{wt}}
\newcommand{\rslt}{\mathrm{rslt}}
\renewcommand{\leq}{\varleq}
\renewcommand{\geq}{\vargeq}

\renewcommand{\vec}[1]{\ensuremath{\mathchoice
                    {\mbox{\boldmath$\displaystyle#1$}}
                    {\mbox{\boldmath$\textstyle#1$}}
                    {\mbox{\boldmath$\scriptstyle#1$}}
                    {\mbox{\boldmath$\scriptscriptstyle#1$}}}}%
\newcommand{\iu}{\vec{i}}
\newcommand{\Res}{\mathrm{Res}}
\newcommand{\Der}{\mathrm{Der}}
\newcommand{\ph}{\mathrm{ph}}
\usepackage[hidelinks]{hyperref}
\hypersetup{%
 setpagesize=false,%
 pdftitle={},%
 pdfsubject={},%
 pdfauthor={},%
 pdfkeywords={}%
 colorlinks,%
 citecolor=black,%
 filecolor=black,%
 linkcolor=black,%
 urlcolor=black%
}%
\makeatletter
\long\def\@makecaption#1#2{%
  \vskip\abovecaptionskip
  \iftdir\sbox\@tempboxa{#1\hskip1zw#2}%
    \else\sbox\@tempboxa{#1\hskip1zw #2}%
  \fi
  \ifdim \wd\@tempboxa >\hsize
    \iftdir #1\hskip1zw#2\relax\par
      \else #1\hskip1zw #2\relax\par\fi
  \else
    \global \@minipagefalse
    \hbox to\hsize{\hfil\box\@tempboxa\hfil}%
  \fi
  \vskip\belowcaptionskip}
\def\tbcaption{\addtocounter{equation}{1}\def\@captype{table}\caption}
\def\fgcaption{\def\@captype{figure}\caption}
\makeatother
\newcounter{item}
\newenvironment{oitem}{%
\begin{list}{{\rm (\arabic{item})}}
{\usecounter{item}
 \setlength{\topsep}{0pt}
 \setlength{\parskip}{0pt}
 \setlength{\partopsep}{0pt}
 \setlength{\leftmargin}{20pt}
 \setlength{\labelsep}{5pt}
 \setlength{\labelwidth}{30pt}
 \setlength{\rightmargin}{0pt}
 \setlength{\parsep}{0.0pt}
 \setlength{\itemsep}{0.0pt}
 \setlength{\itemindent}{0pt}}}
{\end{list}}
\newenvironment{Aitem}{%
\begin{list}{{\rm (A\arabic{item})}}
{\usecounter{item}
 \setlength{\topsep}{0pt}
 \setlength{\parskip}{0pt}
 \setlength{\partopsep}{0pt}
 \setlength{\leftmargin}{28pt}
 \setlength{\labelsep}{5pt}
 \setlength{\labelwidth}{30pt}
 \setlength{\rightmargin}{0pt}
 \setlength{\parsep}{0.0pt}
 \setlength{\itemsep}{0.0pt}
 \setlength{\itemindent}{0pt}}}
{\end{list}}
 \AtBeginDocument{%
   \abovedisplayskip     =0.5\abovedisplayskip
   \abovedisplayshortskip=0.5\abovedisplayshortskip
   \belowdisplayskip     =0.5\belowdisplayskip
   \belowdisplayshortskip=0.6\belowdisplayshortskip}
\title{Theory of Heat Equations for Sigma Functions
  \footnote{Version : \today}}
\author[1]{J. Chris Eilbeck} 
\affil[1]{\footnotesize Department of
Mathematics and the Maxwell
Institute for Mathematical Sciences\\
Heriot-Watt University, Edinburgh, UK\\
{\tt J.C.Eilbeck@hw.ac.uk}} 
\author[2]{John Gibbons}
\affil[2]{\footnotesize Department of Mathematics, Imperial College,
London, UK\\
{\tt j.gibbons@imperial.ac.uk}} 
\author[3]{Yoshihiro \^Onishi}
\affil[3]{\footnotesize Department of Mathematics, 
Faculty of Science and Technology\\
Meijo University, Tenpaku, Nagoya, Japan\\
{\tt yonishi@meijo-u.ac.jp}}
\author[4]{\\ Seidai Yasuda}
\affil[4]{\footnotesize 
Department of Mathematics, 
Hokkaido University, Sapporo, Japan\\
{\tt sese@math.sci.hokudai.ac.jp}} 
\begin{document}
\raggedbottom
\allowdisplaybreaks
\pagestyle{plain}
\maketitle
\begin{abstract}
  Let \,\(e\) \,and \,\(q\) \,be fixed co-prime integers satisfying
  \(1<e<q\).  Let \(\mathscr{C}\) be a certain family of deformations of
  the curve \(y^e=x^q\).  That family is called the \((e,q)\)-curve
  and is one of the types of curves called plane telescopic curves.
  Let \(\varDelta\) be the discriminant of \(\mathscr{C}\).  Following
  pioneering work by Buchstaber and Leykin (BL), we determine the
  canonical basis \(\{ L_j \}\) of the space of derivations tangent to
  the variety \(\varDelta=0\) and describe their specific properties.
  Such a set \(\{ L_j \}\) gives rise to a system of linear partial
  differential equations (heat equations) satisfied by the function
  \(\sigma(u)\) associated with \(\mathscr{C}\), and eventually gives
  its explicit power series expansion.  This is a natural
  generalisation of Weierstrass' result on his sigma function.  We
  attempt to give an accessible description of various aspects of the
  BL-theory.  Especially, the text contains detailed proofs for
  several useful formulae and known facts since we know of no works
  which include their proofs.
\end{abstract}
\newpage
\noindent
\textbf{\Large Introduction}%
\\[7pt]
Classically, the Weierstrass function \(\sigma(u)\) 
is defined through the Weierstrass elliptic function \(\wp(u)\) as follows:
\begin{equation*}
  \sigma(u)=u\exp\bigg(\int_0^u\int_0^u\Big(\tfrac1{u^2}
  -\wp(u)\Big)du\,du\bigg). 
\end{equation*}
The modern approach is to define the sigma function starting from a
general elliptic curve.  
However, in this introduction, we treat only the curve defined by
\begin{equation}\label{curve_of_genus1}
y^2=x^3+\mu_4x+\mu_6 \ \ \ \mbox{(Weierstrass form)}. 
\end{equation}
(We refer the reader to \cite{eo_2019} for the case of the most
general elliptic curve.)
For this curve, we define the function \(\sigma(u)\) by 
\begin{equation}\label{weierstrass_sigma}
\sigma(u)=\Big(\frac{2\pi}{\omega'}\Big)^{1/2}\varDelta^{-\frac18}\,
\exp\big(-\tfrac12{\omega'}^{-1}\eta'u^2\big)
\cdot
\vartheta\bigg[\,
  \begin{matrix}
  \frac12 \\[0pt] \frac12
  \end{matrix}\,
\bigg]\!({\omega'}^{-1}u,\omega''/\omega'),
\end{equation}
where \(\varDelta=-16(4{\mu_4}^3+27{\mu_6}^2)\)
is the discriminant of the curve, 
and \(\omega'\), \(\omega''\), \(\eta'\), and \(\eta''\)
are the periods of the two differential forms 
\begin{equation*}
\dfrac{dx}{2y}, \dfrac{xdx}{2y}
\end{equation*}
with respect to a pair of fixed standard closed paths 
\(\alpha_1\) and \(\beta_1\) 
which represents a symplectic basis of the first homology group, 
though \(\eta''\) does not appear explicitly. 
The last part of (\ref{weierstrass_sigma}) is 
Jacobi's theta series defined by
\begin{equation}\label{jacobi_theta}
\vartheta\Big[\,
  \begin{matrix}
  b \\[-3pt] a  
  \end{matrix}\,
\Big]\!(z,\tau)=\sum_{n\in\mathbb{Z}}
\exp2\pi\iu\big(\tfrac12\tau(n+b)^2+(n+b)(z+a)\big)
\ \ \ \ (a,\ b\in\tfrac12\mathbb{Z}).
\end{equation}
From now on, we suppose the function \(\sigma(u)\)
is defined by (\ref{weierstrass_sigma}).  
In using this definition, it is not clear that \(\sigma(u)\)
is independent of the choice of \(\alpha_1\) and \(\beta_1\).
Indeed, both of the later part (Jacobi's theta series) of (\ref{weierstrass_sigma}) 
and former part are not invariant when we choose another pair of \(\alpha_1\) and \(\beta_1\).
However these changes offset each other and \(\sigma(u)\) itself is invariant.
\par 
Using the Dedekind eta function \(\eta(\tau)\) (not to be confused with the periods \(\eta\)
above and in Section \ref{the_curve}), the discriminant \(\varDelta\)
of the curve above is given by
\(\varDelta=\big(\frac{2\pi}{\omega'}\big)^{12}\eta(\omega''/\omega')^{24}\),
and the first terms in (\ref{weierstrass_sigma}) can be explicitly
written as
\begin{equation}\label{pre_factor}
\Big(\frac{2\pi}{\omega'}\Big)^{1/2}\varDelta^{-\frac18}
= -\frac{\ \omega'\,}{2\pi}\eta(\omega''/\omega')^{-3}.
\end{equation}
Although \(\varDelta\)
is invariant with respect to a change of \(\alpha_1\)
and \(\beta_1\), 
both sides of (\ref{pre_factor}) and \(\omega'\) 
are not invariant.  
The function \(\sigma(u)\) has a power series expansion 
at the origin as follows:
\begin{equation}\label{23sigma_expansion}
  \begin{aligned}
    \sigma(u)&=
{\smash
    u\sum_{\substack{n_4,n_6\geqq0\\}}b(n_4,n_6)
    \frac{(\mu_4u^4)^{n_4}(\mu_6u^6)^{n_6}}{(1+4n_4+6n_6)!}
}
\\[-3pt]
    &\ \ \ =u 
    +2\mu_4\frac{u^5}{5!}
    +24\mu_6\frac{u^7}{7!}
    -36\mu_4^2\frac{u^9}{9!}
    -288\mu_4\mu_6\frac{u^{11}}{11!}
    +\cdots,
  \end{aligned}
\end{equation}
where \(b(n_4,n_6)\in\mathbb{Z}\).  
(This expansion also shows the independence of \(\sigma(u)\) 
with respect to the choice of \(\alpha_1\) and \(\beta_1\)).  
\par
In this work, 
we are interested in the \textit{sigma function}, 
which is a multivariate entire function, associated to a curve 
which we call \((n,s)\)-\textit{curves} or 
{\it plane telescopic curves} (see Section \ref{the_curve} for definitions). 
One of the motivations of a theory of heat equations 
is to get a recurrence relation for the expansion coefficients of the sigma function, 
like that of the \(b(n_4,n_6)\)s. 
But, there is another motivation as follows. 
For such a general non-singular curve, 
there is an intrinsic or axiomatic definition 
(see \ref{char_sigma}) 
of the sigma function. 
It would be useful if we have an expression, like (\ref{weierstrass_sigma}), 
for the generalised sigma function. 
Indeed it is not so difficult to show that 
the natural generalisation of the right hand side of (\ref{weierstrass_sigma}),
but dropping the factor corresponding to (\ref{pre_factor}), satisfies 
some of the conditions in \ref{char_sigma}.  
Before \cite{bl_2008}, except for curves of genus one and two, 
the validity of the expression including the natural generalisation of 
the factor (\ref{pre_factor}) was not yet proved completely.  
We shall discuss this motivation again at the end of this introduction.
\par
It is well known that the sigma function for a general non-singular
algebraic curve is expressed by Riemann's theta series with a
characteristic coming from the Riemann constant of the curve
multiplied by some exponential factor and some constant factor.  This
constant factor might be a natural generalisation of
(\ref{pre_factor}).  But, there seems to be no proof of the
determination of this constant factor except in genus one and two (we
mention this again later).  To fix the last constant is another
motivation of the theory, which is described around Lemma 4.17 of
\cite{bl_2005}.
\par
We now review the classical theory of the heat equations for
\(\sigma(u)\).  Let \(z\) and \(\tau\) be complex numbers with the
imaginary part of \(\tau\) positive.  We define
\(L=4\pi\iu\frac{\partial}{\partial\tau}\) \ and \
\(H=\frac{\partial^2}{\partial{z}^2}\).  Then Jacobi's theta function
(\ref{jacobi_theta}) satisfies the following equation, which is known
as the heat equation,
\begin{equation}\label{heq_jacobi_theta}
(L-H)\,\vartheta\Big[\,\begin{matrix}b \\[-3pt] a \end{matrix}\,\Big]\!(z,\tau)=0. 
\end{equation}
Weierstrass' result in \cite{weierstrass_1882}, which is displayed as
(\ref{intro_heat23}) below, is regarded as an interpretation of
(\ref{heq_jacobi_theta}) in the form attached to his function
\(\sigma(u)\).  Strictly speaking, he did not derive it directly,
but only by elementary and quite technical integrations from the well-known
differential equation \(\wp'(u)^2=4\wp(u)^3-g_2\wp(u)-g_3\), 
where \(g_2=-4\mu_4\) and \(g_3=-4\mu_6\), satisfied
by the \(\wp(u)=-\frac{d^2}{\,du^2\,}\log\sigma(u)\). 
He eventually obtained the recurrence relation 
\begin{equation}\label{intro_recurrence}
\begin{aligned}
b(n_4,& n_6)=\tfrac23(4n_4+6n_6-1)
(2n_4+3n_6-1)b(n_4-1,n_6)\\
&-\tfrac{8}{3}(n_6+1)b(n_4-2,n_6+1) +12(n_4+1)b(n_4+1,n_6-1)
\end{aligned}
\end{equation}
for the coefficients of the expansion
(\ref{23sigma_expansion}) of \(\sigma(u)\) (see also \S\ref{section_23-curve} and \cite{onishi_2015b}). 
\par
Frobenius and Stickelberger approached (\ref{intro_recurrence}) via a
different method in their paper \cite{fs_1882} which was published
in the same year as \cite{weierstrass_1882}. 
Using the expansion 
\begin{equation}\label{wp_expansion}
  \wp(u)=\frac1{u^2}+\frac{g_2}{20}u^2+\frac{g_3}{28}u^4 
  +\frac{{g_2}^2}{1200}u^6+\cdots, 
\end{equation}
and the corresponding expansion of the Weierstrass function
\begin{equation*}
\zeta(u)=\frac1u-\int_0^u\big(\wp(u)-\frac1{u^2}\big)du
\end{equation*}
with 
\begin{equation}
\label{g2_g3}
g_2=60\!\!\!\!\!\!\sum_{(n',n'')\neq(0,0)}\!\frac{1}{(n'\omega'+n''\omega'')^4}, \ \ 
g_3=140\!\!\!\!\!\!\sum_{(n',n'')\neq(0,0)}\!\frac{1}{(n'\omega'+n''\omega'')^6},
\end{equation}
they obtained the formulae
  \begin{align*}
    \omega'\frac{\partial g_2}{\partial \omega'}
    +\omega''\frac{\partial g_2}{\partial \omega''}&=-4g_2,\ \ \ \
    \omega'\frac{\partial g_3}{\partial \omega'}
    +\omega''\frac{\partial g_3}{\partial \omega''}=-6g_3,\\
    \eta'\frac{\partial g_2}{\partial \omega'} +\eta''\frac{\partial
      g_2}{\partial \omega''}&=-6g_3, \ \ \ \ \eta'\frac{\partial
      g_3}{\partial \omega'} +\eta''\frac{\partial g_3}{\partial
      \omega''}=-\frac13{g_2}^2,
  \end{align*}
where \(\eta'=\zeta(u+\omega')-\zeta(u)\), \(\eta''=\zeta(u+\omega'')
 -\zeta(u)\) which are independent of \(u\), 
and  (see (\ref{23L_operators}))
\begin{equation}\label{L_for_genus1}
  \begin{aligned}
    \omega'\frac{\partial}{\partial \omega'}
    +\omega''\frac{\partial}{\partial \omega''}
    &=-4g_2\frac{\partial}{\partial g_2}
    -6g_3\frac{\partial}{\partial g_3}, \\
    \eta'\frac{\partial}{\partial \omega'}
    +\eta''\frac{\partial}{\partial \omega''}
    &=-6g_3\frac{\partial}{\partial g_2}
    -\frac13{g_2}^2\frac{\partial}{\partial g_3}.
  \end{aligned}
\end{equation}
In our notation of the present paper, these operators (\ref{L_for_genus1}) 
are denoted by \(4L_0\) and \(4L_2\), respectively. 
Moreover, they found (p.318 in \cite{fs_1882}) 
\begin{equation*}
  \omega'\frac{\partial\eta'}{\partial\omega'}+\omega''\frac{\partial\eta'}
  {\partial\omega''}=\eta', \ \ \ 
  \omega'\frac{\partial\eta''}{\partial\omega'}+\omega''\frac{\partial\eta''}
  {\partial\omega''}=\eta''.
\end{equation*}
and gave, on p.326 of \cite{fs_1882}, 
exactly the same system of heat equations in \cite{weierstrass_1882}. 
Observing the work of Weierstrass from the viewpoint of the paper \cite{bl_2008}, 
the left hand sides of (\ref{L_for_genus1}) correspond to
\(L=4\pi\iu\frac{\partial}{\partial\tau}\), namely the operations with
respect to the period integrals \(\tau\) or
\(\{\omega',\ \omega'',\ \eta',\ \eta''\}\) of the curve, 
which adopt to the expression (\ref{weierstrass_sigma}); 
while the right hand sides of (\ref{L_for_genus1}) are 
an interpretation of such operations in order to adopt to 
the expansion (\ref{23sigma_expansion}) of (\ref{weierstrass_sigma}) 
given by \cite{weierstrass_1882}.  
Although we suspect there are fruitful correspondences between Weierstrass,
Frobenius, and Stickelberger, the authors have no details of these.

It appears difficult to generalise the Weierstrass method to the
higher genus cases.  There is some hint in the work of
Frobenius-Stickelberger to generalise the result to these cases.  In
order to do so, it seems necessary to have generalisation of relations
(\ref{L_for_genus1}).  But we do not have naive generalisations of
(\ref{wp_expansion}) and (\ref{g2_g3}).

Recently Buchstaber and Leykin were able to generalise the above
results to the sigma functions of higher genus curves (\cite{bl_2008},
see also \cite{bl_2002,bl_2004,bl_2005}).  In \cite{bl_2008},
Buchstaber and Leykin generalise (\ref{L_for_genus1}) to higher genus
curves by using the first de Rham cohomology \(H_{\mathrm{dR}}^1\) of
the curve over the base ring, that is the space of the differential
forms of the second kind modulo the exact forms (see Section
\ref{Frob_Stick}).  The paper \cite{bl_2008} is our main reference for
our work.  Understanding that paper requires some background on the
basic theory of heat equations and singularity theory, so we will
summarise their arguments and those of Frobenius and Stickelberger, with
hopefully accessible explanations.

We shall explain their method by taking the curve (\ref{curve_of_genus1})
as an example.  Firstly, we introduce a certain heat equation (the primary
heat equation) satisfied by the function defined in (\ref{G_b})
below, which is a generalisation of the individual terms of the series
appearing in the definition (\ref{weierstrass_sigma}) of \(\sigma(u)\).
Let us take the subalgebra \(\vec{L}\) generated by 
\(L_0\) and \(L_2\) over \(\mathbb{Q}[\mu_4,\,\mu_6]\) in the Lie algebra 
generated by \(\frac{\partial}{\partial \mu_4}\) and \(\frac{\partial}{\partial \mu_6}\). 
Thanks to a lemma due to Chevalley (Lemma \ref{chevalley}) and 
the \textit{horizontal derivation formula} given 
in \cite{oss2024} which is explained in Section \ref{horizontal_formula}, 
we see that the Lie algebra generated by \(L_0\) and \(L_2\) 
over \(\mathbb{Q}[\mu_4,\mu_6]\) acts on \(H_{\mathrm{dR}}^1\).
Take a symplectic basis \((\frac{dx}{2y},\ \frac{xdx}{2y})\) of \(H_{\mathrm{dR}}^1\) with
respect to a naturally defined inner product in \(H_{\mathrm{dR}}^1\) (See (\ref{symplectic})).  
Let
\begin{equation*}
  \Gamma^L = \bigg[\,\begin{array}{cc}
 -\beta  &\ \alpha \\[-2pt]
 -\gamma &\ \beta
\end{array}\,\bigg]
\end{equation*}
be the representation matrix of the action of 
an operator \(L\in\vec{L}\) 
(see (\ref{gauss_manin_connection})), 
which is called a {\it Gauss-Manin connection} in \cite{bl_2008}.  
Then we see that \(\alpha\), \(\beta\), and \(\gamma\) belong to \(\mathbb{Q}[\mu_4,\mu_6]\).
Taking integrals along the set of closed paths \(\alpha_1\) and \(\beta_1\)
which make a symplectic homology basis of the curve, 
we see that the action of \(L\) gives a linear transformation of 
the period matrix (see (\ref{per_matrix})) 
\begin{equation*}
  \varOmega=\bigg[\,
    \begin{array}{cc}
      \omega' &\ \omega'' \\[-2pt]
      \eta'   &\ \eta''
\end{array}\,\bigg]
\end{equation*}
with respect to the basis \((\frac{dx}{2y},\ \frac{xdx}{2y})\) and 
paths \(\alpha_1\), \(\beta_1\). 
This transformation is also represented by \(\Gamma^L\) as \(L(\varOmega)=\Gamma^L\,\varOmega\) 
(see (\ref{bridge_formula})). 
Moreover, we introduce another operator
\begin{equation*}
H^L=\tfrac12\,[\,\tfrac{\partial}{\partial u}\ \ u\,]
    \bigg[
      \begin{array}{rc}
        \alpha  & \ \beta \\
        \beta   & \ \gamma
      \end{array}
      \bigg]
    \bigg[
      \begin{array}{c}
        \tfrac{\partial}{\partial u} \\ u
      \end{array}
     \bigg]
  =\frac12\Big(\,{\alpha}\,\frac{\partial^2}{\partial u^2}
  +2{\beta}\,u\frac{\partial}{\partial u}
  +{\gamma}\,u^2
  +\beta\,\Big). 
\end{equation*}
Then the function 
\begin{equation}\label{G_b}
\begin{aligned}
G(b,u,\varOmega)
&=\bigg(\frac{2\pi}{\omega'}\bigg)^{\frac12}\,
\exp\big(-\tfrac12 \eta' \omega'^{-1} u^2\big)\\
&\times
  \exp\Big(\,2\pi i\big(\,\tfrac12 \omega'^{-1} \omega''{b''}^2
+ b''(\omega'^{-1}u+b')\,\big)\Big),
\end{aligned}
\end{equation}
where \,\(b=[\,b'\ \ b''\,]\) is an arbitrary constant vector,   
satisfies the heat equation 
\begin{equation}\label{heat_eq_23}
(L-H^L)G(b,u,\varOmega)=0. 
\end{equation}
We call this (and its generalisation) the {\it primary heat equation} (Theorem \ref{primary_heq}).  
While checking the validity of (\ref{heat_eq_23}) is rather complicated, 
no details are given by Buchstaber and Leykin, 
and the description of this equation in \cite{bl_2008} is 
not entirely consistent. 
We denote the expression of the right hand side
(\ref{weierstrass_sigma}) without \(\varDelta^{-\frac18}\)
by \(\tilde{\sigma}(u)\) (see (\ref{pre_classical_sigma})).  
According to the above equation and the fact that \(\tilde{\sigma}(u)\)
is an infinite sum of the \(G(b,u,\varOmega)\)s for various \(b'\)
and \(b''\), we see that \((L-H^L)\,\tilde{\sigma}(u)=0\).
\par
At the next stage, we shall check the operators \(L_0\) and \(L_2\) 
give rise to a system of heat equations which are satisfied 
by the right hand side of (\ref{weierstrass_sigma}), 
including the factor \(\varDelta^{-\frac18}\). 
Indeed, these operators are tangent to 
the singular locus given by \(\varDelta=0\) 
(see the former part of Subsection \ref{operators_and_discriminant}). 
The paper \cite{bl_2008}
uses knowledge of singularity theory and succeeds in
generalising nicely the result of Weierstrass and
Frobenius-Stickelberger.
So we explain techniques from singularity theory 
to calculate \,\(\varDelta\) (Lemma \ref{jet}(1)) 
as well as the operators tangent to the variety defined by \(\varDelta=0\).
This stage is carried out in Subsections \ref{alg_heat_eq} and
\ref{operators_and_discriminant} and the result is given in
(\ref{ell_operators}).  
\par
In summary, the system of heat equations for (\ref{curve_of_genus1}) 
obtained by Weierstrass is 
{\small
\begin{equation}\label{intro_heat23}
\begin{aligned}
(L_0-H^{L_0})\,\sigma(u)
&=\left(4{\mu_4}\frac{\partial}{\partial\mu_4}
  +6{\mu_6}\frac{\partial}{\partial\mu_6}
  -u\frac{\partial}{\partial u}+1\,\right)\sigma(u)=0,\\
(L_2-H^{L_2})\,\sigma(u)
&=\left(6{\mu_6}\frac{\partial}{\partial\mu_4}
    -\frac43{\mu_4}^2\frac{\partial}{\partial\mu_6}
    -\frac12\frac{\partial^2}{\partial u^2}
    +\frac16{\mu_4}u^2\,\right)\sigma(u)=0,
\end{aligned}
\end{equation}
}
which is reproved as (\ref{heat23}).  
For the general curves, the corresponding results are given as 
Theorem \ref{heq_for_hat_sigma} in the text.

It is very important to determine whether the  system of 
heat equations we obtain characterises the sigma function.  
For the genus one case, it was seen by Weierstrass 
that the recurrence (\ref{intro_recurrence}) determines 
all the coefficients if we give an arbitrary value for \(b(0,0)\).
That is, the solution space of (\ref{intro_recurrence}), as well as
(\ref{intro_heat23}), is of dimension one.

For a general non-singular curve, we consider the multivariate
function \(\sigma(u)\) defined similarly to (\ref{weierstrass_sigma}).
We can check that the operators we obtain (of Theorem
\ref{heq_for_hat_sigma}) kill \(\sigma(u)\).  
However, {\it it is not clear whether the solution space is of dimension one}\, 
over the base field.  
The authors could not find any reason which suggests that 
the solution space is one dimensional.  
Nevertheless, we shall show that, 
for any curve of genus less than or equal to three, 
the solution space is one dimensional by giving 
an explicit recurrence relation from the system of heat equations we obtain, 
which is described in Section \ref{section3}.
\par
Although our main results are in Subsections \ref{section_27-curve}, \ref{section_27-sigma},
\ref{section_34-curve} and \ref{section_34-sigma}, we give a number of additional useful results, 
which may be known only by specialists, 
with detailed proofs in Section \ref{section2} and Subsection
\ref{on_2q_curves}.  Subsection \ref{section_23-curve} reproduces the
classical result and it would be helpful to read the following
Subsections.  Subsection \ref{section_25-curve} is rewritten in a
slightly different formulation (Hurwitz-type series expansion of
\(\sigma(u)\)) from \cite{bl_2005}.
\par
We shall explain here a notion called {\it modality} 
which was introduced by Arnol'd 
(see \ref{weierstrass_form_modality} for details).  
For any coprime positive integers \((e,q)\) with \(e<q\), 
we consider a family of the curves 
found by certain deformations of the singularity 
at the origin of the curve \(y^e=x^q\) 
(we do not use these words in the text of this paper). 
This family of curves is called the \((e,q)\)-\textit{curve}, 
which is a type of \textit{plane telescopic curve}. 
The number of parameters necessary for this deformation is 
less than or equal to \((e-1)(q-1)\). 
The last number is twice the genus of a generic curve of the family.  
The difference between \((e-1)(q-1)\) and the number of parameters 
is called the {\it modality} of this family.  
For instance, the curve (\ref{curve_of_genus1}) is 
regarded as the whole of the semi-universal deformations 
of \(y^2=x^3\) with two parameters \(\mu_4\) and \(\mu_6\), 
which is equal to twice its genus (i.e. \(2=1\times2\)).  
In this case the modality is \(0\). 
It is known that the hyperelliptic curve given
by such deformations in the case \(e=2\) is of modality \(0\).  
There are only two types of non-hyperelliptic plane telescopic curves 
of modality \(0\), which are the trigonal quartic curve, 
the \((3,4)\)-curve (genus three), and the trigonal quintic curve, 
the \((3,5)\)-curve (genus four).  
Concerning these two curves, we treat only the former one, 
the \((3,4)\)-curve, in this paper.
More general curves including the \((3,5)\)-curve 
are discussed in \cite{oss2024}.  
For a general hyperelliptic curve, 
we give its corresponding system of heat equations in Lemma
\ref{hyp_T-matrix}.  
For any plane telescopic curve, we gave a simple
formula for its modality in Proposition \ref{modality0}.
\par
In the last paragraph of Section 2 in \cite{bl_2008}, 
there is some description of the positive modality case. 
On the one dimesionality problem of the solution space of the case \(e=2\) 
and some positive modality cases, 
we refer the reader to the forthcoming paper \cite{oss2024}. 
\par
We shall mention two additional consequences of this theory.  
Firstly, for hyperelliptic curves of genus less than or equal to three, 
we again prove partially the result of \cite{onishi_2018} 
on Hurwitz integrality of the expansion of the sigma function.  
For example it is obvious from (\ref{intro_recurrence}) that
\(b(n_4,n_6)\in\mathbb{Z}[\frac13]\).  
Similar results are shown for the hyperelliptic curves of genus two and three.  
This idea was suggested to Y.\^O. by Buchstaber.  
Secondly, this theory of heat equations in turn helps 
the construction of the sigma function, 
as explained in Lemma 4.17 of \cite{bl_2005} 
(see also Section \ref{sigma_is_hat_sigma}).

However, it might be possible to approach this problem 
via the results of Bernatska \cite{Bern22,Bern20d} 
to get the top factor for a hyperelliptic curve 
corresponding to that of (0.2).
   
The formula (\ref{weierstrass_sigma}) is
well-known for the curve (\ref{curve_of_genus1}), and its
generalisation (see (\ref{classical_sigma})) is proved for 
the genus two hyperelliptic curve by Grant \cite{grant} 
by using Thomae's formula.  
For any curve in the family we have investigated, 
there is a rough explanation in Lemma 2.3 in p.98 of \cite{bl_2004}, 
but without using Thomae's formula.  
\par
Here, we mention the discriminant for 
a plain telescopic curve \(\mathscr{C}\). 
There is nice algorithm using a certain determinant 
to compute the discriminant of \(\mathscr{C}\) 
as explained in \ref{jet} and \ref{lemma_discri}. 
However, there is still a gap in this theory; that is, 
we do not have a general proof of coincidence of the discriminant 
and the determinant. 
If we know the weight of the discriminant in general, 
our idea of the proof of \ref{lemma_discri} works well. 
\par
We explain here that BL-theory indeed gives a method to prove such a
formula as (\ref{weierstrass_sigma}) on the sigma functions, at least,
for our curves of genus less than or equal to three.
Firstly, assuming the expression (\ref{weierstrass_sigma}) 
to be the correct sigma function, 
we show that it satisfies the system of heat equations.  
On the other hand, as we mentioned above, 
the solution space of the system is one dimensional and 
the system of heat equations gives a recursion relation, 
by which we have the power series expansion of 
the solution as shown in Section 3.  
Especially,  we see the solution space is of dimension one.  
Therefore, we have a proof that the assumed expression of 
the sigma indeed gives the sigma function 
up to a non-zero absolute constant.
This result is the main theorem of the present paper
(see Theorem \ref{Main_Thm}).  
\par
Finally, one of the authors S.Y.\ wishes to point out 
to the reader that his contribution on this paper is limited to 
the proof of the case \(e=2\) in 
Proposition \ref{hessian_formula_}. 
\par
{\it Acknowledgments}\,: We are grateful to Christophe Ritzenthaler
who explained Sylvester's algorithm, which is very useful to compute
the discriminants for curves of genus less than or equal to three.  We
would also like to thank Julia Bernatska who explained some of the
details of the practical implementation of the algorithm behind the
genus two example in \cite{bl_2008}.  We would like to thank Toshizumi
Fukui for bringing Zakalyukin's paper to our attention, Masataka
Shibata for giving us crucial comments on Section \ref{section3} from
the viewpoint of \cite{oss2024}, and Kouki Sato for showing us the
horizontal derivation formula in Section \ref{horizontal_formula}.  We
are grateful to an anonymous referee for many helpful suggestions and
corrections to improve the paper.  This research is supported by JSPS
grant 25400010, 16K05082, and 23K03157.  \vskip 20pt
\noindent
\textbf{Convention}. \
We use the following convention. \\
As usual, we denote by \(\mathbb{Z}\), \(\mathbb{Q}\), \(\mathbb{R}\),
and \(\mathbb{C}\) the ring of integers, the field of rationals, the
field of real numbers, the field of complex numbers, respectively.  We
denote by \(\iu\) the imaginary unit.  \(\mathrm{Mat}(n,R)\)\, denotes
the ring of square matrices of size \(n\) with all the entries in a
ring \(R\).  \(\mathrm{Sym}(n,R)\)\, denotes the set of symmetric
matrix in \(\mathrm{Mat}(n,R)\).  \(\tp{\!A}\)\, stands for the
transpose of a matrix \(A\).  \(\rev{\!A}\)\, denotes the matrix
obtained by reversing orders both of the rows and the columns of a
matrix \(A\).
\newpage
\tableofcontents
\newpage
\section{Preliminaries}\label{section1}
\subsection{The curves}
\label{the_curve}
We shall use \(e\) and \(q\) instead of \(n\) and \(s\), respectively,
of the \((n,s)\)-curves, which was usual in many previous papers on
generalised sigma functions.  Especially, the name \((n,s)\)-curve
(which comes from singularity theory) is used by Buchstaber and Leykin
in their papers, but we wish to avoid confusion with the many \(n\)
used as subscripts in sections from Section \ref{section3} onward, and
the use of \(s\) for Schur polynomials in Subsection
\ref{materials_for_sigma}.
\par
So, we let \(e\) and \(q\) be  two fixed  positive integers 
such that \(e<q\) and \(\gcd(e,q)=1\).
We define, for these integers, a polynomial of indeterminates \(X\) and \(Y\)
\begin{equation}\label{def_eq}
  f(X,Y)=Y^e-p_1(X)Y^{e-1}-\cdots-p_{e-1}(X)Y-p_e(X),
\end{equation}
where  \(p_j(X)\) is a polynomial of \(X\) of degree 
\(\lceil\tfrac{jq}{e}\rceil\) or smaller
and its coefficients, which are also indeterminates, are denoted by
\begin{equation}\label{5.02}
  \begin{aligned}
  p_j(X)&=\sum_{k:jq-ek>0}\mu_{jq-ek}\,X^k \ \ \ (1\leq j\leq e-1),\\
  p_e(X)&=X^q+\mu_{e(q-1)}X^{q-1}+\cdots+\mu_{eq}. 
  \end{aligned}
\end{equation}
Please note that the sign at the front of each \(p_j(X)\) 
with \(j\neq e\) in \(f(X,Y)\) is different from previous papers 
written by some of the authors. 
The base ring over which we work is quite general.  
For simplicity the reader may start by taking 
the field \(\mathbb{C}\) of complex numbers and 
assume the \(\mu_i\)s to be constants belonging to this field.  
Let \(\mathscr{C}=\mathscr{C}_{{\mu}}^{e,q}\) be the
projective curve defined by
\begin{equation}\label{plane_miura}
  f(x,y)=0
\end{equation}
having a unique point \(\infty\) at infinity. 
This means \((x,y)\) is a generic point of \(\mathscr{C}\) 
in the classical terminology. 
\par
As the general elliptic curve is defined by an equation of the form
\begin{equation*}
  y^2-(\mu_1x+\mu_3)y=x^3+\mu_2x^2+\mu_4x+\mu_6,
\end{equation*}
the curves \(\mathscr{C}\) discussed here are a natural generalisation of
elliptic curves.%
\par 
The reason why we omit the terms of \(\mu_j\) with \(j<0\) 
from \(f(X,Y)\) is seen in the proof of {\rm\ref{basis_mod_f1_f2}}. 
\par
Our principal situation is that all the coefficients \(\mu_j\) of
\(f(X,Y)\) should be indeterminates.  
We denote by \(\mathbb{Q}[{\mu}]\) the ring generated 
over the rationals \(\mathbb{Q}\) by all the \(\mu_j\)s.  
Then we shall treat \(\mathscr{C}\) as a scheme over 
the \(\mathrm{Spec}\,\mathbb{Q}[\mu]\).  
Since we need to use analytic methods from time to time, 
we freely switch the standing position where \(\mu_j\)s are 
assumed to be complex numbers or indeterminates. 
\par
This \(\mathscr{C}\) should be called an {\it \((e,q)\)-curve}
\,following Buchstaber, Enolskii, and Leykin \cite{bel_1997}, or a
{\it plane telescopic curve} after the paper \cite{miura_1998}.
Assuming all the \(\mu_j\)s are complex numbers, 
the genus of \,\(\mathscr{C}\)\, is \((e-1)(q-1)/2\) provided that 
it is non-singular. 
We will use \(g\) to denote this quantity throughout this paper 
whether the curve \(\mathscr{C}\) is non-singular or singular as well as 
in the case of the \(\mu_j\)s being indeterminates: 
\(g=(e-1)(q-1)/2\). 
In this paper, 
we denote by 
\begin{equation*}
f_X(x,y)\ \ \text{or}\ \ f_1(x,y) \ \ \ \
\text{[\,resp.}\ \  
f_{\,Y}(x,y)\ \ \text{or}\ \ f_2(x,y)]
\end{equation*}
the polynomials obtained by substitution \(X=x\), \(Y=y\) 
for the partial derivative of the polynomial \(f(X,Y)\) 
with respect to \(X\) [\,resp. \(Y\)]
\par
Now we introduce a weight function as follows. 
For a point \((x,y)\) in the curve  \(\mathscr{C}\)  given by (\ref{def_eq}), 
we define the weight \(\mathrm{wt}()\) on \(\mathbb{Q}[\mu][x,y]\) 
and \(\mathbb{Q}[\mu][X,Y]\) by 
\begin{equation}\label{wt_on_C}
\wt(\mu_j)=-j, \ \ \ 
\wt(x)=\wt(X)=-e, \ \ \ 
\wt(y)=\wt(Y)=-q. 
\end{equation}
Then all the equations for functions, power series, differential
forms, and so on in this paper are of homogeneous weight.  
We see that \(\wt\big(f(X,Y)\big)=-eq\). 
We will extend the notion of weight (\ref{wt_on_C}) in Subsection \ref{ext_wt}. 
\subsection{Definition of the discriminant}
We shall define the discriminant of the curve \(\mathscr{C}\). 
\vskip 5pt
\begin{definition}\label{def_discriminant}
Suppose all the \(\mu_j\)s are indeterminates. 
The discriminant \,\(\varDelta\)\, of the form \(f(X,Y)\) or 
of the curve \(\mathscr{C}\) defined by \,\(f(x,y)=0\)\, is 
the polynomial \,{\rm (}up to the signs \(\pm\){\rm )} of
the least degree in the \(\mu_j\)s with integer coefficients such
that the greatest common divisor of the coefficients is \(1\), 
and every zero of \(\varDelta\) corresponds exactly to 
the case that \(\mathscr{C}\) has a singular point.
\end{definition}
\vskip 3pt
For a curve \(\mathscr{C}_0\) given by some fixed constants \(\mu_j\in\mathbb{C}\),  
we always define its discriminant as the one obtained by substituting
these constants to the discriminant \(\varDelta\) 
defined in \ref{def_discriminant}. 
For instance, if \((e,q)=(2,3)\), then 
\begin{equation*}
f(X,Y)=Y^2-(\mu_1X+\mu_3)Y-(X^3+\mu_2X^2+\mu_4X+\mu_6)
\end{equation*}
and its discriminant is given by 
\begin{equation}\label{def_disc_23-curve}
\begin{aligned}
\varDelta
&=-\mu_6{\mu_1}^6+\mu_3\mu_4{\mu_1}^5+((-{\mu_3}^2-12\mu_6)\mu_2+{\mu_4}^2){\mu_1}^4\\
&\ \ +(8\mu_3\mu_4\mu_2+{\mu_3}^3+36\mu_6\mu_3){\mu_1}^3+((-8{\mu_3}^2-48\mu_6){\mu_2}^2+8{\mu_4}^2\mu_2\\
&\ \ +(-30{\mu_3}^2+72\mu_6)\mu_4){\mu_1}^2+(16\mu_3\mu_4{\mu_2}^2+(36{\mu_3}^3+144\mu_6\mu_3)\mu_2-96\mu_3{\mu_4}^2)\mu_1\\
&\ \ +(-16{\mu_3}^2-64\mu_6){\mu_2}^3+16{\mu_4}^2{\mu_2}^2+(72{\mu_3}^2+288\mu_6)\mu_4\mu_2\\
&\ \ -64{\mu_4}^3-27{\mu_3}^4-216\mu_6{\mu_3}^2-432{\mu_6}^2.
\end{aligned}
\end{equation}
Specializing \(\mu_1=\mu_3=\mu_2=0\), 
we have the discriminant \(\varDelta=-16(4{\mu_4}^3+27{\mu_6}^2)\) 
of the curve defined by the Weierstrass form 
\,\(y^2=x^3+\mu_4x+\mu_6\)\,
which appeared in (\ref{curve_of_genus1}).  
\par
For \((e,q)=(2,2g+1)\), as in Section \ref{weierstrass_form_modality},
we rewrite the equation as \(y^2=x^{2g+1}+\cdots\), where the right
hand side is a polynomial of \(x\) only.  Then the discriminant of
this curve is a non-zero integer multiple of 
the discriminant of the right hand side as a polynomial of \(x\) only.
For the \((3,4)\)-curve, we have Sylvester's method as described in
\cite{GKZ} pp.118-120, as explained to the authors by C. Ritzenthaler.  
However, Lemma \ref{jet} below gives a quite general method, 
which seems to cover the \((3,5)\)-curve and more. 
We do have explicit forms of the discriminants of the curves
with \((e,q)=(2,3)\), \((2,5)\), \((2,7)\), \((3,4)\) which we treat
in this paper.  Using the resultant of two forms, we mention here an
alternative (but conjectural) construction for the interest of the
reader, though it is essentially not used in this paper.
\begin{definition}\label{def_R}
Let the coefficients \(\mu_j\) of {\rm(\ref{def_eq})} be indeterminates,
and define
\begin{equation*}
\begin{aligned}
R_1&={\rslt}_X
  \Big(\rslt_Y\big(f(X,Y), f_1(X,Y)\big), 
       \rslt_Y\big(f(X,Y), f_2(X,Y)\big)
  \Big), \\
R_2&=\rslt_Y
\Big(\rslt_X\big(f(X,Y), f_1(X,Y)\big), 
     \rslt_X\big(f(X,Y), f_2(X,Y)\big)
\Big),\\
R&=\gcd(R_1,R_2)\ \ \ \ \mbox{in \,\(\mathbb{Z}[{\mu}]\)}. 
\end{aligned}
\end{equation*}
Here \(\rslt_Z\) is the Sylvester resultant with respect to \(Z\).
\end{definition}
\newpage
Now we recall the conjecture from the paper \cite{eemop_2009}. 
\begin{conjecture}\label{disc_conj} 
Defining \(R\) by {\rm\ref{def_R}}, 
we have the following\,{\rm :}\\
{\rm (1)} 
\(R\) is always a perfect square in \(\mathbb{Z}[\mu]\)
and \(R=\varDelta^2\){\rm ;} \\
{\rm (2)} 
The discriminant \(\varDelta\) of the \((e,q)\)-curve is
of weight \(-2eqg=-eq(e-1)(q-1)\).
\end{conjecture}
\begin{remark}\label{wt_discriminant}
{\rm%
It can be confirmed that (1) of \ref{disc_conj} is correct for the cases 
 \((e,q)=(2,3)\), \((2,5)\), \((2,7)\), \((3,4)\), and \((3,5)\). 
Actually, computation by \texttt{Maple} for these cases shows 
that \(R\) is a square of some \(\varDelta'\in\mathbb{Z}[{\mu}]\). 
It is easy to check by \texttt{Maple} that \(\varDelta'\) is irreducible. 
Then, from the definition of \(R\), 
we see \(\varDelta'\) must be \(\varDelta\) up to the sign, 
and we checked (2) of \ref{disc_conj} for these cases. 
We prove also that (2) of \ref{disc_conj} is true 
if \(\gcd(e-1,q-1)=1\) in \ref{lemma_discri}. 
}
\end{remark}
\subsection{The Weierstrass form of the curve and its modality}
\label{weierstrass_form_modality} 
Starting from the equation \(f(x,y)=0\) in (\ref{def_eq}) and removing
the terms of \(y^{e-1}\) and \(x^{q-1}\) by replacing \(y\) by
\( y+\frac1e\,p_1(x)\), and \(x\) by \( x+\frac1q\mu_{(q-1)e}\),
respectively, we get a new equation \(f(x,y)=0\) which is called the
{\it Weierstrass form} of the original one.  After making such
transformations, we re-label the coefficients by \(\mu_j\).
\par
For example, if \((e,q)=(2,2g+1)\), the new equation is
\begin{equation*}
f(x,y)=y^2-(x^{2g+1}+\mu_{4g-2}x^{2g-1}+\mu_{4g-4}x^{2g-2}+\cdots+\mu_{4g+2})=0;
\end{equation*}
and if \((e,q)=(3,4)\), the new one is
\begin{equation*}
f(x,y)=y^3-(\mu_2x^2+\mu_5x+\mu_8)y-(x^4+\mu_6x^2+\mu_9x+\mu_{12})=0.
\end{equation*}
In these cases, the number of remaining \(\mu_j\)s is \(2g\).
However, in general, we can have some cases such that 
this number is less than \(2g\).  
The difference
\begin{equation*}
2g-\mbox{\lq\lq the number of \(\mu_j\)''}
\end{equation*}
is called the {\it modality} (a term used in singularity theory) 
of the \((e,q)\)-curve.
We give here a simple formula giving modalities and, especially,
determine all the curves of modality \(0\).  
\begin{proposition}\label{modality0}
  The modality of an \((e,q)\)-curve is given by
  \(\frac12(e-3)(q-3)+\lfloor\frac{q}{e}\rfloor-1\).  The only curves
  of modality \(0\) are the \((2,2g+1)\)-, \((3,4)\)-, and
  \((3,5)\)-curves.
\end{proposition}
\begin{proof}
The number of \(\mu_j\)s appearing in the Weierstrass form is
\begin{equation*}
\begin{aligned}
\sum_{j=1}^{e-2}&\bigg(\bigg\lfloor\frac{(e-j)q}e\bigg\rfloor+1\bigg)
  +(q-1)
  =\frac12\sum_{j=1}^{e-1}\bigg(\bigg\lfloor\frac{(e-j)q}e\bigg\rfloor
  +\bigg\lfloor\frac{jq}e\bigg\rfloor\bigg)-\bigg\lfloor\frac{q}{e}
  \bigg\rfloor+(e-2)+(q-1)\\
  &=\frac12\sum_{j=1}^{e-1}(q-1)-\bigg\lfloor\frac{q}{e}\bigg\rfloor+e+q-3
  =\frac12(e-1)(q-1)+e+q-3-\bigg\lfloor\frac{q}{e}\bigg\rfloor.
\end{aligned}
\end{equation*}
Subtracting the above result from \(2g=(e-1)(q-1)\) gives the required
expression.  The latter part follows directly from this.  This
completes the proof.
\end{proof} %
\vskip 5pt
On the case for a curve with positive modality, 
there is some description in \cite{bl_2008} (the end of Section 2). 
Since it is not clear for us how positive modality causes difficulty, 
we do not discuss this theme here, though we give 
an example of positive modality in \ref{positive_modality_example}. 
\vskip 3pt
\par
\textsl{From now to the end of the paper, 
we always assume that the equation \,\(f(x,y)=0\)\, of the curve \(\mathscr{C}\) is given by a Weiserstrass form. }
\newpage
\subsection{Weight}\label{weight}
\relax%
Besides (\ref{def_eq}), that is 
\begin{equation*}
f(X,Y)
=Y^e-X^q-\hskip -10pt\sum_{\substack{0\leq i\leq q-2\\0\leq j\leq e-2\\ie+jq<eq}}\hskip -8pt\mu_{eq-ie-jq}\,X^iY^j,
\end{equation*}
we introduce
\begin{equation}\label{tilde_f}
\tilde{f}(X,Y)
=Y^e-X^q-\hskip -10pt\sum_{\substack{0\leq i\leq q-2\\
    0\leq j\leq e-2}}\hskip -8pt\mu_{eq-ie-jq}\,X^iY^j. 
\end{equation}
Here, we also assume that the coefficients \(\mu_{eq-ie-jq}\) are 
indeterminates or independent complex variables. 
Note that, contrary to (\ref{def_eq}), 
the polynomial (\ref{tilde_f}) may contain 
some coefficients \(\mu_{eq-ie-jq}\)  with  \(eq-ie-jq<0\).  
We extend the weight defined in (\ref{wt_on_C}) for new coefficients by 
\begin{equation*}
\wt(\mu_j)=-j.
\end{equation*}
By this definition, \(\tilde{f}(X,Y)\) is also of homogeneous weight
with respect to (\ref{wt_on_C}).  Moreover, any quantity appeared in
this paper is of homogeneous weight with respect to this weight.  For
any pair \((i,j)\) with \(0\leq i\leq q-2\), \(0\leq j\leq e-2\),
we define
\begin{equation}\label{def_M's}
\begin{aligned}
  M_{ei+qj}&=M_{ei+qj}(X,Y)=X^iY^j, \\
  \wt(M)&=\{\,-\wt(X^iY^j)\,(=-\wt\big(M_{ei+qj}(X,Y)\big)\,|\,0\leq i\leq q-2, 0
          \leq j\leq e-2\,\}\\
           &=\{\,-(ei+qj)\,|\,0\leq i\leq q-2, 0\leq j\leq e-2\,\}. 
\end{aligned}
\end{equation}
The sequence constituted by the elements in \(-\wt(M)\) in increasing order is denoted by
\begin{equation*}
v_1(=0),\ v_2(=e), v_3, \ \cdots,\ v_{2g-2}, \ v_{2g-1}(=4g-2-e), \ v_{2g}(=4g-2). 
\end{equation*}
Because of the assumption \(\gcd(e,q)=1\), 
for any \(v\in\wt(M)\) there exists a unique pair \((i,j)\) 
such that \(0\leq i\leq q-2\), \(0\leq j\leq e-2\) 
and \(M_{v}=X^iY^j\). 
We introduce the notation 
\begin{equation}\label{def_all_M}
  \begin{aligned}
       M(X,Y)&=\tp{[\,M_{v_j}(X,Y) \ \ (j=1,2,\cdots,2g)\,]}, \\
  \rev{M}(X,Y)&=\tp{[\,M_{v_j}(X,Y) \ \ (j=2g,2g-1,\cdots,1)\,]}. 
  \end{aligned}
\end{equation}
\par
We denote the Weierstrass gap sequence of 
the semigroup generated by \(e\) and \(q\) 
in the positive integers in the increasing order by 
\begin{equation}\label{gaps}
w_1(=1), \ w_2, \ \cdots, \ w_g(=2g-1),
\end{equation}
which is also the Weierstrass gap sequence of \(\mathscr{C}\) at \(\infty\). 
Namely, this is the unique (finite) increasing sequence of positive integers 
which cannot be written in the form \,\(ae+bq\)\,
with non-negative integers \(a\), \(b\). 
We denote the set of the terms in (\ref{gaps}) by
\begin{equation*}
\mathrm{wgs}(e,\,q). 
\end{equation*}
It is well-known that each term of the sequence is written in the form 
\begin{equation*}
w_j=2g-1-v_{g-j+1} \ \ \ (1\leq j\leq g),
\end{equation*}
and that the assumption \(\gcd(e,q)=1\) implies the Young tableaux
associated with the sequence
\begin{equation*}
w_{2g}-(g-1), \ w_{2g-1}-(g-2), \ \cdots, \ w_2-1, \ w_1-0
\end{equation*}
given by the Weierstrass gap sequences 
is symmetric with respect to the diagonal line from the top-left to the bottom-right. 
The terms in the  sequence \,\smash{\(\{v_j\}_{j=1}^{2g}\)}\, are written also as
\begin{equation*}
  v_j=
  \bigg\{\ 
\begin{aligned}
  &2g-1-w_{g-j+1}  \ \ \ \mbox{if \ \(1\leq j\leq g\)},\\
  &2g-1+w_{j-g} \ \ \ \ \ \mbox{if \ \(g+1\leq j\leq 2g\)}.
\end{aligned}
\end{equation*}
We see that all the terms
\(\{M_{v_j}(X,Y)\}\) appear in \(f(X,Y)\) provided the {\it modality} of
the curve is \(0\).  
We shall give below sample values of the data above 
for the convenience of the reader to follow 
the calculation in Section \ref{section3}. 
\par
Finally, we introduce the following notation on the coefficients \(\mu_j\)\,;
\begin{equation*}
\begin{aligned}
\widetilde{\mu}&=\{\,\mu_{eq-ie-jq}\,;\,0\leq i\leq q-2,\ 0\leq j\leq e-2\,\}, \\
           \mu &=\{\,\mu_{eq-ie-jq}\,;\,0\leq i\leq q-2,\ 0\leq j\leq e-2,\ ie+jq\leq eq\,\},    
\end{aligned}
\end{equation*}
and
\begin{equation*}
\begin{aligned}
  \wt(\widetilde{\mu})&=\{\,-\wt(\mu_{eq-ie-jq})\,;\,0\leq i\leq q-2,\ 0\leq j\leq e-2\,\}\\
                      &=\{\,-(eq-ie+jq)\,;\,0\leq i\leq q-2,\ 0\leq j\leq e-2\,\}, \\
  \wt(\mu)&=\{\,-\wt(\mu_{eq-ie-jq})\,;\,0\leq i\leq q-2,\ 0\leq j\leq e-2,
\ ie+jq\leq eq\,\}\\
                      &=\{\,-(eq-ie-jq)\,;\,0\leq i\leq q-2,\ 0\leq j\leq e-2,
                        \ ie+jq\leq eq\,\}. 
\end{aligned}
\end{equation*}
\vskip 8pt
\begin{example}\label{modality_examples}
{\rm
If \((e,q)=(2,2g+1)\), then the \(M_j\)s are given as follows:
\begin{equation*}
\begin{array}{|c|ccccccccc|c}
\cline{1-10}
     j          &\   1  &\ 2    &\ \cdots &\ g       &       &\ g+1  &\ g+2   &\ \cdots &\ 2g \\ \hline
w_{g-j+1}         &\ 2g-1  &\ 2g-3 &\ \cdots &\ 1      & \vrule &\ 1    &\ 3     &\ \cdots &\ 2g-1\ \ &\  w_{j-g}\\ \hline
v_j              &\   0   &\ 2    &\ \cdots &\ 2g-2   &        &\ 2g   &\ 2g+2  &\ \cdots &\ 4g-2\ \ \\ \cline{1-10} 
\,M_{v_j}(X,Y)\,\ &\   1  &\ X    &\ \cdots &\ X^{g-1} &       & X^g    & X^{g+1} & \cdots & X^{2g-1} \strutf{13pt}{5pt}\\  \cline{1-10}
\end{array}
\end{equation*}
If \((e,q)=(3,4)\) or \((3,5)\), then the \(M_j\)s are given as follows:
\begin{equation*}
\mbox{\((3,4)\)-curve} : \ \ 
\begin{aligned}
\begin{array}{|c|ccccccc|c}
\cline{1-8}
    j           &\ 1 &\ 2 &\ 3 &\        &\   4 &\  5 &\   6  \\ \hline                                              
w_{3-j+1}         &\ 5 &\ 2 &\ 1 &\ \vrule &\   1 &\  2 &\   5   &\ w_{j-3}\\ \hline                                              
  v_j            &\ 0 &\ 3 &\ 4 &\        &\   6 &\  7 &\  10  \\ \cline{1-8}                                         
\,M_{v_j}(X,Y)\,\ &\ 1 &\ X &\ Y &\        &\ X^2 &\ XY &\ X^2Y\,\ \strutf{13pt}{5pt}\\ \cline{1-8}
\end{array}
\end{aligned}
\end{equation*}
\begin{equation*}
\mbox{\((3,5)\)-curve} : \ \ 
\begin{aligned}
\begin{array}{|c|ccccccccc|c}
\cline{1-10}
    j           &\ 1 &\ 2 &\ 3 &\  4  &\        &\  5 &\   6  &\   7  &\   8  \\\hline                                                            
w_{4-j+1}         &\ 7 &\ 4 &\ 2 &\  1  &\ \vrule &\  1 &\   2  &\   4  &\   7   &\ w_{j-4}\\ \hline                                                            
  v_j            &\ 0 &\ 3 &\ 5 &\  6  &\        &\  8 &\   9  &\  11  &\  14  \\ \cline{1-10}                                                      
\,M_{v_j}(X,Y)\,\ &\ 1 &\ X &\ Y &\ X^2 &\        &\ XY &\ X^3  &\ X^2Y &\ X^3Y\,\ \strutf{13pt}{5pt}\\\cline{1-10}
\end{array}
 \end{aligned}
\end{equation*}
All these examples, and only these, are of modality \(0\) as explained in (\ref{modality0}).
}
\end{example}
\begin{example}\label{positive_modality_example}
{\rm
In contrast to the cases above, we have for \((e,q)=(3,7)\)\,:
\begin{equation*}
\begin{aligned}
\begin{array}{|c|ccccccccccccc|c}
\cline{1-14}
 j              &\ 1\ &\ 2\ &\ 3\ &\ 4 &\ 5   &\  6 &\       &\  7   &\  8 &\  9  &\ 10   &\ 11   &\ 12             \\\hline
w_{6-j+1}         &\ 11 &\ 8 &\ 5   &\ 4 &\ 2   &\  1 &\ \vrule &\  1  &\  2 &\  4  &\  5   &\  8   &\ 11 &\ w_{j-6}   \\\hline                                            
 v_j             &\ 0  &\ 3 &\ 6   &\ 7 &\ 9   &\ 10 &\        &\ 12  &\ 13 &\ 15  &\ 16   &\ 19   &\ 22             \\\cline{1-14}                                         
\,M_{v_j}(X,Y)\,\ &\ 1 \ &\ \ X \ \ &\ X^2\ \ &\ Y \ \ &\ X^3 &\ XY &\       &\ X^4 &\ X^2Y &\ X^5 &\ X^3Y &\ X^4Y &\ X^5Y\,\ \strutf{13pt}{5pt}\\\cline{1-14}
\end{array}
\end{aligned}
\end{equation*}
However, the Weierstrass form of the \((3,7)\)-curve is given by
\begin{equation*}
\begin{aligned}
 Y^3&-(\mu_2X^4+\mu_5X^3+\mu_8X^2+\mu_{11}X
  +\mu_{14})Y\\
& -(X^7+\mu_6X^5+\mu_9X^4+\mu_{12}X^3+\mu_{15}X^3+\mu_{18}X+\mu_{21})
\end{aligned}
\end{equation*}
and this equation does not include a term in \(M_{22}(X,Y)=X^5Y\).
This curve is of modality \(1\). 
We will not discuss this curve further in this paper.  
} 
\end{example}%
\vskip 0pt
\newpage
\subsection{The representation matrix for \texorpdfstring{\(f(X,Y)\)}{Lg}-multiplication}%
In this subsection, we shall define a certain matrix
\(T\in\mathrm{Mat}(2g,\mathbb{Q}[\mu])\) whose determinant might be
essentially the discriminant of \(\varDelta\).  
\vspace{-2pt}
\begin{lemma}\label{basis_mod_f1_f2}%
As a \(\mathbb{Q}[\mu]\)-module, \(\mathbb{Q}[\mu][X,Y]/(f_1(X,Y),f_2(X,Y))\) is of rank \(2g\) 
and spanned by \(M(X,Y)\) which is defined in {\rm(\ref{def_all_M})}. 
\end{lemma}
\begin{proof}
(1) \,Let \(G(X,Y)\) be any element in \(\mathbb{Q}[\mu,X,Y]\). 
Paying attention to the weight with respect to \(X\) and \(Y\), 
we are reducing the terms in \(G(X,Y)\) to lower degree by using 
\(f_1(X,Y)=\cdots-qX^{q-1}+\cdots\) and \(f_2(X,Y)=eY^{e-1}+\cdots\). 
Whenever we reduce degrees of highest weight term(s) of \(G(X,Y)\) 
by using one of these relation, 
\(G(X,Y)\) is replaced by a lower degree polynomial. 
So we will finally arrived at a polynomial 
that is a linear combination of only terms in \(M(X,Y)\). 
\end{proof}
\begin{definition}\label{matrix_T}
The transpose of the representation matrix 
of the \((-eq)f(X,Y)\)-plication map
\begin{equation}\label{multiplication_f}
  \cdot(-eq)\,f(X,Y)\,:\,\mathbb{Q}[\vec{\mu}][X,Y]/(f_1,f_2)\longmapsto
  \mathbb{Q}[\vec{\mu}][X,Y]/(f_1,f_2)
\end{equation}
with respect to the basis \(M(X,Y)\) 
is denoted by 
\,\(T\in\mathrm{Mat}(2g,\,\mathbb{Q}[\mu])\)\,, 
that is
\begin{equation}\label{def_T_eq}
-eq\,f(X,Y)M(X,Y)\equiv M(X,Y)\,\tp{T}\,\bmod{(f_1,\,f_2)}.
\end{equation}%
We define the subscript of the entries in \(T\) as follows. 
The row-index \(a\) runs through \(\wt(M)\) in increasing order, 
and \(b\) runs through \(\wt(\widetilde{\mu})\) in decreasing order, 
and we write
\begin{equation*}
T=[\,T_{a,b}\,]\in\mathrm{Mat}(2g,\,\mathbb{Q}[\mu])
\end{equation*}
Then we have \,\(\wt(T_{a,b}){=}-(a+b)\).
This is written in the usual manner as \,\(T\,{=}\,[T_{v_i,\,eq-v_{2g+1-j}}]\).
\end{definition}
\vspace{-2pt}
One of the reasons why we have an extra factor \(-eq\) in (\ref{multiplication_f}) 
appears in \ref{det_V_and_det_T} later. 
\vspace{-3pt}
\begin{lemma}\label{jet}
We have the following{\rm:}
\begin{oitem}
\item The determinant \,\(\det(T)\) of the matrix \,\(T\) as defined in
  {\rm \ref{matrix_T}} is a non-zero rational constant multiple of a
  power of \,\(\varDelta\).
\item 
The linear map {\rm (\ref{multiplication_f})} is of rank \(2g\). 
\end{oitem}
\end{lemma} 
\begin{proof}
(1) 
We assume all the coefficients \(\mu_j\) are constants in \(\mathbb{C}\). 
We have \(\mathrm{det}(T)=0\) if and only if the
rank (called the Tjurina number at \(\mu\)) of the co-kernel 
\(\mathbb{Q}[\mu][X,Y]/(f,f_1,f_2)=\mathbb{Q}[x,y]/(f_1,f_2)\) 
(which is the dimension of \(\mathbb{C}[X,Y]/(f,f_1,f_2)\) over \(\mathbb{C}\))
of the map is positive.  
This is exactly the case that the ideal \((f,f_1,f_2)\) 
does not contain \(1\in\mathbb{Q}[\mu][X,Y]\).  
By Hilbert's Nullstellensatz 
(Theorem 5.4(i) in \cite{matsumura}, for instance), 
we see this is equivalent to  existience of a
\((x,y)\in\mathbb{C}^2\) such that
\begin{equation}\label{f_f1_f2}
  f(x,y)=f_1(x,y)=f_2(x,y)=0.
\end{equation}
Conversely \(\varDelta=0\) means (\ref{f_f1_f2}), and it implies 
the co-kernel is non-trivial. 
Therefore the zeroes of \(\varDelta\) 
and those of \(\det(T)\) coincide. 
So \(\det(T)\) must be a non-zero rational constant multiple 
of a power of \(\varDelta\). 
(2) \,Since \(\det(T)\neq 0\) by virtue of (1), this is obvious. 
\end{proof}
\vskip 5pt
Now we present the following\,:%
\begin{conjecture}\label{conj_D}
The determinant \,\(\det(T)\) is irreducible in \,\(\mathbb{Q}[\mu]\). 
\end{conjecture}%
\noindent
For each \((e,q)\), if this conjecture is true then we have immediately from \ref{jet} that
\begin{equation*}
\det(T)=c\cdot\varDelta \ \ \ \text{for some \(c\in\mathbb{Q}^{\times}\)}.
\end{equation*}
We give the explicit value of \(c\) above for 
the cases in Sections \ref{section_23-curve}, 
\ref{section_25-curve}, \ref{section_27-curve}, 
and \ref{section_34-curve}.%
\newpage
\section{The horizontal derivation formula}
Much of this section was lacking in the previous versions of
this paper.  However, the authors are pleased that the horizontal derivation formula (see
\ref{sato_formula} below) and its closely related results been used
 in the work of \cite{oss2024}, and these results make our arguments in the
rest of the present paper much clearer.
\subsection{Preliminary on derivations}
In this subsection, we prepare a general notation for the succeeding subsections. 
Let \(A\) be a commutative ring. 
A map \(D\,:\,A\longrightarrow A\), \(a\mapsto Da\) satisfying 
the following properties is called a \textit{derivation} on \(A\)\,;
for any \(a\), \(b\in A\), 
\begin{equation*}
D(a+b)=Da+Db, \ \ \
D(ab)=(Da)b+a(Db).
\end{equation*}
We denote by \(\Der(A)\) the set of all derivations on \(A\), 
which is an \(A\)-module. 
Since \(\Der(A)\) is equipped with the natural Lie bracket
\begin{equation*}
[D_1,\ D_2]=D_1D_2-D_2D_1 \ \ \ (D_1, \ D_2\in\Der(A)), 
\end{equation*}
\(\Der(A)\) is a Lie algebra over \(A\). 
Let \(M\) and \(N\) be two \(A\)-modules and 
\(D\in\mathrm{Der}(A)\) be a fixed derivation on \(A\). 
A map \(\widetilde{D}\,:\,M\longrightarrow N\), \(m\mapsto \widetilde{D}m\) satisfying 
\begin{equation*}
\widetilde{D}(m_1+m_2)=\widetilde{D}m_1+\widetilde{D}m_2, \ \ \
\widetilde{D}(am)=(Da)m+a\widetilde{D}m
\end{equation*}
is called a \textit{derivation from \(M\) to \(N\) associated with \(D\)}. 
We denote by \(\Der_A(M,N;D)\) the set of the derivations from \(M\) to \(N\) associated with \(D\). 
If \(M=N\), then we denote this simply by \,\(\Der_A(M;D)=\Der_A(M,M;D)\). 
Let \(D_1\), \(D_2\in\Der(A)\), and
\(\widetilde{D_1}\in\Der_A(M,N;D_1)\), \(\widetilde{D_2}\in\Der_A(M,N;D_2)\). 
Then we see that
\begin{equation*}
a\widetilde{D_1}\in\Der_A(M,N;aD_1) \ \ (a\in A), \ \ \
\widetilde{D_1}+\widetilde{D_2}\in\Der_A(M,N;\,D_1+D_2). 
\end{equation*}
Now let \(B\) be an \(A\)-algebra and take a derivative \(D\in\Der(A)\). 
Let \(M\) and \(N\) be two \(B\)-modules, hence also \(A\)-modules. 
Regarding \(B\) as an \(A\)-module, 
we take a derivation \(\widetilde{D}\in\Der_A(B;D)\) associated with \(D\in\Der(A)\). 
Since we can regard naturally \(\widetilde{D}\in\Der(B)\), 
we have the inclusion
\begin{equation*}
\Der_B(M,N;\widetilde{D})\subset\Der_A(M,N;D). 
\end{equation*}
For \(\widetilde{D_1}\in\Der_A(M;D_1)\), \(\widetilde{D_2}\in\Der_A(M;D_2)\), 
and for any \(a\in A\), \(m\in M\), we have
\begin{equation*}
  [\widetilde{D_1},\ \widetilde{D_2}]\ (=\widetilde{D_1}\widetilde{D_2}
  -\widetilde{D_2}\widetilde{D_1}) \in\Der_A(M;[D_1,D_2]), 
\end{equation*}
since  
\,\(\widetilde{D_1}\widetilde{D_2}(am)
=(D_1D_2a)m+(D_2a)(\widetilde{D_1}m)+(D_1a)(\widetilde{D_2}m)
+a\widetilde{D_1}\widetilde{D}_2m\). 
This implies that if \(\widetilde{D_1}\) and \(\widetilde{D_2}\in\Der_A(M;D)\) 
for a derivation \(D\in\Der(A)\), then
\begin{equation*}
[\widetilde{D_1},\ \widetilde{D_2}]\in\Der_A(M;0), 
\end{equation*}
where \(0\) stands for the \(0\)-map from \(A\) to \(A\) itself, which
is a derivation on \(A\).  Namely,
\([\widetilde{D_1},\ \widetilde{D_2}]\) is a \(A\)-homomorphism.  In
the set of all maps from \(M\) to \(M\), we take union of these sets
over \(\Der(A)\) and denote it as
\begin{equation*}
\Der_A(M)=\bigcup_{D\in\Der(A)}\Der_A(M;D).
\end{equation*}
This is also a Lie algebra over \(A\).  Instead of writing
\(D\in\Der_A(M)\), we can also say that the derivation \(D\) \textit{acts
  on}\, \(M\), for instance.
\par
\newpage
\subsection{Chevalley's Lemma}
We quote the following without proof 
from \cite{chevalley_1951}, p.112, Lemma 2. 
\begin{lemma}\label{unique_extension_der}
Let \(K\) be a field and \(R\) be a function field of one variable
with \(K\) the field of constants.  
Take a transcendental element \(\xi\) in \(R\) over \(K\) and fix it.  
For \(D\in\Der(K)\), there exists 
a unique derivation  \(D_{\xi}\in\Der_K(R;D)\) satisfying
\begin{equation*}
  D_{\xi}(\xi)=0. 
\end{equation*}
\vskip 3pt
\end{lemma}
\begin{corollary}\label{coroll_chevalley}
  Suppose a derivation \,\(D\in\Der(\mathbb{Q}[\mu])\)\, and an
  element \,\(\xi\) in \,\(\mathbb{Q}[\mu,x,y]\) transcendental over
  \(\mathbb{Q}(\mu)\) are given.  There is a unique extension
  \(D_{\xi}\in\Der_{\mathbb{Q}[\mu]}(\mathbb{Q}(\mu,x,y),D)\) of \(D\)
  satisfying
\begin{equation*}
  D_{\xi}(\xi)=0. 
\end{equation*}
\end{corollary}
\begin{proof}
We regard \(\xi\) to be an element in \(\mathbb{Q}(\mu,x,y)\). 
Since \(D\) extends to an element in \(\Der(\mathbb{Q}(\mu))\) as usual, 
the statement follows from \ref{unique_extension_der}. 
\end{proof}
\vskip 5pt Under the situation of \ref{unique_extension_der}, we
extends \(D_{\xi}\) to a derivation \(\widetilde{D_{\xi}}\) on the
space of differentials \(R\,d\xi\) via
\begin{equation}\label{der_on_forms}
  \widetilde{D_{\xi}}(\omega)=D_{\xi}\Big(\frac{\omega}{d\xi}\Big)d\xi \ \ \
  \mbox{for any \(\omega\in{R}\,d\xi\)}. 
\end{equation}
Namely, we have
\(\widetilde{D_{\xi}}\in\Der_R(R\,d\xi;\,D_{\xi})\subset\Der_K(R\,d\xi;\,D)\).
However, in the following lemma, \(\widetilde{D_{\xi}}\) is denoted by
\(D_{\xi}\) for simplicity.
\begin{lemma}{\rm (Chevalley \cite{chevalley_1951}, p.125, Lemma 3)} \
\label{chevalley}
Here we use the notation above. 
Let \(\xi\) and \(\zeta\) be two transcendental elements in \(R\) over \(K\). 
Then we have the following relation between \(D_{\xi}\) and \(D_{\zeta}\). 
For any \(w\in{R}\), we have
\begin{equation*}
D_{\xi}(wd\xi) - D_{\zeta}(wd\xi) = d(-wD_{\zeta}\xi).
\end{equation*}
\end{lemma}
\begin{proof} (From Manin \cite{manin_1958}) \ 
The operator
  \,\(D_{\xi}-D_{\zeta}+(D_{\zeta}\xi)\frac{d}{d\xi}\) \,is 
a derivative on \(R\) which vanishes on \(K\) 
and also  kills \(\xi\).  
By the uniqueness of extension of a derivation on \(K\) to \(R\),  
this vanishes on \(R\).  
So that
\begin{equation*}
  (D_{\xi}-D_{\zeta})w+(D_{\zeta}\xi)\frac{dw}{d\xi}=0.
\end{equation*}
Moreover \,
\((D_{\xi}-D_{\zeta})(wd\xi)=(D_{\xi}-D_{\zeta})w\cdot{d\xi}
+w\cdot(D_{\xi}-D_{\zeta})d\xi\). \ 
Since \(\zeta\) is transcendental,
we see \,\(\frac{d}{d\zeta}D_{\zeta}=D_{\zeta}\frac{d}{d\zeta}\)\, by 
\ref{commute_ell_and_d} below or \cite{chevalley_1951}, p.125, Lemma 1, 
and have
\begin{equation*}
  D_{\xi}d\xi=D_{\xi}\Big(\frac{d\xi}{d\xi}\Big)d\xi=0,\ \ \ 
  D_{\zeta}d\xi=D_{\zeta}\Big(\frac{d\xi}{d\zeta}\Big)d\zeta=
  \frac{d}{d\zeta}(D_{\zeta}\xi)d\zeta=d\big(D_{\zeta}(\xi)\big). 
\end{equation*}
Therefore
\begin{equation*}
  (D_{\xi}-D_{\zeta})(wd\xi)=-(D_{\zeta}\,\xi)\frac{dw}{d\xi}d\xi
  -w{\cdot}d(D_{\zeta}\,\xi)=-d(wD_{\zeta}\,\xi)
\end{equation*}
as desired. 
\end{proof}%
\vskip 5pt
For an explicit sample calculation, 
see Section \ref{section_23-curve} on the derivation \(\frac{\partial}{\partial\mu_j}\). 
\newpage
\subsection{The horizontal derivation formula}\label{horizontal_formula} 
We present a useful formula \ref{sato_formula} 
called the \textit{horizontal derivation formula}, 
which is important for our definition of the first de Rham cohomology 
equipped with a structure of a differential module and 
eventually giving a sophisticated algorithm to compute 
the Gauss-Manin connection  \(\varGamma_j\) (a sort of Christoffel symbol). 
The formula was obtained by Kouki Sato while working on \cite{oss2024}, 
and coauthors of that paper have permitted us to quote from it.  
This formula would be useful for various applications. 
\par
In addition to the matrix  \(T=[T_{ij}]\) in \ref{matrix_T}, 
we introduce \(A_{i+e}\) and \(B_{i+e}\in\mathbb{Q}[\mu,X,Y]\) 
for each \(i\in\wt(M)\) defined by
\begin{equation}\label{def_T_tilde}
  -eq\cdot{}M_i\cdot{f}(X,Y)=\sum_{j\,\in\wt(M)}T_{ij}M_j+A_{i+e}\,
  f_X+B_{i+q}\,f_{\,Y},
\end{equation}
where \(f(X,Y)\) is the defining polynomial 
of \(\mathscr{C}\) and \(M_i=M_i(X,Y)\) is the one defined in \ref{def_M's}. 
In this notation, we have 
\begin{equation*}
\wt(A_{i+e})=-(i+e), \ \ \ 
\wt(B_{i+q})=-(i+q).
\end{equation*}
We denote by 
\begin{equation*}
\frac{\,\partial\,}{\partial\,\mu_j\,}\in\Der_{\mathbb{Q}[\mu]}\Big(\mathbb{Q}(\mu,x,y);\frac{\,\partial\,}{\partial\,\mu_j\,}\Big)
\end{equation*}
the derivation uniquely determined by \ref{coroll_chevalley} as an extension of the derivation 
\(\frac{\,\partial\,}{\partial\,\mu_j\,}\in\Der(\mathbb{Q}[\mu])\) having the property
\begin{equation}\label{Dx=0}
\frac{\,\partial\,}{\partial\,\mu_j\,}x=0.
\end{equation}
Namely, we use the same notation for the extension.  It is natural to
define the weight to be
\begin{equation}\label{wt_on_d_dmu}
\wt\Big(\frac{\,\partial\,}{\partial\,\mu_j\,}\Big)=j.
\end{equation}
Moreover, there is a unique extension \(\widetilde{\ell_i}\) of the derivation 
\begin{equation}\label{small_ell_operators}
\ell_i=\sum_jT_{ij}\,\frac{\,\partial\,}{\partial\,\mu_j\,}\in\Der(\mathbb{Q}[\mu])
\end{equation}
satisfying 
\begin{equation*}
\widetilde{\ell_i}x=0,
\end{equation*}
by \ref{coroll_chevalley} again, for which we use the same notation 
\begin{equation}\label{def_ell}
\widetilde{\ell_i}=\sum_jT_{ij}\,\frac{\,\partial\,}{\partial\,\mu_j\,}
\in\Der_{\mathbb{Q}[\mu]}(\mathbb{Q}(\mu,x,y);\ell_i). 
\end{equation}
Therefore, for a 1-form \,\(\omega\in\mathbb{Q}(\mu,x,y)\,dx\), 
defining \(\overline{\ell_i}\) by
\begin{equation*}
\overline{\ell_j}(\omega)=\widetilde{\ell_j}\Big(\frac{\omega}{\,dx\,}\Big)dx
\end{equation*}
as (\ref{der_on_forms}) 
(be careful as this depends on the condition (\ref{Dx=0})), 
and we have an extension 
\begin{equation}\label{def_ell_dx}
\overline{\ell_i}\in\Der_{\mathbb{Q}[\mu]}(\mathbb{Q}(\mu,x,y)dx;\ell_i). 
\end{equation}
However, we omit \,\(\widetilde{}\)\, 
and \,\(\overline{\strutf{6pt}{0pt}\,\ }\)\, for simplicity, 
and use the notation 
\begin{equation*}
\ell_i=\sum_jT_{ij}\,\frac{\,\partial\,}{\partial\,\mu_j\,}
\end{equation*}
for these extensions. This convention will be practical. 
\par
In this paper, 
for a polynomial \(G=G(\mu,X,Y)\in\mathbb{Q}[\mu,X,Y]\), 
we denote by \((G)_x=G_x\), \((G)_y=G_y\) , and \((G)_{\mu_j}=G_{\mu_j}\) the polynomials 
given by the substitution \((X,Y)=(x,y)\) to the partial derivations 
\strutf{14pt}{6pt}\(\frac{\partial}{\,\partial X}\,{G}(\mu,X,Y)\), 
\strutf{0pt}{6pt}\(\frac{\partial}{\,\partial X}\,{G}(\mu,X,Y)\), and 
\strutf{0pt}{6pt}\(\frac{\partial}{\,\partial\mu_j}\,{G}(\mu,X,Y)\). 
This rule is applied also for \(f(X,Y)\). 
Moreover, the higher order case is similarly defined. 
For example, \(f_{yy}\) means the quantity given by 
the substitution \((X,Y)=(x,y)\)\, to \(\frac{\partial^2}{\,\partial Y^2\,}f(X,Y)\). 
\par
Under this convention, we present the following formula, 
which is mainly used for \ref{Def_dR_cohom}. 
\begin{theorem}\label{sato_formula}
{\rm (the horizontal derivation formula)}
Let \(N=N(X,Y)\) be a monomial of \(X\) and \(Y\) {\rm(}with non-negative powers{\rm)}, 
and \(P\in\mathbb{Q}[\mu]\) be an arbitrary polynomial. 
Then we have the following formula\,{\rm :}
\begin{equation}\label{HDF}
\begin{aligned}
\ell_j\Big(&\frac{\,P\cdot N\,}{f_y}\,dx\Big)
=P\,\frac{(N)_y\,B_{j+q}-N\big(M_j+(B_{j+q})_y\big)-(N\,A_{j+e})_x}{f_y}\,dx\\
&\hskip 200pt+\ell_j(P)N\,\frac{\,dx\,}{f_y}+d\bigg(P\,\frac{\,NA_{j+e}\,}{f_y}\bigg).
\end{aligned}
\end{equation}
Especially, if \(e=2\) then we have
\begin{equation*}
\ell_j\bigg(\frac{\,M_k\,dx\,}{f_y}\bigg)
=\frac{\,-(M_k\,A_{j+2})_x-M_{j+k}\,}{2y}\,dx+d\bigg(\frac{\,M_k\,A_{j+2}\,}{f_y}\bigg).
\end{equation*}
\end{theorem}
\vskip 8pt
\begin{remark}
{\rm 
The formula (\ref{HDF}) shows
\begin{equation}\label{ell_in_Der}
\ell_j\in\Der_{\mathbb{Q}[\mu]}\Big(\mathbb{Q}[\mu,x,y]\frac{\,dx\,}{f_y},\ 
\mathbb{Q}[\mu,\,x,\,y]\frac{\,dx\,}{f_y}
+d\,\Big(\mathbb{Q}[\,\mu,\,x,\,y]\frac{1}{\,f_y\,}\Big)
\,;\ \ell_j\Big). 
\end{equation}
Since \(\det(T)\) equals \(\varDelta\) up to a 
non-zero rational multiplicative constant 
(see \ref{jet}(3) and \ref{lemma_discri} below) 
at least in our cases  
\((e,q)=(2,3)\), \((2,5)\), \((2,7)\), and \((3,4)\), 
we also have
\begin{equation*}
\frac{\,\partial\,}{\partial\,\mu_j\,}
\in\Der_{\mathbb{Q}[\mu]}\Big(
\mathbb{Q}\Big[\frac{\,1}{\,\varDelta\,},\,\mu,\,x,\,y,\,\frac{1}{\,f_2(x,y)\,}\Big]\,;\ 
\frac{\,\partial\,}{\partial\,\mu_j\,}\Big).
\end{equation*}
}
\end{remark}%
\begin{remark}\label{remark2}{\rm 
(1) \,While operating \(\ell_j\) on an element in
\(\mathbb{Q}[\mu,x,y]\frac{\,dx\,}{f_y}\), some terms with
\({f_y}^2\) in its denominators appear.  But, it is seen by
(\ref{HDF}) or (\ref{ell_in_Der}) that we have only terms with
\(f_y\) of 1st degree in their denominators modulo the exact
forms.
\\
(2) \,Moreover, operating \(\frac{\partial}{\,\partial\mu_j\,}\)
on an element in \(\mathbb{Q}[\mu,x,y]\frac{\,dx\,}{f_y}\), we
may have the denominator with the discriminant \(\varDelta\) besides
\(f_y\) because of (\ref{HDF}) and the fact \(\det([T_{ij}])\)
is \(\varDelta\) times a rational.
\\
(3) \,Eventually, we have a nice definition \ref{Def_dR_cohom}
of the 1st de Rham cohomology
\(H^1_{\mathrm{dR}}(\mathscr{C}/\mathbb{Q}[\mu])\) and a good
differential module structure on it.  This makes a clear view of
our application of the theory of the Gauss-Manin connection.
\\
(4) \,One more advantage of the formula (\ref{HDF}) is that it
gives tremendously fast algorithm to compute many examples of
the systems of heat equations satisfied by the sigma functions.
Analysing such examples, the results in \cite{oss2024} are
obtained.  
}
\end{remark}
\vskip 10pt
We refer the reader to \cite{oss2024} on a proof of \ref{sato_formula}. 
Here we note the following equality for the proof of the next lemma. 
Applying \(\frac{\partial\,y}{\,\partial\mu_{eq-k}}\) to \(f(x,y)=0\), we have
\,\(f_y\,\frac{\partial\,y}{\,\partial\mu_{eq-k}}-M_k=0\), so that
\begin{equation}\label{dydmu_M_fy}
\frac{\partial\,y}{\,\partial\mu_{eq-k}}=\frac{\,M_k\,}{f_y}. 
\end{equation}
\newpage
\par
We give the following lemma, 
by which we understand (\ref{HDF}) below more clearly, 
is a special case of Lemma 1 in p.125 of \cite{chevalley_1951} 
(proved by the uniqueness assertion in \ref{unique_extension_der}). 
\begin{lemma}\label{commute_ell_and_d}
Denote by \(\frac{d}{\,dx\,}\) the derivation in 
\(\Der_{\mathbb{Q}[\mu]}(\mathbb{Q}(\mu,x,y);0)\)
determined by the properties that 
\(\frac{d\mu_k}{dx}=0\) for any \(k\in\wt(\mu)\) and 
\(\frac{d}{\,dx\,}x=1\). 
Then for any \(k\in\wt(\widetilde{\mu})\), we have
\begin{equation*}
 \frac{d}{\,dx\,}\frac{\partial y}{\,\partial\mu_k\,}
=\frac{\partial}{\,\partial\mu_k\,}\frac{dy}{\,dx\,}. 
\end{equation*}
Therefore, we have additionally
\begin{equation*}
 \frac{d}{\,dx\,}\ell_j=\ell_j\frac{d}{\,dx\,}
\end{equation*}
as a map from 
\(\mathbb{Q}(\mu,x,y)dx\) into itself. 
\end{lemma}
\begin{proof}
This is exactly Lemma 2 in \cite{chevalley_1951}, p.125, which
we can check by using (\ref{dydmu_M_fy}) as follows. 
The left-hand side is 
\begin{equation*}
\begin{aligned}
\frac{\,\partial\,}{\,\partial\mu_k\,}\frac{\,dy\,}{\,dx\,}
&=-\frac{\,\partial\,}{\,\partial\mu_k\,}\frac{\,f_x\,}{\,f_y\,}
=-\frac{\,\frac{\,\partial f_x\,}{\,\partial\mu_k\,}\,}{\,f_y\,}
+\frac{\,f_x\frac{\,\partial f_y\,}{\,\partial\mu_k\,}\,}{\,{f_y}^2\,}\\
&=-\frac{\,-(M_{eq-k})_x+f_{xy}\tfrac{\,\partial y\,}{\,\partial\mu_k\,}\,}{\,f_y\,}
+\frac{\,f_x\cdot\big(-(M_{eq-k})_y+f_{yy}\frac{\,\partial y\,}{\,\partial\mu_k\,}\big)\,}{\,{f_y}^2\,}\\
&=-\frac{\,-(M_{eq-k})_x+f_{xy}\tfrac{\,M_{eq-k}\,}{\,f_y\,}\,}{\,f_y\,}
+\frac{\,f_x\cdot\big(-(M_{eq-k})_y+f_{yy}\tfrac{\,M_{eq-k}\,}{\,f_y\,}\big)\,}{\,{f_y}^2\,}. 
\end{aligned}
\end{equation*}
On the other hand, we have
\begin{equation*}
\begin{aligned}
\frac{\,d\,}{\,dx\,}\frac{\,\partial y\,}{\,\partial\mu_k\,}
&=\frac{\,d\,}{\,dx\,}\frac{\,M_{eq-k}\,}{\,f_y\,} \ \ \ (\,\because \ \ref{dydmu_M_fy})\\
&=\frac{\,(M_{eq-k})_x+(M_{eq-k})_y\frac{\,dy\,}{\,dx\,}\,}{\,f_y\,}
-\frac{\,M_{eq-k}\big(f_{yx}+f_{yy}\frac{\,dy\,}{\,dx\,}\big)\,}{\,{f_y}^2\,}\\
&=\frac{\,(M_{eq-k})_x-(M_{eq-k})_y\frac{\,f_x\,}{\,f_y\,}\,}{\,f_y\,}
-\frac{\,M_{eq-k}\big(f_{yx}-f_{yy}\frac{\,f_x\,}{\,f_y\,}\big)\,}{\,{f_y}^2\,}.
\end{aligned}
\end{equation*}
These two coincide and the formula follows. 
\end{proof}
\vskip 3pt
\begin{remark}\label{pre_de_Rham}
{\rm
Thanks to \ref{sato_formula}, and 
\begin{equation*}
\ell_j\bigg(d\bigg(\frac{h}{\,f_2\,}\bigg)\bigg)
=\ell_j\bigg(\frac{d}{\,dx\,}\bigg(\frac{h}{\,f_2\,}\bigg)\bigg)dx
=\frac{d}{\,dx\,}\bigg(\ell_j\bigg(\frac{h}{\,f_2\,}\bigg)\bigg)dx
=d\!\bigg(\ell_j\bigg(\frac{h}{\,f_2\,}\bigg)\bigg)
\end{equation*}
for any \(h\in\mathbb{Q}[\mu,x,y]\) by \ref{commute_ell_and_d}, we see that
\begin{equation}\label{action_ell}
\begin{aligned}
\ell_j&\in
\Der_{\mathbb{Q}[\mu]}\Big(\mathbb{Q}[\mu,\,x,\,y]\frac{dx}{\,f_y\,}
+d\Big(\mathbb{Q}[\,\mu,\,x,\,y]\frac{1}{\,f_y\,}\Big);\,\ell_j\Big) \ \ \text{and}\\
\ell_j&\in
\Der_{\mathbb{Q}[\mu]}\Big(
\Big(\mathbb{Q}[\mu,\,x,\,y]\frac{dx}{\,f_y\,}
+d\Big(\mathbb{Q}[\,\mu,\,x,\,y]\frac{1}{\,f_y\,}\Big)\Big)\Big/
d\Big(\mathbb{Q}[\,\mu,\,x,\,y]\frac{1}{\,f_y\,}\Big)
\Big).
\end{aligned}
\end{equation}
While we use the same notation \(\ell_j\) for elements in different sets above, 
it is supposed that there is no confusion. 
}
\end{remark}
\vskip 3pt
\begin{definition}
For later use, we denote the Lie subalgebra generated by 
all the \(\ell_j\)s over the ring \(\mathbb{Q}[\mu]\)  
in the Lie algebra 
obtained as the union of the above sets by \(\vec{L}\)\,{\rm:}
\begin{equation}\label{the_L}
\vec{L}\subset
\Der_{\mathbb{Q}[\mu]}\left(
\mathbb{Q}[\mu,\,x,\,y]\frac{dx}{\,f_y\,}
+d\,\bigg(\mathbb{Q}[\,\mu,\,x,\,y]\frac{1}{\,f_y\,}\bigg)\right). 
\end{equation}
\end{definition}
\vskip 8pt
\begin{remark}
{\rm 
The inclusion above might be an equality. 
}
\end{remark}
\newpage
\section{de Rham cohomology etc.}\label{S_dR}
\subsection{Differential forms of the first kind of the curve}
Recall that any term of the sequence \(w_1\), \(w_2\), \(\cdots\),
\(w_g\in\mathrm{wgs}(e,q)\) in (\ref{gaps}) is written as
\begin{equation*}
w_j=2g-1-a_je-b_jq,
\end{equation*}
with non-negative integers \(a_j\), \(b_j\). 
Using this notation, 
we define differential forms
\begin{equation}\label{1st_kind}
\omega_{w_j}=\frac{x^{a_j}y^{b_j}}{f_2(x,y)}dx, \ \ \ 
(j=1, \ \cdots, \ g). 
\end{equation}
These are of the first kind 
(namely holomorphic everywhere on \(\mathscr{C}\)) 
and of weight \(w_j\).  
\vskip 10pt
\subsection{Forms of the second kind and the first de Rham cohomology}
\label{S3.2}
In this paper, we should
consider the curve \(\mathscr{C}\) and other objects arising from
\(\mathscr{C}\) to be defined over the ring \(\mathbb{Q}[{\mu}]\); 
in which the period matrix \(\varOmega\) is exceptional 
as is explained later, 
and is defined over the field \(\mathbb{C}\) of complex numbers. 
\par
Throughout this paper, we mean by the (\textit{differential}) 
\textit{forms of the second kind} 
all the forms without non-zero residue everywhere on \(\mathscr{C}\). 
Therefore we include in them the differential forms of the first kind. 
However, we mean by the differential forms of the third kind 
the forms having non-zero residue pole somewhere on \(\mathscr{C}\), 
on which we never assume that their order of such poles to be \(1\), 
and that they does not contain any form of the second kind. 
\par
It is easy to see that the set of differential forms of the second kind on \(\mathscr{C}\) with a pole only at \(\infty\) 
is exactly
\begin{equation*}
\mathbb{Q}[\mu,x,y]\frac{dx}{\,f_2(x,y)\,}.
\end{equation*}
In concord with our principal that we shall treat our materials not
over the fields \(\mathbb{Q}(\mu)\) but over the ring
\(\mathbb{Q}[\mu]\), we have the following definition thanks to the
horizontal derivation formula \ref{sato_formula} and
\ref{pre_de_Rham}.
\begin{definition}\label{Def_dR_cohom}
  We define the \textit{first de Rham cohomology} of \,\(\mathscr{C}\)
  over \(\mathbb{Q}[\mu]\) by
\begin{equation}\label{def_dR_cohom}
H_{\mathrm{dR}}^1(\mathscr{C}/\mathbb{Q}[{\mu}])
=\frac{\mathbb{Q}[\mu,x,y]\frac{dx}{\,f_2(x,y)\,}+d\,\big(\mathbb{Q}[\mu,x,y]
  \frac{1}{\,f_2(x,y)\,}\big)}
{d\,\big(\mathbb{Q}[\mu,x,y]\frac{1}{\,f_2(x,y)\,}\big)}. 
\end{equation}
\end{definition}
\vskip 5pt
\noindent
We are never concerned with higher cohomologies in this paper. 
We have an important remark here. 
\begin{remark}\label{on_def_of_dR}
{\rm
In order to define the first de Rham cohomology over \(\mathbb{Q}[\mu]\) 
which endures for calculation for \lq\lq horizontal\rq\rq\ differentiations, 
namely, for extensions of the elements in \(\Der(\mathbb{Q}[\mu])\), 
the right hand side of (\ref{def_dR_cohom}) is 
the\! \lq\lq slimmest form\rq\rq\ as explained in the following. 
For an integer \(k\), we denote by \(\omega_{\mathscr{C}}(k\cdot\infty)\) 
(resp. by \(L_{\mathscr{C}}(k\cdot\infty)\)) 
the module consists of the forms in 
\(\mathbb{Q}[\mu,x,y]\frac{dx}{\,f_2(x,y)\,}\) 
(resp. the functions (the polynomials) in \(\mathbb{Q}[\mu,x,y]\)) 
whose order of the pole at \(\infty\) is at most \(k\). 
Then we see by Theorem 8.2 in p.30 of \cite{lang_1982} that 
\begin{equation*}
H_{\mathrm{dR}}^1(\mathscr{C}/\mathbb{Q}[{\mu}])
\simeq\frac{\,\omega_{\mathscr{C}}((2g-2)\cdot\infty)\,}
      {\,dL_{\mathscr{C}}((2g-1)\cdot\infty)\,}. 
\end{equation*}
and this is a module of rank \(2g\) over \(\mathbb{Q}[\mu]\). 
We place emphasis on this that 
the derivatives \(\ell_j\)s in (\ref{def_ell_dx}) and 
the derivatives \(L_j\)s which appear in (\ref{ell_operators}) 
(and any elements in the Lie algebra generated by them) 
do not act on \(\omega_{\mathscr{C}}((2g-2)\cdot\infty)\) 
but on the module 
\(\mathbb{Q}[\mu,x,y]\frac{dx}{\,f_2(x,y)\,}+d\,\big(\mathbb{Q}[\mu,x,y]
\frac{1}{\,f_2(x,y)\,}\big)\) 
in (\ref{def_dR_cohom}). 
So that, the definition (\ref{def_dR_cohom}) is much suitable 
for explicit calculation 
on \(H_{\mathrm{dR}}^1(\mathscr{C}/\mathbb{Q}[{\mu}])\) as 
a differential \(\vec{L}\)-module, 
where \(\vec{L}\) is defined in (\ref{the_L}). 
}
\end{remark}
\subsection{Symplectic inner product}
We want to chose good \(g\) differential forms \(\eta_{-w_j}\)
(\(j=g\), \(\cdots\)\,, \(1\)) of the second kind and of weight
\(-w_j\) such that the \(2g\) forms consists of them and the forms of
the first kind in (\ref{1st_kind}) give rise to a symplectic basis of
the space \(H_{\mathrm{dR}}^1(\mathscr{C}/\mathbb{Q}[{\mu}])\) which
is equipped with a natural symplectic inner product explained below.  We
denote
\,\(\vec{\omega}=(\omega_{w_1},\cdots,\omega_{w_g},\eta_{-w_g},\cdots,\eta_{-w_1})\).
\par
As before, \((x,y)\) is a generic point of \,\(\mathscr{C}\).  It is
known that the \(\eta_{-w_j}\)s as well as the \(\omega_{w_j}\)s are
defined over \(\mathbb{Q}[{\mu}]\), namely, they are of the form
\(\frac{h(x,y)}{f_2(x,y)}dx\) with \(h(x,y)\in\mathbb{Q}[{\mu},x,y]\)
(see \cite{onishi_2018}, \cite{oss2024}).  We already defined in
(\ref{def_M's}) that
\begin{equation*}
M(x,y)=\{\,x^iy^j\,|\,0\leq i\leq q-2, \ 0\leq j\leq e-2\,\}. 
\end{equation*}
The set
\(\big\{\frac{h(x,y)}{f_2(x,y)}\,dx\,|\,h(x,y)\in M(x,y)\big\}\) 
forms a basis of \(H_{\mathrm{dR}}^1(\mathscr{C}/\mathbb{Q}[{\mu}])\) as
a \(\mathbb{Q}[{\mu}]\)-module.  
Here we note that the order of pole of \(x^{q-2}y^{e-2}\) 
at \(\infty\), which is the largest in \(M(x,y)\), is
\begin{equation*}
e(q-2)+q(e-2) = 2(e-1)(q-1)-2=2g-2, 
\end{equation*}
and \(H_{\mathrm{dR}}^1(\mathscr{C}/\mathbb{Q}[{\mu}])\) is a free
\(\mathbb{Q}[{\mu}]\)-module of rank \(2g\).  
\par
The space \(H_{\mathrm{dR}}^1(\mathscr{C}/\mathbb{Q}[{\mu}])\) is
equipped with the anti-symmetric inner product given by
\begin{equation*}
  \omega\star\eta=\mathop{\Res}_{P=\infty}
  \Big(\int_{\infty}^P\omega\Big)\eta(P)
\end{equation*}
for any \(\omega\),
\(\eta\in H_{\mathrm{dR}}^1(\mathscr{C}/\mathbb{Q}[{\mu}])\).
This is formally defined using the formal expansion with respect to a
local parameter at \(\infty\) (see \cite{onishi_2018}, \cite{oss2024}), 
which is a formal interpretation of the product
\(\omega\star\eta={\displaystyle\mathop{\Res}_{P\in\mathscr{C}^{\circ}}}
\big(\int_{\infty}^P\omega\big)\eta(P)\)
if we regard the \(\mu_j\)s as complex numbers and \(\mathscr{C}\) as a
non-singular curve (i.e. a compact Riemann surface).  Here
\(\mathscr{C}^{\circ}\) is the regular polygon obtained from the Riemann
surface attached to the curve \(\mathscr{C}\) 
with respect to the paths 
\begin{equation}\label{alpha_and_beta}
\{\alpha_j,\,\beta_j\,|\,j=1,\cdots,g\}
\end{equation}
which form a symplectic basis of the homology group
\(H_1(\mathscr{C},\mathbb{Z})\) as usual.
\par
We choose \(\eta_{-w_i}\) to satisfy the symplectic relations
\begin{equation}\label{symplectic}
\omega_{w_i}\star\omega_{w_j}=0, \ \ 
\eta_{-w_i}\star\eta_{-w_j}=0, \ \ 
\omega_{w_i}\star\eta_{-w_j}=-\eta_{-w_j}\star\omega_{w_i}=\delta_{ij}. 
\end{equation}
The choice of \(\eta_{-w_j}\) is not unique but given by using
so-called the \textit{fundamental \(2\)-form of Klein}.  For a more
concrete construction of these forms, we refer the reader to
\cite{onishi_2018} (or to \cite{oss2024} for more detailed
description).  Especially, as a \(\mathbb{Q}[{\mu}]\)-module, the
\(\omega_{w_j}\)s and \(\eta_{-w_j}\)s form a basis of
\(H_{\mathrm{dR}}^1(\mathscr{C}/\mathbb{Q}[{\mu}])\).
\par
\newpage
\section{The sigma function}\label{section_const_sigma}
\subsection{Materials for the construction of the sigma functions}
\label{materials_for_sigma}
In this section, 
we recall the definition of the natural generalisation of 
(\ref{weierstrass_sigma}) for the general curve \(\mathscr{C}\), 
which is a function of \,\(g\)\, variables
\,\(u=\tp{[u_{w_g} \ \cdots \ u_{w_1}]}\)\, with \(w_j\in\mathrm{wgs}(e,q)\) 
and written as 
\(\sigma(u)\) \(=\sigma(u_{w_g},\cdots,u_{w_1})\).  
It is called, analogously, the sigma function for \(\mathscr{C}\).  
To define it precisely, we need to introduce the Schur polynomial, 
period matrices, and others.
\par
We use the classical notation of matrices concerning theta series, so
our notation is transposed from the notation of \cite{bl_2008}.
Specifically, we will denote the period matrix of a curve as
(\ref{per_matrix}) whereas Buchstaber and Leykin's papers use the
transpose of this matrix.  The other differences between their
notation and ours will follow from this.
\par
Letting \(\vec{T}=\sum_{j=1}^gu_{w_j}T^{w_j}\) 
with an indeterminate \(T\), 
we define \(\{p_k\}\) 
by \,\(p_k=0\) \,for negative \,\(k\) \,and 
\begin{equation*}
  \sum_{k=0}^{\infty}p_kT^k=\sum_{n=0}\frac{\vec{T}^n}{n!}.
\end{equation*}
Then we define \(s(u)=s(u_{w_g},\cdots,u_{w_1})\) 
  (see Section 4 in \cite{nakayashiki_2010}) by
\begin{equation}\label{schur_pol}
  s(u)=\det(\mathop{[\ p_{w_{g+1-i}+j-g}\ ]}_{1\leq i\leq g,\ 1\leq j\leq g}). 
\end{equation}
This is the {\it Schur polynomial}\, corresponding to the
sequence \((w_g,\cdots,w_1)\). 
\par
From now on in this Section, we assume 
all the  \(\mu_j\)s are complex numbers and \(\mathscr{C}\) is
non-singular. 
The period matrices are defined by 
\(\varOmega=\bigg[\ 
\begin{matrix}
\omega' &\ \omega'' \\
\eta'   &\ \eta''
\end{matrix}\ 
\bigg]\) 
with 
\begin{equation}\label{per_matrix}
\begin{aligned}
\omega' &=\bigg[\int_{\alpha_j}\!\omega_{w_{g-i+1}}\bigg], \ \ 
\omega''=\bigg[\int_{\beta_j}\!\omega_{w_{g-i+1}}\bigg], \\
\eta' &=\bigg[\int_{\alpha_j}\!\eta_{-w_{g-i+1}}\bigg], \ \ 
\eta''=\bigg[\int_{\beta_j}\!\eta_{-w_{g-i+1}}\bigg],
\end{aligned}
\end{equation}
where the \(\eta_{-w_j}\)'s are ones of (\ref{symplectic}) and 
\(\alpha_i\)s and \(\beta_j\)'s are ones of (\ref{alpha_and_beta}). 
We introduce the \(g\)-dimensional space \(\mathbb{C}^g\) 
with coordinates
\begin{equation}\label{u_coordinates}
u=\tp{[u_{w_g} \ u_{w_{g-1}} \ \cdots \ u_{w_1}]}
\end{equation}
for the domain on which the sigma function is defined. 
We define a lattice in this space by 
\begin{equation}\label{Lambda}
\varLambda=\omega'\mathbb{Z}^g+\omega''\mathbb{Z}^g.   
\end{equation}
For any \(u\in\mathbb{C}^g\), we define \(u'\), \(u''\in\mathbb{R}^g\)
by \, \(u=\omega'u'+\omega''u''\).  
Likewise, for \(\ell\in\varLambda\), 
we write \(\ell=\omega'\ell'+\omega''\ell''\) 
with \(\ell'\), \(\ell''\in\mathbb{Z}^g\). 
In addition, we write 
\({\omega'}^{-1}\tp{(\omega_{w_g},\,\cdots,\,\omega_{w_1})}=
\tp{(\hat{\omega}_1,\,\cdots,\,\hat{\omega}_g)}\),
\({\omega'}^{-1}\omega''=\big[\tau_{ij}\big]\), and define 
(see \cite{mumford_tata2}, p.3.82 or \cite{lewittes_1964}, p43)
\begin{equation}\label{riemann_const}
\delta_j=-\tfrac12\tau_{jj}
+\sum_{i=1}^g\int_{\alpha_i}\bigg(\int_{\infty}^{P}\hat{\omega}_j\bigg)\,
\hat{\omega}_i(P),\ \ \ 
\delta={\omega'}\,\tp{[\delta_1 \ \cdots \ \delta_g]}\in\mathbb{C}^g.
\end{equation}
Then we define the Riemann constant by \([\delta'\ \delta'']\).  
It is well known that \(\delta'\), \(\delta''\in(\frac12\mathbb{Z})^g\) 
for our curve.  
The \([\delta'\ \delta'']\) can be taken independent 
of the values of \(\mu_j\)s with a suitable choice of 
the paths of integrals in (\ref{riemann_const}). 
\par
Using the above notation, we define a linear form \(L(\ ,\ )\) 
on \(\mathbb{C}^g\times\mathbb{C}^g\) by
\begin{equation}\label{linear_form_L}
  L(u,v)=\tp{u}\,(\eta'v'+\eta''v'').
\end{equation}
This is \(\mathbb{C}\)-linear with respect to the first variable, 
and \(\mathbb{R}\)-linear with respect to the second. 
Moreover, we define 
\begin{equation}\label{chi}
  \chi(\ell)=\exp\,\big(2\pi\iu(\tp{\delta'}\ell''-\tp{\delta'}\ell'
  +\frac12\tp{\ell'}\ell'')\big)\in\{1,\,-1\}
\end{equation}
for any \(\ell\in\varLambda\).  
Let 
\begin{equation}\label{kappa}
\kappa\,:\,\mathbb{C}^g\,\longmapsto\,\mathbb{C}^g/\varLambda
\end{equation}
be the map given by modulo \(\varLambda\), 
and \(\mathrm{Sym}^k\,\mathscr{C}\) be 
the \(k\)-th symmetric product of \(\mathscr{C}\). 
Then we have the Abel-Jacobi mapping
\begin{equation}\label{iota}
\iota\,:\,\mathrm{Sym}^k\,\mathscr{C}
\longrightarrow\ \mathbb{C}^g/\varLambda,
\ \ \ 
(P_1,\,\cdots\!,\,P_k)\longmapsto\ 
\sum_{j=1}^k\bigg( 
\int_{\infty}^{P_j}\omega_{w_g},
\,\cdots,\,
\int_{\infty}^{P_j}\omega_{w_1}
\bigg),
\end{equation}
whose image is denoted by \(\Theta^{[k]}\).  
We denote \(\Theta=\Theta^{[g-1]}\), 
which is called the standard theta divisor 
of the Jacobian variety 
\(\mathbb{C}^g/\varLambda\) of \(\mathscr{C}\).  
For \(k=1\), the map \(\iota\) is an isomorphism from \(\mathscr{C}\) 
to \(\Theta^{[1]}\), by which we frequently identify these two.
\vskip 20pt
\subsection{Weight revisited}\label{ext_wt}
In this subsection, we assume all \(\mu_j\)s in (\ref{def_eq}) are
complex numbers.  We denote by \(\mathscr{C}_{\mu}\) the curve given
by (\ref{def_eq}) and assume it is non-singular, namely, we assume it
gives a Riemann surface.  Here, we shall extend the notion of the
weight \,\(\wt\) of (\ref{wt_on_C}) as follows.  
Firstly, we define the weight of
\,\(u_j\), a coordinate of \(u\) in (\ref{u_coordinates}), by
\begin{equation}\label{wt_u}
\wt(u_j)=j.
\end{equation}
We explain below that this definition is consistent with the settings
of weight so far.  For a later use in (\ref{alphabetagamma}) for
instance, we prepare another notation
\begin{equation*}
\wt(u)=\{\,\wt(u_{w_j})\,|\,j=1,\,\cdots,\,g\,\}=\mathrm{wgs}(e,q), 
\end{equation*}
too. 
Let \(\varepsilon\) be an arbitrary non-zero complex number 
and \(\mathscr{C}_{\varepsilon\mu}\) be the curve defined by
(\ref{def_eq}) with every \(\mu_j\) being replaced by \(\varepsilon^{-j}\mu_j\). 
Then the map 
\begin{equation*}
\vec{\varepsilon}\,:\,
\mathscr{C}_{\mu}\longrightarrow\mathscr{C}_{\varepsilon\mu}, \ \ \ 
(x,y)\longmapsto (\varepsilon^{-e}x, \varepsilon^{-q}y)
\end{equation*}
is an isomorphism (i.e.\ an isomorphism of compact Riemann surfaces). 
Now we extend the defining domain of \(\mathrm{wt}\) 
to the space \(\mathbb{C}^g\) of (\ref{u_coordinates}) 
as follows. 
We take a fixed reference point \((\{{\mu_j}^{(0)}\},\,u^{(0)})\) 
of \((\mu,\,u)=(\{\mu_j\},\,u)\) of (\ref{def_eq}) and (\ref{u_coordinates}), 
and assume \,\(\mathscr{C}_{\mu^{(0)}}\)\, is also non-singular. 
Let 
\begin{equation}\label{M}
\mathscr{M}=\mathscr{M}(\mu^{(0)},\,u^{(0)})
\end{equation}
be the germ of the analytic functions of \((\mu,\,u)\) at
\((\mu^{(0)},\,u^{(0)})\), where we regard them as complex variables.
Let \(\varPhi(\mu,u)\in\mathscr{M}\).  We denote by
\(\varPhi(\varepsilon\mu,\varepsilon{u})\) the function obtained from
\(\varPhi(\mu,u)\) by replacing all \(\mu_j\) by
\(\varepsilon^{-j}\mu\) and all \(u_i\) by \(\varepsilon^iu_i\) for
\(i=w_g\), \(\cdots\), \(w_1\).  We say that the function obtained is
induced by the mapping \,\(\vec{\varepsilon}\).  If there is a
constant integer \(w\) such that
\begin{equation*}
\varPhi(\varepsilon\mu,\varepsilon{u})=\varepsilon^{w}\,\varPhi(\mu,{u})
\end{equation*}
for any  \(\varepsilon\in\mathbb{C}\), 
then we say the function \(\varPhi(\mu,u)\) is of weight \(w\), and denote this as
\begin{equation*}
\wt(\varPhi(\mu,u))=w.
\end{equation*}
\newpage
\par
Now we explain that the definition of \(\wt\) is actually an extension
of the weight defined at (\ref{wt_on_C}).  The coordinates \(x\) and
\(y\) of \(\mathscr{C}_{\mu}\) are naturally seen as functions of
\(u_g\) of \(u\) in (\ref{u_coordinates}) on certain restricted domain
in \(\mathbb{C}^g\) for any \(\mu_j\in\mathbb{C}\), which should be
denoted by \(\kappa^{-1}\iota(\mathscr{C})\) by using the notation
(\ref{kappa}) and (\ref{iota}).  Then, observing the weight of \(x\)
and \(y\) as power series of \(u_g\), their weights are \(-e\) and
\(-q\), respectively.  The function \(u_g\mapsto x\) above is seen as
the restriction to \(\kappa^{-1}\iota(\mathscr{C})\) of the Abelian
function
\begin{equation}\label{g_sum}
u\longmapsto 
\frac{x_1\cdots x_g}
     {\displaystyle\sum_{1\leq i_1<\cdots<i_{g-1}\leq g}\!\!\!\!\!x_{i_1}\cdots{}x_{i_{g-1}}},
\end{equation}
where \,\(u=\iota((x_1,y_1),\,\cdots,\,(x_g,y_g))\,\bmod{\varLambda}\).
The function (\ref{g_sum}) is easily checked to be of weight \(-e\). 
Under similar observation, the function \(u_g\mapsto y\) is seen of weight \(-q\). 
Summarizing the above, the notion of weight of (\ref{wt_on_C}) is completely compatible 
with the former notion above as well as one on \(\mathscr{M}\).
\par
Take any circuit integrals of an entry of one of the matrices 
in (\ref{per_matrix}) and consider the map \(\vec{\varepsilon}\). 
For example, we choose \(\int_{\alpha_j}\omega_{w_i}\). 
Note that it is an element in \,\(\mathscr{M}\) 
which is a function on \(\mu_j\)s and independent of \(u\). 
While the map \(\vec{\varepsilon}\) changes 
\(\omega_{w_i}\) to \(\varepsilon^{w_i}\omega_{w_i}\) 
as \(\wt(\omega_{w_i})=w_i\), 
the deformation of the path \(\alpha_j\) of the integral 
no longer affects this change of value integral 
because of Cauchy's theorem. 
So that, the mapping \,\(\vec{\varepsilon}\) induces a change of the integral
\begin{equation*}
\int_{\alpha_j}\omega_{w_i} \longmapsto 
\ \varepsilon^{w_i}\!\int_{\alpha_j}\omega_{w_i},
\end{equation*}
for any non-zero \(\varepsilon\). 
Therefore, we have
\begin{equation*}
\wt\Big(\,\int_{\alpha_j}\omega_{w_i}\Big)=w_i.
\end{equation*}
By the argument above, the weight of such a circuit integral does not depend on 
its path of integral as well as the choice of 
the reference points \({\mu_j}^{(0)}\)s and \(u^{(0)}\). 
Because of (\ref{iota}) for \(k=g\) is surjective (Jacobi's theorem) 
and any values \(\int_{\alpha_j}\omega_{w_i}\)s are appeared 
as the coordinates \(u_{w_j}\) the definition of weight (\ref{wt_u}) is justified. 
\subsection{Construction of the sigma function}
\label{construct_sigma}
Using the notions defined in the previous subsections, 
we define the sigma function by the following characterization. 
Now we fix the curve \,\(\mathscr{C}=\mathscr{C}_{\mu}^{e,q}\). 
\begin{theorem}\label{char_sigma}
There exists a unique function 
\,\((\mu,u)\mapsto\sigma(u)\,{=}\,\sigma(\mu,u)\)\, 
having the following properties\,{\rm:} 
\begin{oitem}
\item 
\(\sigma(u)\) is an entire function on \,\(\mathbb{C}^g\)\, 
for any fixed \(\mu_j\)s in \(\mathbb{C}\)\,{\rm;} 
\item 
Supposing that the \(\{\mu_j\}\) are constants in \(\mathbb{C}\)
and the discriminant \(\varDelta\) is not zero, we have
\begin{equation*}
\sigma(u+\ell)=\chi(\ell)\,\sigma(u)\exp\,L(u+\tfrac12\ell,\ell) \ \ \ 
(u\in\mathbb{C}^g, \ \ell\in\varLambda), 
\end{equation*}
where \(\varLambda\), \(L\), and \(\chi\) are those of {\rm(\ref{Lambda})}, 
{\rm (\ref{linear_form_L})}, and {\rm(\ref{chi})},
respectively\,{\rm;}
\item 
Viewing \(\sigma(u)\) as an element in \,\(\mathscr{M}\)\, of {\rm(\ref{M})}
with an arbitrary reference point
\begin{equation*}
\wt(\sigma(u))=\tfrac1{24}{(e^2-1)(q^2-1)};\,
\end{equation*}
Moreover, \(\sigma(u)\) is expanded as a power series 
around the origin \((0,\cdots,0)\) with coefficients 
in \(\mathbb{Q}[{\mu}]\), and is of homogeneous weight 
\((e^2-1)(q^2-1)/24\)\,\mbox{\rm ;}
\item 
\(\sigma(u)|_{{\mu}=\vec{0}}\) is 
the Schur polynomial  \,\(s(u)\) of {\rm(\ref{schur_pol})}{\rm;}
\item 
\(\sigma(u)=0\) \(\iff\) \(u\in\kappa^{-1}(\Theta)\).
\end{oitem}
\end{theorem}
\vskip 8pt
\begin{definition}
We call the function \(\sigma(u)\) whose existence is 
guaranteed by {\rm\ref{char_sigma}} 
the {\it sigma function}\, of the curve \(\mathscr{C}\). 
\end{definition}
The theorem \ref{char_sigma} was proved in various ways in the literature, 
each version has a slightly different point of view. 
It is convenient to summarise here these versions from our point of view.
We define
\begin{equation}\label{classical_tilde_sigma}
\begin{aligned}
  \tilde{\sigma}(u)&=\tilde{\sigma}(u,\varOmega)
=\Big(\frac{(2\pi)^g}{\mathrm{det}\,{\omega'}}\Big)^{\tfrac12}
  \exp\big(-\tfrac12\,\tp{u}\,\eta'{\omega'}^{-1}u\big)\\
  &\cdot\sum_{n\in\mathbb{Z}^g}
  \exp\big(\tfrac12\,\tp{(n+\delta'')}\,{\omega'}^{-1}\omega''(n+\delta'')
  +\tp{(n+\delta'')}({\omega'}^{-1}u+\delta')\big),
\end{aligned}
\end{equation}
where \(\mathrm{det}\) denotes the determinant.  
Note that the arguments in any exponential are of weight \(0\), 
so that the infinite series part of (\ref{classical_tilde_sigma}) is 
of weight \(0\) as well. 
It is easy to show that this function has property (1) 
(see \ref{heq_tilde_sigma} below), 
and that this function satisfies (2) using (\ref{legendre}) below.  
Frobenius' method shows that the solutions of the equation (2) form 
a one dimensional space (see p.93 of \cite{lang_1982}).  
Although this function is constructed by using \(\varOmega\), 
it is independent of \(\varOmega\).  
Namely, the function \(\tilde{\sigma}(u,\varOmega)\) is
invariant under a modular transform, 
that is the transform of \(\varOmega\) by \(\mathrm{Sp}(2g,\mathbb{Z})\)
which come from changing the choice of \(\alpha_i\)s and \(\beta_i\)s. 
Therefore, it is expressed as a power series of \(u\) 
with coefficients being functions of only the \(\mu_j\)s.  
On the latter part of (3) and (4), we refer the reader to \cite{nakayashiki_2010}. 
In that paper, 
\(\tilde{\sigma}(u)\) times some constant is expressed as a
determinant of infinite size (see also \cite{onishi_2018}). 
See also the paper \cite{nakayashiki_2008} by Nakayashiki, 
in which he proved the sigma function is no other than (\ref{classical_tilde_sigma}) 
times a non-zero constant depending on \(\mu_j\)s with emphasizing the later part of (3). 
That \(\tilde{\sigma}(u)\) has the property (5) 
come from a well-known property of Riemann theta function. 
Using notation we have explained previously, we define 
\begin{equation}\label{classical_sigma}
\hat{\sigma}(u)=\varDelta^{-\frac18}\tilde{\sigma}(u),
\end{equation}
with an appropriate choice of the \(\frac18\)th 
root of the discriminant \(\varDelta\). 
It is known for \(g=1\) and \(2\) that this function exactly 
satisfies all the properties in \ref{char_sigma}, 
namely, we have \(\sigma(u)=\hat{\sigma}(u)\). 
For \(g=1\), it is shown as in \cite{onishi_2010a} by 
a transformation formulae for \(\eta(\tau)\) and 
the theta series described in pp.176--180 of \cite{rademacher_1973},
and for \(g=2\) the paper \cite{grant} by D.~Grant, 
in which the property (4) is shown by using Thomae's formula. 

However, it is not known for \(g\geq3\) if the function (\ref{classical_sigma}) is 
really the sigma function. 
The paper \cite{bl_2008} is the first one to seriously attack this problem. 

In this paper, we show, 
following the idea of \cite{bl_2008}, 
that \(\sigma(u)=\hat{\sigma}(u)\) for the genus \(3\) curves 
i.e. for the \((2,7)\)-curve and the \((3,4)\)-curve.  
This means (\ref{classical_sigma}) satisfies especially (4) up to 
an absolutely numerical multiplicative constant.  
The strategy of the proof is as follows: We construct a system of 
linear partial differential equations (heat equations) 
satisfied by \(\hat{\sigma}(u)\) 
and show that the solution space is of dimension one 
(over the base field) by {\it explicit construction} of a
recursive system for the coefficients of the power series expansion
of any unknown solution and {\it showing the uniqueness} of the
solution of this system.  
Then we see \(\sigma(u)=\hat{\sigma}(u)\) up to 
a non-zero multiplicative absolutely numerical constant.  
This result might be seen as a generalisation 
of {\it Thomae's formula} (\cite{thomae_1870}).
Section \ref{section3} is devoted to these last steps. 
\vskip 3pt
\newpage
\section{Theory of heat equations}\label{section2}
\subsection{Generalisation of the Frobenius-Stickelberger theory} 
\label{Frob_Stick}
This and the following sections are devoted to explaining the theory of
Buchstaber and Leykin \cite{bl_2008}, on the differentiation of
Abelian functions with respect to their parameters, 
as clearly as we can. 
That generalises the work of Frobenius and Stickelberger
\cite{fs_1882}, discussed above, 
on the elliptic function case of this problem.

For higher genus cases, we do not have a naive generalisation of
(\ref{wp_expansion}) and (\ref{g2_g3}), which are mentioned 
in the Introduction.  
However, we can give a natural generalisation of the
relations (\ref{L_for_genus1}) to the curve \(\mathscr{C}_{\mu}\) as
explained in the next section. 
\par
As we discussed in Subsection \ref{S3.2},
any element \(L\) in \(\vec{L}\) of (\ref{the_L}) 
operates linearly on the space
\(H_{\mathrm{dR}}^1(\mathscr{C}/\mathbb{Q}[\mu])\).  
For complex variables \(\mu_j\)s, 
we let \(L\) operate firstly on the forms with 
{\it variables} \,\(\mu_j\)s\, of representative of basis of 
\(H_{\mathrm{dR}}^1(\mathscr{C}/\mathbb{Q}[\mu])\), \
then we restore \(\mu_j\)s to the original values in \(\mathbb{C}\). 
Taking a derivative \(L\in\vec{L}\), 
we define \(\Gamma^L\in\mathrm{Mat}(2g,\mathbb{Q}[{\mu}])\)  
as the representation matrix of the action of \(L\) by
\begin{equation}\label{gauss_manin_connection}
  L(\vec{\omega})=\vec{\omega}\,\tp{\Gamma^L} \ \ \ 
  \mbox{for} \ \ 
  \vec{\omega}=(\omega_{w_g},\,\cdots,\,\omega_{w_1},\,\eta_{-w_g},\,\cdots,\,\eta_{-w_1})
  \in H_{\mathrm{dR}}^1(\mathscr{C}/\mathbb{Q}[{\mu}]).
\end{equation}
In \cite{bl_2008}, the matrix \(-\!\tp(\Gamma^L)\) defined by 
(\ref{gauss_manin_connection}) is called 
the {\it Gauss-Manin connection}\, for the derivation (vector field) \(L\).

From here to the end of Section \ref{section2}, 
we assume that all the \(\mu_j\)s are complex variables 
which vary as \(\mathscr{C}\) is non-singular. 
We can then use the periods \(\omega_{ij}\)s and \(\eta_{ij}\)s.

By integrating (\ref{gauss_manin_connection}) along each element in the
chosen symplectic basis of \(H_1(\mathscr{C},\,\mathbb{Z})\), 
we get the natural action
\begin{equation}\label{bridge_formula}
  L(\varOmega)=\Gamma^L\,\varOmega
\end{equation}
of \(L\) on the space of \(\varOmega\)s for all the choices of 
symplectic basis of \(H_1(\mathscr{C},\,\mathbb{Z})\). 
So, we see also how \(L\) operates on 
the field \(\mathbb{Q}(\{{\omega'}_{ij}\},\,\{{\omega''}_{ij}\})\). 
Of course, since \(\varOmega\) is the period
matrix of a symplectic basis, these elements must satisfy the
constraint
\begin{equation}\label{legendre}
  \tp{\varOmega} J \varOmega = 2 \pi i J, \ \ \ 
  \mbox{where \(J=\bigg[\begin{matrix}  & 1_g\\ -1_g & \end{matrix}\bigg]\)}
\end{equation}
by (\ref{symplectic}).  
This is none other than the generalisation of 
Frobenius-Stickelberger's relation (\ref{L_for_genus1}),  
and is to say that 
\((2\pi{i})^{-\frac12}\varOmega\in \mathrm{Sp}(2g,\mathbb{C})\).  
It follows immediately that
\begin{equation}\label{legendreinv}
  \varOmega^{-1}   =\frac{1}{2\pi i}J^{-1}\,
  \tp{\varOmega}J  =\frac{1}{2\pi i}
  \bigg[
\begin{array}{rr}
  \tp{\eta''} &\ -\tp{\omega''}\\
  -\tp{\eta'} &\  \tp{\omega'}
\end{array}
  \bigg].
\end{equation}
After operating \(L\) on both sides of (\ref{legendre}), using
(\ref{bridge_formula}) and (\ref{legendre}), we see that the matrix
\(\Gamma^L\) satisfies
\begin{equation}\label{Gamma_property}
  \tp\Gamma^LJ + J \Gamma^L = 0,
  \ \ \mbox{i.e.}\ \ 
  \tp{(\Gamma^LJ)}=\Gamma^LJ
\end{equation}
because \(L(\tp{\varOmega})=\tp{\varOmega}\,\tp{\Gamma^L}\), 
which is to say that \,\(\Gamma^L\)\, is in 
the Lie algebra \,\(\mathfrak{sp}(2g,\mathbb{Q}[\mu])\)\,  
of \(\mathrm{Sp}(2g,\mathbb{Q}[\mu])\). 
Thus we may write
\begin{equation}\label{GammaL}
  -\Gamma^LJ=\bigg[\,
\begin{matrix}
     \alpha & \ \beta \\
  \tp{\beta}& \ \gamma
\end{matrix}\,
\bigg], \ \ \ \Gamma^L = \bigg[
\begin{array}{cr} 
  -\beta  &  \ \alpha \\ 
  -\gamma & \ \tp{\beta}
\end{array}
  \bigg]
\end{equation}
with \(\tp{\alpha}=\alpha\) and \(\tp{\gamma}=\gamma\). 
\begin{remark}\label{dictionary_notation}
{\rm
(1) \,We use a different notation for \(D(x,y,\lambda)\) 
compared to p.273 of \cite{bl_2008}, 
and for \(\Omega\), \(\Gamma\) and \(\beta\) 
compared to p.274 of loc.\ cit.
Our \(\tp{\vec{\omega}}\) equals \(D(x,y,\lambda)\) by
transposing and changing the sign on the latter half entries. 
The others are naturally modified according to this difference 
and taking transposes.  We will give a detailed comparison of our
notation with theirs in \ref{dictionary} 
at the end of Subsection \ref{sigma_is_hat_sigma}. \\
(2) \,In general, it requires some work to write down 
the given operator \(L\) as a partial differential operator
with respect to the periods \({\omega'}_{ij}\) and
\({\omega''}_{ij}\) similar to the LHS of (\ref{L_for_genus1}).
However, we do not use such an expression in the present paper.
} 
\end{remark}
Conversely, starting from a matrix 
\begin{equation*}
  \Gamma
  = \bigg[
    \begin{array}{cr}
      -\beta  & \ \alpha \\ 
      -\gamma & \ \tp{\beta}
    \end{array}
    \bigg]\in\mathfrak{sp}(2g,\mathbb{Q}[{\mu}]),
\end{equation*}
with \(\tp{\alpha}=\alpha\) and \(\tp{\gamma}=\gamma\), we get
uniquely an operator
\(L\in\vec{L}\)\, such that \(\Gamma^L=\Gamma\).
So far, this is a natural generalisation of the situation investigated
by Frobenius-Stickelberger \cite{fs_1882}.
\subsection{The primary heat equation}
\label{S2.4}
In this Section, we review the general heat equations satisfied by
the sigma functions. 
\par
If we want to find second-order linear partial differential
equations (heat equations) satisfied by the sigma function, we should
proceed in as general a way as possible.  Here, note that the equation
(\ref{heq_jacobi_theta}) is satisfied not only by the Jacobi theta
function (\ref{jacobi_theta}) but also by each individual term 
of the sum in (\ref{jacobi_theta}).  
This corresponds a statement in the
proof of Theorem 13 in page 274 of \cite{bl_2008}.  Here we will
review their theory of such equations, correcting a few minor errors,
and apply it explicitly to more general curves than considered in
\cite{bl_2008}.  We shall start from this point of view.
\par
Take an derivative \(L\in\vec{L}\) and assume it is of homogeneous weight, 
say \(\wt(L)=k\) (see (\ref{wt_on_d_dmu})), 
we have the symmetric matrix 
\(-\Gamma^L J=
\bigg[\,
\begin{matrix}
   \alpha & \beta \\
\tp{\beta}& \gamma
\end{matrix}\, \bigg]\)
with 
\begin{equation*}
\alpha=\underset{\substack{i\in\wt(u)\\ j\in\wt(u)}}{[\alpha_{-i,k-j}]}, \ \ 
\beta=\underset{\substack{i\in\wt(u)\\ j\in\wt(u)}}{[\beta_{j,k-i}]}, \ \ 
\gamma=\underset{\substack{i\in\wt(u)\\ j\in\wt(u)}}{[\gamma_{-i,k+j}]}
\end{equation*}
and we also associate with it 
a second-order differential operator \(H^L\), 
given by
\begin{equation}\label{H_attached_to_L}
  \begin{aligned}
H^L&=\frac12[\,\tp{\partial_u}\ \ \tp{u}\,]
    \bigg[
      \begin{array}{rc}
        \alpha     & \ \beta \\
        \tp{\beta} & \ \gamma
      \end{array}
      \bigg]
    \bigg[
      \begin{array}{c}
        \partial_u \\ u
      \end{array}
     \bigg]\\
  &=\sum_{i\in\wt(u)}\sum_{j\in\wt(u)}
  \Big(\tfrac12\,{\alpha}_{-i,\,k-j}\ \frac{\partial^2}{\partial u_i\partial u_j}
  +{\beta}_{-j,\,k-i}\ u_i\frac{\partial}{\partial u_j}
  +\tfrac12\,{\gamma}_{-i,\,k+j}\ u_iu_j\Big)
  +\tfrac12\,{\mathrm{Tr}\,\beta}. 
\end{aligned}
\end{equation}
Here \(\partial_u\) denotes the column vector with \(g\) 
components \(\frac{\partial}{\partial u_i}\), 
and \(u\) the column vector with \(g\) components \(u_i\)s. 
The very last term comes from the commutation relation 
\(\tfrac{\partial}{\partial u_i}u_j=u_j\tfrac{\partial}{\partial u_i}+\delta_{ij}\). 
It is then straightforward to verify the following
\begin{lemma}
If we define a function \(G_0(u,\varOmega)\) {\rm(}this is a Green's function{\rm)} by
\begin{equation*}
G_0(u,\varOmega)
=\Big(\frac{(2\pi)^g}{\mathrm{det}\,\omega'}\Big)^{\frac12}
\exp\big(-\tfrac12\,\tp{u}\,\eta'{\omega'}^{-1}u\big),
\end{equation*}
then the following equation {\rm(}a heat equation{\rm )} holds\,{\rm:}
\begin{equation}\label{fundamentalsoln}
(L - H^L) G_0 =0.
\end{equation}
\end{lemma}
\begin{proof}
We have
{\footnotesize
\begin{equation}\label{eq5.11}
  \begin{aligned}
    L\,G_0 &=-\frac12 \mathrm{Tr}(\omega'^{-1}
    L(\omega'))G_0 -\frac12
    (\tp{u}\,L(\eta'){\omega'}^{-1}u)G_0
    +\frac12 (\tp{u}\,\eta'{\omega'}^{-1}L(\omega'){\omega'}^{-1}u)G_0\\
    &=\frac12\left[\mathrm{Tr}(\omega'^{-1}(\beta \omega' - \alpha
      \eta')) +(\tp{u}\,((\gamma \omega' - \tp{\beta}
      \eta')\omega'^{-1} -\eta'{\omega'}^{-1}(\beta \omega' - \alpha
      \eta')\omega'^{-1})
      \,u)\right]G_0\\
    &=\frac12\left[\mathrm{Tr}(\beta-\alpha \eta' \omega'^{-1}) +
      (\tp{u}\,(\gamma-\tp{\beta}\eta'{\omega'}^{-1}
      -\eta'{\omega'}^{-1}\beta
      +\eta'{\omega'}^{-1}\alpha\eta'{\omega'}^{-1})\,u) \right]G_0
\end{aligned}
\end{equation}
}
by using
\begin{equation}\label{ell_and_Gamma}
L(\varOmega)=\Gamma^L \varOmega
=\bigg[
\begin{array}{rr}
-\beta  & \ \alpha   \\ 
-\gamma & \ \tp{\beta}
\end{array}
\bigg]
\bigg[ 
\begin{array}{rr}
\omega' & \ \omega'' \\
\eta'   & \ \eta''
\end{array}
\bigg]
=\bigg[
\begin{array}{cc}
-\beta \omega'+\alpha \eta' & -\beta\omega''+\alpha \eta''  \\
-\gamma \omega' +\tp{\beta}\eta'  & -\gamma \omega'' +\tp{\beta} \eta''
\end{array}
\bigg].
\end{equation}
The result of (\ref{eq5.11}) coincides with that of
{\small 
\begin{equation*}
\begin{aligned}
H^{L_0}G_0
&= \frac12 \big[\tp{\partial_u}\ \ \tp{u}\big]
\bigg[\,
\begin{matrix}
   \alpha  & \ \beta \\ 
\tp{\beta} & \ \gamma
\end{matrix}\,
\bigg]
\bigg[
\begin{matrix}
\partial_u \\ u 
\end{matrix}
\bigg]G_0
=\frac12 \big[\tp{\partial_u} \ \  \tp{u}\big]
\bigg[\,
\begin{matrix}
   \alpha  & \ \beta \\
\tp{\beta} & \ \gamma
\end{matrix}\,
\bigg]
\bigg[\,
\begin{matrix}
 -\eta'\omega'^{-1}u  \\ u
\end{matrix}\,
\bigg]G_0 \\
&=\frac12 \mathrm{Tr}(\tp{\beta}
-\alpha \eta' \omega'^{-1})\,G_0
+\frac12
\big[{-}\tp{u}\eta' \omega'^{-1} \ \  u \,\big]
\bigg[
\begin{array}{rr}
    \alpha & \ \beta \\ 
\tp{\beta} & \ \gamma
\end{array}
\bigg]
\bigg[\,
\begin{matrix}
-\eta'\omega'^{-1}u \\ u
\end{matrix}\,
\bigg]G_0.
\end{aligned}
\end{equation*}
}
Here we have used the generalized Legendre relation (\ref{legendre}) and the
symmetry of \,\(\eta' \omega'^{-1}\).
\end{proof}
Now we recall that the different terms in the expansion of the theta
function are periodic translates of one another.  Analogously, to
construct the different terms appearing in the expansion of the sigma
function, we act on \(G_0\) by iterating an element of the Heisenberg
group.  
For a variable
\(z\in\mathbb{Q}(\{{\omega'}_{ij}\},\,\{{\omega''}_{ij}\})\), and a
two \(g\)-component column vectors $p$, $q$ whose components also
belong to \(\mathbb{Q}(\{{\omega'}_{ij}\},\,\{{\omega''}_{ij}\})\), we
introduce 
\begin{equation*}
F(z,p,q)=\exp(z)\exp(\tp{p}u)\exp(\tp{q}\partial_u).
\end{equation*}
We write its inverse operator as
\begin{equation*}
F^{-1}(z,p,q)=\exp(-\tp{q}\partial_u)\exp(-\tp{p}u)\exp(-z).
\end{equation*}
\begin{lemma}
Defining \(F(z,p,q)\) for any  \(p\), \(q\) and \(z\),  
and \(H^L\) for \(L\in\vec{L}\) as in {\rm (\ref{H_attached_to_L})}, 
the operator equality 
\begin{equation}\label{conjugacy}
  F^{-1}(z,p,q)(L - H^L)F(z,p,q)=L - H^L
\end{equation}
holds if and only if 
\begin{equation}\label{L_p_q_z}
L
\bigg[\,
\begin{matrix}
q \\ p
\end{matrix}\,
\bigg]
=\Gamma^L
\bigg[\,
\begin{matrix}
q \\ p
\end{matrix}\,
\bigg] \ \ \mbox{and} \ \ 
L(z)=\frac12\,(\,\tp{p}\alpha p - \tp{q}\gamma q\,). 
\end{equation}
Here \,\(\Gamma^L\) is that of \,{\rm(\ref{GammaL})}. 
\end{lemma}
\begin{proof}
We calculate directly that
{\small 
\begin{equation*}
\begin{aligned}
F^{-1}(z,p,q)\,LF(z,p,q)
&=F^{-1}(z,p,q)
\big(\,L(z) +L(\tp{p})\,u+L(\tp{q})
\,\partial_u\,\big)
F(z,p,q)+L \\
&=L+
\big(\,
L(z) +L(\tp{p})(u-q)+L(\tp{q})\,\partial_u\,\big).
\end{aligned}
\end{equation*}
}
Similarly, we find
\begin{equation*}
\begin{aligned}
F^{-1}(z,p,q)&\,H^LF(z,p,q)
=\frac12
\big[\tp{\partial_u}+\tp{p} \ \ \tp{u}-\tp{q}\big]
\bigg[\,
\begin{matrix}
    \alpha & \beta \\
 \tp{\beta}& \gamma 
\end{matrix}\,
\bigg]
\bigg[\,
\begin{matrix}
\partial_u+p \\ u-q
\end{matrix}\,
\bigg] \\
&=H^L
+\big[\tp{p}\ \ -\tp{q}\big]
\bigg[\,
\begin{matrix}
    \alpha  & \beta \\
\tp{\beta}  & \gamma
\end{matrix}\,
\bigg]
\bigg[\,
\begin{matrix}
\partial_u \\  u
\end{matrix}\,
\bigg]
+\frac12\big[\tp{p}\ \ -\tp{q}\big]
\bigg[\,
\begin{matrix}
    \alpha  & \beta \\
 \tp{\beta} &\gamma
\end{matrix}\,
\bigg]
\bigg[
\begin{matrix}
p \\ -q
\end{matrix}
\bigg].
\end{aligned}
\end{equation*}
Matching coefficients of \(\partial_u\), $u$, and $1$, and transposing and
rearranging, we see that, respectively,
\begin{equation*}
L(q)=[\,-\beta \ \ \alpha\,]
\bigg[\,
\begin{matrix}
q \\ p
\end{matrix}\,
\bigg], \ \ \ 
L(p)=[\,-\gamma \ \ \tp{\beta}\,]
\bigg[\,\begin{matrix}
q \\ p
\end{matrix}\,\bigg], \ \ \ 
L(z)=\frac12\,(\,\tp{p}\alpha p - \tp{q}\gamma q\,)
\end{equation*}
as desired. 
\end{proof}
\begin{corollary}
The formula {\rm (\ref{L_p_q_z})} holds if
\begin{equation}\label{ellGamma_p_q_z}
\bigg[\,
\begin{matrix}
q \\ p
\end{matrix}\,
\bigg]
=\varOmega
\bigg[
\begin{matrix}
b' \\ b''
\end{matrix}
\bigg], \ \ 
z = z_0 + \tfrac12\tp{p}q, 
\end{equation}
where $b'$ and $b''$ are arbitrary numerical constant vectors, 
and $z_0$ is an irrelevant numerical constant, which we set to zero below. 
\end{corollary}
\begin{proof}
We have
\(L\Big(
\bigg[
\begin{matrix}
q \\ p    
\end{matrix}
\bigg]\Big)
=L(\Omega)
\bigg[
\begin{matrix}
b' \\ b''    
\end{matrix}
\bigg]
=\Gamma^L\Omega
\bigg[
\begin{matrix}
b' \\ b''    
\end{matrix}
\bigg]
=\Gamma^L
\bigg[
\begin{matrix}
p \\ q    
\end{matrix}
\bigg]
\), which is the former relation in (\ref{L_p_q_z}). 
The latter one is checked easily. 
\end{proof}
Denoting the constant vector \(\tp{[b'\ b'']}\) simply by \(b\), 
we denote \(p\), \(q\), and \(z\) with \(z_0=0\) in (\ref{ellGamma_p_q_z}) 
by \(p(b)\), \(q(b)\), and \(z(b)\), respectively. We define
\begin{equation*}
G(b,u,\varOmega)=F\big(z(b),p(b),q(b)\big)\,G_0(u,\varOmega). 
\end{equation*}
Using the Legendre relation (\ref{legendreinv}), we note
that \(\tp{p}\omega'-\tp{q}\eta'= 2\pi i \tp{b''}\), 
and hence we obtain
\begin{equation*}
\begin{aligned}
G&(b,u,\varOmega)
 = F\big(z(b),p(b),q(b)\big)\,G_0(u,\varOmega)\\
&=\exp(\tfrac12 \tp{p}q)\,\exp(\tp{p}u)\,\exp(-\tp{q}\eta'\omega'^{-1}u)
  \exp(-\tfrac12 \tp{q}\eta' \omega'^{-1}q)\,G_0\\
&= \exp(-\tfrac12 (2\pi i \tp{b''} \omega'^{-1}q) \exp(2\pi i \tp{b}'' 
\omega'^{-1}u)\,G_0\\
&= \exp(-\tfrac12 (2\pi i \tp{b''} \omega'^{-1}(\omega' b' +\omega''b'') 
\exp(2\pi i \tp{b}'' \omega'^{-1}u)\,G_0 \\
&=\bigg(\frac{(2\pi)^g}{\mathrm{det}(\omega')}\bigg)^{\frac12}\,
\exp\big(-\tfrac12 \tp{u}\eta' \omega'^{-1} u\big)
\,\exp\Big(\,2\pi i\big(\,\tfrac12 \tp{b''} \omega'^{-1} \omega'' b''+ \tp{b''}(\omega'^{-1}u+b')\,\big)\Big).
\end{aligned}
\end{equation*}
Now, the following theorem, which is the foundation of BL theory, 
is obvious from (\ref{fundamentalsoln}) and (\ref{conjugacy}). 
\begin{theorem}
\label{primary_heq}
 {\rm (The primary heat equation)} \ 
For the function \(G(b,u,\varOmega)\) above, one has
\begin{equation}\label{modified_heq}
(L-H^L)\,G(b,u,\varOmega)=0
\end{equation}
for any \(L\in\vec{L}\). 
\end{theorem}
\subsection{The algebraic heat operators}
\label{alg_heat_eq} 
For the coordinates of the space in which the sigma function is defined, 
we do not use \((u_1,\cdots,u_g)\) for subscripts of the variable \(u\), 
but denote instead, as in (\ref{u_coordinates})
\begin{equation*}
u=(u_{w_g},\cdots,u_{w_1}).   
\end{equation*}
That is, the components of \(u\) are labelled by their weights, 
which are the Weierstrass gaps.  
\(\Gamma^L\in\mathfrak{sp} (2g,\mathbb{Q}[\mu])\).  
As a corollary to Theorem \ref{primary_heq}, we have
\begin{corollary}
\label{heat_op_pre_classical_sigma} \ 
Let 
\begin{equation}
\rho(u)=\sum_{b}c_b\,G(b,u,\varOmega),
\end{equation}
where \(b\) runs through the elements of any set 
\(\subset\) \(\mathbb{C}^{2g}\) and \(c_b\in\mathbb{C}\) are constants 
such that the sum converges absolutely. 
Then we have, for any \(L\in\vec{L}\), that
\begin{equation*}
(L-H^L)\,\rho(u,\varOmega)=0.
\end{equation*}
\end{corollary}
\begin{proof}
Since both of \(L\) and \(H^L\) are independent of \(b\),
each term of \(\rho(u)\) satisfies (\ref{modified_heq}). 
\end{proof}
\begin{remark}\label{heq_tilde_sigma} 
{\rm 
Assume that \(\mathscr{C}_{{\mu}}\) is non-singular.  
Let \(\varOmega\) be the usual period matrix defined by 
{\rm (\ref{per_matrix})}, and \(\delta\) be its Riemann constant.
The function defined at {\rm (\ref{classical_tilde_sigma})} is
written as
\begin{equation}\label{pre_classical_sigma}
\tilde{\sigma}(u)=\tilde{\sigma}(u,\varOmega)
=\sum_{n\in\mathbb{Z}^g}G( \tp{[ \delta'\ n+\delta'']},u,\varOmega).
\end{equation}
Since the imaginary part of \(\omega'^{-1}\omega''\) is positive
definite, this series converges absolutely.  This is a special case of
\(\rho(u)\) of {\rm(\ref{heat_op_pre_classical_sigma})}.  
}
\end{remark}
Because both \(L\) and \(H^L\)
are independent of \(b\), there are infinitely many linearly 
independent entire functions \(\rho(u)\) on \(\mathbb{C}^g\)
satisfying \((L-H^L)\rho(u)=0\).  
Moreover, since, for a fixed \(b\), 
the function \(G(b,u,\varOmega)\) is independent of \(L\), 
we see that, by switching the choice of \(L\),
there are infinitely many linearly independent operators of 
the form \((L-H^L)\) which satisfy \((L-H^L)\,\rho(b,u,\varOmega)=0\) 
for some fixed \(\rho(u)\). 
\par
However, because our aim is to find a method to calculate
the power series expansion of the sigma function, 
we need a more detailed discussion.  
For our purpose, we require
\begin{Aitem}
\item 
to show that for any (or any element of a good generators) \(L\in\vec{L}\) 
find a good mapping to associate \(L\) to 
a (quadratic linear differential) operator 
which annihilate the sigma function (\ref{classical_sigma}) 
as an element in the ring \(\mathbb{Q}[{\mu}][[u]]\),
and 
\item 
to show that the sigma function and its absolute constant multiples 
are exactly the functions in \(\mathbb{Q}[{\mu}][[u]]\) 
killed by the operators obtained in (A1) 
\end{Aitem}
Since \(L\) is a derivation with respect to the \(\mu_j\)s 
but \(H^L\) is a differential operator with respect to the \(u_{w_j}\)s, 
we see that, for some function \(\varXi\)
depending only on the \(\mu_j\)s,
\begin{equation*}
L\,\varXi\tilde{\sigma}(u)
=(L\varXi)\,\tilde{\sigma}(u)+\varXi(L\tilde{\sigma}(u)),\ \ \ \ 
H^L\,\varXi\tilde{\sigma}(u)
=\varXi(H^L\tilde{\sigma}(u)).
\end{equation*}
Therefore, \(\varXi\tilde{\sigma}(u)\) satisfies
\begin{equation}\label{pre_alg_heq}
(L-H^L)(\varXi\tilde{\sigma}(u))
=\tfrac{L\varXi}{\varXi}\,\varXi\,\tilde{\sigma}(u)
=(L\log\varXi)\varXi\,\tilde{\sigma}(u).
\end{equation}
If \(\varXi\,\tilde{\sigma}(u)\) is the correct sigma function, 
the left hand side of the above is in \(\mathbb{Q}[\mu][[u]]\), 
So, in order to get the function \(\hat{\sigma}(u)\) 
of (\ref{classical_sigma}), 
\(L\log\varDelta\in\mathbb{Q}[\mu]\). 
We shall show that this indeed holds for any \(L\in\vec{L}\) 
in the next section. 
\newpage
\subsection{The matrix \texorpdfstring{\(V\)}{Lg}}
\label{Section Matrix_V}
Throughout this section, we suppose that all the \(\mu_j\)'s, \(x\),
and \(y\) are variables or indeterminates.  
In view of the approach in \cite{bl_2008},
we require that \(L\log\varDelta\) belongs to \(\mathbb{Q}[{\mu}]\)
because of conditions (A1), (A2) and equation (\ref{pre_alg_heq}).
We also explain suitable choices for \(\Gamma\) for which \(L\)
satisfies the condition.  
\par
Now we explain another method, known in singularity theory, 
to calculate the {\it discriminant} \ \(\varDelta\) \ 
and a basis of the space of the vector fields tangent to 
the variety defined by \(\varDelta=0\). 
\par
Let \((x,y)\) and \((z,w)\) are different generic points of \(\mathscr{C}\). 
Following p.112 of \cite{bl_2005}, we define a function 
\(\ph\big((x,y),(z,w)\big)\) (pre-hessian), which is defined by
\begin{equation*}
\ph\big((x,y),(z,w)\big)
=\frac12\left|\,
\begin{matrix}
\dfrac{f_1(x,y)-f_1(z,w)}{x-z} & \dfrac{f_2(x,y)-f_2(z,w)}{x-z} \\[8pt]
\dfrac{f_1(z,y)-f_1(x,w)}{y-w} & \dfrac{f_2(z,y)-f_2(x,w)}{y-w}
\end{matrix}\,
\right|\in\mathbb{Z}[\mu,\,x,\,y,\,z,\,w].
\end{equation*}
For any \(F\in\mathbb{Q}[\mu][x,y]\), we define
\begin{equation}\label{hessian}
\mathrm{Hess}\,F=\Bigg|\,
 \begin{matrix}
   \frac{\partial^2}{\partial x^2}F         &
   \frac{\partial^2}{\partial x\partial y}F \\[4pt]
   \frac{\partial^2}{\partial y\partial x}F &
   \frac{\partial^2}{\partial y^2}F 
 \end{matrix}\, 
 \Bigg|.
\end{equation}
\begin{lemma}\label{properties_H}
{\rm (Buchstaber-Leykin \cite{bl_2005}, p.64)} \ 
Let \,\(I\)\, be the ideal in \(\mathbb{Q}[{\mu}][x,y,z,w]\) 
generated by \(f_1(x,y)\), \(f_2(x,y)\), \(f_1(z,w)\), and \(f_2(z,w)\). 
The determinant \,\(\ph\big((z,w),(x,y)\big)\) 
has the following properties. \\
{\rm(1)} \(\ph\big((x,y),(x,y)\big)=\mathrm{Hess}\,f(x,y)\). \\
{\rm(2)} \(\ph\big((x,y),(z,w)\big)=\ph\big((z,w),(x,y)\big)\). \\
{\rm(3)} We have 
\begin{equation*}
\ph\big((z,w),(x,y)\big)\,F\big((x,y),(z,w)\big)
\equiv \ph\big((x,y),(z,w)\big)\,F\big((z,w),(x,y)\big)\,\bmod{I}
\end{equation*}
for any \(F\big((x,y),(z,w)\big)\in\mathbb{Q}[{\mu}][x,y,z,w]\).
\end{lemma}
\begin{proof}
(1) Taking the limit \(z\to x\) after subtracting the second row times
\(\frac{y-w}{x-z}\) from the first row in 
\(\ph\big((z,w),(x,y)\big)\,F\big((x,y),(z,w)\big)\), 
we have
\begin{equation*}
  \frac12
  \left|
\begin{array}{cc}
  f_{11}(x,y)+f_{11}(x,w) & f_{21}(x,y)+f_{21}(x,w) \\[3pt]
  \dfrac{f_1(x,y)-f_1(x,w)}{y-w} & \dfrac{f_2(x,y)-f_2(x,w)}{y-w} 
\end{array}
  \right|,
\end{equation*}
where \(f_{11}(x,y)=\frac{\partial^2}{\partial x^2}(x,y)\), etc. 
Then, by taking limit \(y\to w\) we get \(\mathrm{Hess}\,f(x,y)\). \\
(2) is trivial. \\
(3) By expanding the matrix, we see that the numerator
\begin{equation}\label{symmetry_of_H-det}
\begin{aligned}
&\ \ \ \ \big(
  f_1(x,y)f_2(z,y)
 -f_1(x,y)f_2(x,w)
 -f_1(z,w)f_2(z,y)
 +f_1(z,w)f_2(x,w)
 \big)\\
&\hskip 5pt
-\big(
  f_1(z,y)f_2(x,y)
 -f_1(z,y)f_2(z,w)
 -f_1(x,w)f_2(x,y)
 +f_1(x,w)f_2(z,w)
 \big)\\
&=\big(f_1(x,y)f_2(z,y)-f_1(z,y)f_2(x,y)\big)
 -\big(f_1(x,y)f_2(x,w)-f_1(z,y)f_2(z,w)\big)\\ 
&\hskip 5pt
 -\big(f_1(z,w)f_2(z,y)-f_1(x,w)f_2(x,y)\big)
 +\big(f_1(z,w)f_2(x,w)-f_1(x,w)f_2(z,w)\big)\\
&=\big(f_1(x,y)f_2(z,y)-f_1(x,w)f_2(z,w)\big)
 -\big(f_1(x,y)f_2(x,w)-f_1(x,w)f_2(x,y)\big)\\ 
&\hskip 5pt 
 -\big(f_1(z,w)f_2(z,y)-f_1(z,y)f_2(z,w)\big)
 +\big(f_1(z,w)f_2(x,w)-f_1(z,y)f_2(x,y)\big)
\end{aligned}
\end{equation}
is divisible by \((z-x)(w-y)\),  because 
the second expression is clearly divisible by \((z-x)\), 
while the third expression is divisible by \((w-y)\). 
Hence, \(\ph\big((x,y),(x,y)\big)\in\mathbb{Q}[\mu][x,y]\). 
Moreover, the second expansion is equal to
\begin{align*}
   &\big(f_1(x,y)f_2(z,y)-f_1(x,y)f_2(x,y)+f_1(x,y)f_2(x,y)-f_1(z,y)f_2(x,y)\big)\\
-\,&\big(f_1(x,y)f_2(x,w)-f_1(x,y)f_2(z,w)+f_1(x,y)f_2(z,w)-f_1(z,y)f_2(z,w)\big)\\ 
-\,&\big(f_1(z,w)f_2(z,y)-f_1(z,w)f_2(x,y)+f_1(z,w)f_2(x,y)-f_1(x,w)f_2(x,y)\big)\\
+\,&\big(f_1(z,w)f_2(x,w)-f_1(z,w)f_2(z,w)+f_1(z,w)f_2(z,w)-f_1(x,w)f_2(z,w)\big)\\
=\,&f_1(x,y)\big(f_2(z,y)-f_2(x,y)\big)+\big(f_1(x,y)-f_1(z,y)\big)f_2(x,y)\\
-\,&f_1(x,y)\big(f_2(x,w)-f_2(z,w)\big)-\big(f_1(x,y)-f_1(z,y)\big)f_2(z,w)\\ 
-\,&f_1(z,w)\big(f_2(z,y)-f_2(x,y)\big)-\big(f_1(z,w)-f_1(x,w)\big)f_2(x,y)\\
+\,&f_1(z,w)\big(f_2(x,w)-f_2(z,w)\big)+\big(f_1(z,w)-f_1(x,w)\big)f_2(z,w),
\end{align*}
which implies that \,\(\ph\big((x,y),(z,w)\big)(w-y)\)\, already belongs to \(I\). 
A similar calculation shows that \,\(\ph\big((x,y),(z,w)\big)(z-x)\in{I}\). 
For \(F\big((z,w),(x,y)\big)=x^ay^b\), by using (2), we see
\begin{equation*}
  \begin{aligned}
    \ph&\big((z,w),(x,y)\big)\,x^ay^b
    -\ph\big((x,y),(z,w)\big)\,z^aw^b\\
   & =\ph\big((z,w),(x,y)\big)\,(x^ay^b-z^aw^b)\\
   & =\ph\big((z,w),(x,y)\big)\,(x^ay^b-x^aw^b+x^aw^b-z^aw^b)\\
   & =\ph\big((z,w),(x,y)\big)\,\big(x^a(y^b-w^b)+(x^a-z^a)w^b\big)
    \in I.
  \end{aligned}
\end{equation*}
Hence, (3) has been proved. 
\end{proof}
\vskip 5pt
Below, we will use, instead of \(T\), the {\em symmetric}
\(2g\times 2g\) matrix 
\begin{equation*}
  V=\underset{\substack{a\in\wt(M)\\b\in\wt(\mu)}}{[\,V_{a,b}\,]}
  =\underset{\substack{1\leq i\leq 2g\\1\leq j\leq 2g}}
  {[\,V_{v_i,\,eq-v_{\kern0.5pt 2g-j+1}}\,]}
  \in\mathrm{Sym}(2g,\,\mathbb{Q}[\mu])
\end{equation*}
where all the indices run in increasing order, 
defined by the equation
\begin{equation}\label{V}
  \tp{(\rev{M}}(x,y))\,\cdot{V}\cdot\,\rev{M}(z,w)=f(x,y)\,
  \ph\big((z,w),(x,y)\big)\,F\big((x,y),(z,w)\big)
\end{equation}
in the ring
\begin{equation*}
\mathbb{Q}[\mu][x,y,z,w]\,\big/\,I. 
\end{equation*}
We note that the weight of any entry of \(V\) is given by
\begin{equation}\label{V_1j}
\wt(V_{a,\,b})=-(a+b). 
\end{equation}
We define
\begin{equation*}
H=\underset{\substack{a\in\wt(M)\\b\in\wt(\mu)}}{[\,H_{a-q+2,\,b-q-2}\,]}
 =\underset{\substack{1\leq i\leq 2g\\1\leq j\leq 2g}}{[\,H_{v_i-q+2,\,eq-v_{\kern0.5pt 2g-j+1}-q-2}\,]}
\end{equation*}
as the matrix given by 
\begin{equation*}
\ph\big((x,y),(z,w)\big)=\tp{(\rev{M}}(x,y))\,H\,\rev{M}(z,w).
\end{equation*}
We see \(H\in\mathrm{Mat}(2g,\,\mathbb{Q}[\mu])\) by consideration of 
\ref{symmetry_of_H-det} in the proof of \ref{properties_H}. 
Moreover, as for \(V\), the weight of any entry of \(H\) is given by
\begin{equation*}
\wt(H_{a-q+2,\,b-q-2})=-(a+b-2q). 
\end{equation*}
\begin{example}
{\rm In the case \((e,q)=(2,2g+1)\), 
for \(a=v_i=2i-2\) and \(b=eq-v_{2g-j+1}=4g+2-2(2g-j+1)+2=2j+2\) 
with \(1\leq i\leq g\) and \(1\leq j\leq g\), 
we see that
\begin{equation*}
\begin{aligned}
\wt(V_{a,b})&=\wt(V_{2i-2,2j+2})=-2(i+j), \\
\wt(H_{a-q+2,\,b-q-2})&
\,(=a+b-2q)=\wt(H_{v_i-(2g+1)+2,\,v_{2g-j+1}-(2g+1)-2})
=2(i-j)-(2g+1)
\end{aligned}
\end{equation*}
} 
\end{example}
\begin{lemma}
The matrix \(H\) is of the form 
\begin{equation}\label{H_ij}
H=
\left[\,
  \begin{matrix}
       &         &       & -eq \\[-3pt]
       &         &   -eq  & *   \\[-3pt]
       & \iddots & \vdots & \vdots \\[-3pt]
  -eq  & \cdots  &   *    & *
  \end{matrix}\,
    \right]. 
\end{equation}
If \(e=2\), then we have explicitly 
\begin{equation}\label{hyperell_H_ij}
H=
\mbox{\footnotesize\(\left[\,
\begin{matrix}
      &      &              &         &               &               &  -2q      \\
      &      &              &         &               &  -2q          &  0         \\
      &      &              &         &     -2q       &   0           & -2(q-2)\mu_4   \\
      &      &              & \iddots &  \vdots       & \vdots        &   \vdots       \\ 
      &      & -2q          & \cdots  & -12\mu_{2(q-6)} & -10\mu_{2(q-5)} & -8\mu_{2(q-4)}   \\
      & -2q  &  0           & \cdots  & -10\mu_{2(q-5)} & -8\mu_{2(q-4)}  & -6\mu_{2(q-3)}   \\
 -2q  &  0   & -2(q-2)\mu_4 & \cdots  & -8\mu_{2(q-4)}  & -6\mu_{2(q-3)}  & -4\mu_{2(q-2)}
\end{matrix}\,
\right]\)}.  
\end{equation}
\end{lemma}
\begin{proof}
Setting all the \(\mu_j\) to be \(0\), we have
\begin{equation*}
    \frac12\,\left|\,
    \begin{matrix}
      \frac{q(-x^{q-1}+z^{q-1})}{x-z} & \frac{e(y^{e-1}-w^{e-1})}{x-z} \\[10pt]
      \frac{q(x^{q-1}-z^{q-1})}{y-w} & \frac{e(y^{e-1}-w^{e-1})}{y-w} 
    \end{matrix}\,
    \right|
=-eq(x^{q-2}+x^{q-3}z+\cdots+z^{q-2})(y^{e-2}+y^{e-3}w+\cdots+w^{e-2}).
\end{equation*}
It follows that the counter-diagonal entries of \(H\) are \(-eq\). 
From the definitions, the weight of \(H\) is \(-2(eq-e-q)\) and 
\(\mathrm{wt}(M_{2g}(x,y))=-2g-1+w_g=-2(eq-q-e)\). 
Therefore the entries below the counter-diagonal must be \(0\). 
For the case \(e=2\), we have
\begin{equation*}
  \begin{aligned}
    \frac12\,&
    \left|\ 
    \begin{matrix}
      -\frac{q(x^{q-1}-z^{q-1})+(q-2)\mu_4(x^{q-3}-z^{q-3})+\cdots+\mu_{eq-e}(x-z)}{x-z} &
      \frac{2y-2w}{x-z} \\[3pt]
  \ \ \frac{q(x^{q-1}-z^{q-1})+(q-2)\mu_4(x^{q-3}-z^{q-3})+\cdots+\mu_{eq-e}(x-z)}{y-w} &
      \frac{2y-2w}{y-w}
    \end{matrix}\ 
    \right|\\
 &=-2\,\frac{q(x^{q-1}-z^{q-1})+(q-2)\mu_4(x^{q-3}-z^{q-3})+\cdots+\mu_{eq-e}(x-z)}{x-z}\\
 &=-2\,\big(q(x^{q-2}+x^{q-3}z+\cdots+z^{q-2}) 
    +(q-2)\mu_4(x^{q-4}+x^{q-5}z+\cdots+z^{q-4})+\cdots+\mu_{eq-e}\big),
  \end{aligned}
\end{equation*}
giving the desired form of \(H\). 
\end{proof}
\begin{lemma}\label{det_V_and_det_T}
We have \ \(\det(V)=(-1)^g\det(T)\).
\end{lemma}
\begin{proof}
Since
\begin{equation*}
\begin{aligned}
 f(x,y)\,\mathrm{ph}\big((x,y),(z,w)\big)
&=f(x,y)\,\tp{(\rev{M}(z,w))}\,H\cdot\,\rev{M}(x,y)\\
&\equiv\tp{(\rev{M}}(z,w))\,\rev{\!(-\tfrac{1}{eq}\,T)}\,\,H\cdot\,\rev{M}(x,y)\bmod{I}
\end{aligned}
\end{equation*}
by (\ref{matrix_T}), and the entries in \(M(x,y)\) form a basis of \(\mathbb{Q}[\mu][x,y]/(f_1(x,y),f_2(x,y))\), 
we see
\begin{equation}\label{V_and_T}
V=-\tfrac1{\,eq\,}\,\rev{T}\,H. \ \ 
\end{equation}
Since \(H\) is a skew-upper-triangular
matrix of the form (\ref{H_ij}), we have proved 
\,\(\det(V)=(-1)^g\det(T)\)\, as desired. 
\end{proof}
\begin{lemma}\label{weight_detV_detT}
We have \,\(\wt(\det V)=\wt(\det T)=-eq(e-1)(q-1)\). 
\end{lemma}
\begin{proof}
Since the determinants of \(T\) and \(V\) are of homogeneous weight, 
it suffices to check the sum of weights of the counter-diagonal entries, which is given by
\begin{equation*}
\sum_{q\in\wt(M)}\wt(T_{a,\,eq-a})
=\sum_{q\in\wt(M)}(-eq)=-2geq=-eq(e-1)(q-1).
\end{equation*}
This is the same as \(\wt(\det(V))\) and \(\wt(\det(T))\). 
\end{proof}
\begin{lemma}\label{lemma_discri}
We have \,\(\wt(\varDelta)=-eq(e-1)(q-1)\). 
If \,\(\gcd(e-1,q-1)=1\), then we have
\begin{equation*}
\det T=(-1)^g\det V=c\cdot\varDelta
\end{equation*}%
with \,\(c\in\mathbb{Q}^{\times}\). 
\end{lemma}%
\begin{proof}%
Letting all the coefficients \(\mu_j\) of
\(p_j(x)\) for \(1\leq j\leq e-1\) to be zero, the discriminant
\(\varDelta\) becomes a power of the square of 
the difference of all the roots of \(p_e(x)=0\).  
Since the weight of any root is \(-e\), 
the weight of the square of the difference is
\begin{equation*}
2\cdot\binom{q}{2}\cdot(-e)=-eq(q-1). 
\end{equation*}
Similar arguments on \(y\) shows that \(\wt\,\varDelta\) is
\(-qe(e-1)\) times an positive integer.  By the assumption
\(\gcd(e-1,q-1)=1\), we have \(\wt\,\varDelta\) is \(-eq(e-1)(q-1)\)
times a positive integer.  The statement follows from
\ref{weight_detV_detT} combined with \ref{jet} (1).
\end{proof}
\begin{remark}
The condition \(\gcd(e-1,q-1)=1\) in {\rm\ref{lemma_discri}} holds 
if \(e=2\) or \((e,q)=(3,4)\), \((3,5)\), 
for which we already know \(\varDelta\) explicitly as mentioned 
in {\rm \ref{wt_discriminant}}. 
\end{remark}
From now on we assume the modality of \(C\) to be \(0\). 
For any \(a\in\wt(M)\), the coefficient \(\mu_{eq-a}\) 
appears in the Weierstrass form, 
which is the reason why the modality is so important. 
Let
\begin{equation}\label{ell_operators}
  L_k = \sum_{j\in\wt(\mu)}\,V_{kj}\,
  \frac{\partial}{\partial \mu_j} \ \ \ (k\in\wt(M)). 
\end{equation}
It is natural to define
\begin{equation*}
\wt\Big(\frac{\partial}{\partial \mu_{j}}\Big)=j, \ \ \ 
\wt({L}_{v_i})=v_i.
\end{equation*}
We recall the definition of the matrix \(\Gamma^{L_k}\) 
for each \(L_k\) (\(\wt(L_k)=-k\)) (see \ref{gauss_manin_connection}), 
which we denote as\,\footnote{These should not be confused with
the symplectic basis of cycles \(\alpha_j\) and  \(\beta_j\)
in (\ref{per_matrix})}
\begin{equation*}
\Gamma^{L_k}
= \bigg[
      \begin{array}{rc}
       -\beta_k  &\ \ \alpha_k       \\
       -\gamma_k &\ \ \tp{\!\beta_k} 
      \end{array}
      \bigg],
\end{equation*}
where \(\alpha_k\), \(\gamma_k\in\mathrm{Sym}(g,\mathbb{Q}[\mu])\), and
\begin{equation}\label{alphabetagamma}
\alpha_k=\underset{\substack{i\in\wt(u)\\ j\in\wt(u)}}{[\alpha_{k;\,-i,k-j}]}, \ \ 
\beta_k=\underset{\substack{i\in\wt(u)\\ j\in\wt(u)}}{[\beta_{k;\,j,k-i}]}, \ \ 
\gamma_k=\underset{\substack{i\in\wt(u)\\ j\in\wt(u)}}{[\gamma_{k;\,-i,k+j}]}
\end{equation}
with decreasing order in \(\wt(u)\). 
Then the sum of two indices of any entry gives its weight;
\begin{equation*}
\wt(\alpha_{k;\,-i,k-j})=-i+(k-j), \ \ 
\wt( \beta_{k;\,j,k-i})=j+(k-i), \ \ 
\wt(\gamma_{k;\,-i,k+j})=-i+(k+j). 
\end{equation*}
Note once again that \,\(\Gamma^{L_k}\in\mathrm{Mat}(2g,\mathbb{Q}[\mu])\) 
follows from \ref{sato_formula}. 
Just to be sure, we shall redefine, as defined in (\ref{H_attached_to_L}), 
\begin{equation}\label{H_attached_to_L_j}
  \begin{aligned}
H^{L_k}
&=\strutf{40pt}{15pt}\smash{ 
  \frac1{\,2\,}\,\big[\,\tfrac{\partial\ }{\partial u_{w_g}} \ \cdots \ \tfrac{\partial\ }
  {\partial u_{w_1}} \ \ u_{w_g} \ \cdots \ u_{w_1}\,\big]
    \Bigg[
      \begin{array}{rc}
        \alpha_k     &\ \,\beta_k \\[5pt]
        \tp{\!\beta_k} &\ \,\gamma_k
      \end{array}
      \Bigg]\!
    \left[
      \begin{array}{c}
      \tfrac{\partial\ }{\partial u_{w_g}} \\
      \vdots \\
      \tfrac{\partial\ }{\partial u_{w_1}}\\
      u_{w_g} \\[-2pt]
      \vdots \\[-4pt]
      u_{w_1}
      \end{array}
     \right]
      }
\\
  & =\frac12\,\sum_{i\in\wt(u)}\sum_{j\in\wt(u)}
  \Big({\alpha}_{k;\,-i,\,k-j}\,\frac{\partial^2}{\partial\,u_i\partial\,u_j}\\
  &\ \ \ \ \ \ \ 
   +2{\beta}_{k;\,j,k-i}\,u_i\,\frac{\partial\ \ }{\partial u_j}
   +{\gamma}_{k;-i,k+j}\ u_iu_j\Big)
  +\tfrac12\,{\mathrm{Tr}\,\beta_k}. 
\end{aligned}
\end{equation}
\subsection{The operators tangent to the discriminant}
\label{operators_and_discriminant}
We prove the following proposition for which there is no proof in \cite{bl_2008}.
\begin{proposition}\label{hessian_formula_}
\,{\rm (Buchstaber-Leykin)}\, Let
\begin{equation*}
L={M(x,y)}\,\tp{[\,L_{v_1}\ L_{v_2}\ \cdots\ L_{v_{2g}}\,]}.
\end{equation*}
Then, in the ring \,\(\mathbb{Q}[\mu][x,y]/(f_1, f_2)\), we have
\begin{equation}\label{hessian_formula}
  L(\varDelta)=-\mathrm{Hess}f\cdot\varDelta. 
\end{equation}
\end{proposition}%
\begin{proof} 
For the case \(e=2\), \(f(x,y)\) is of the form \(y^2-p_2(x)\). 
We denote \(p_2(x)=p(x)\) for simplicity. 
Moreover, we denote \(p'(x)=\frac{\partial}{\partial x}p(x)\) 
and \(p''(x)=\frac{\partial^2}{\partial x^2}p(x)\). 
In this case, the ring \(\mathbb{Q}[\mu][x,y]/(f_1,f_2)\) is identified 
with \(\mathbb{Q}[\mu][x]/(p'(x))\) 
since \(f_1(x,y)=\frac{\partial}{\partial y}f(x,y)=2y\). 
Let \(F\) be a splitting field of \(p(x)\). 
We write the factorisation of \(p(x)\) in \(F\) as 
\(p(x) = (x-a_1)\cdots(x-a_q)\). 
Then \(\mu_{2i}\) is \((-1)^i\) times the fundamental symmetric 
function of \(a_1\), \(\cdots\), \(a_q\) of degree \(i\). 
Of course the ring \(\mathbb{Q}[\mu]\) is a sub-ring of 
\(\mathbb{Q}[a_1,\,\cdots,\,a_q]\). 
The Hessian of \(f(x,y)=y^2-p(x)\) is 
\begin{equation*}
\begin{aligned}
\bigg|\,
\begin{matrix}
-p''(x) & 0 \\
0       & 2 
\end{matrix}\,
\bigg|
&=-2\,p''(x).
\end{aligned}
\end{equation*}
The main idea is to consider \(\frac{\mathrm{Hess}\,f}{2f}
 =\frac{p''(x)}{p(x)}\) in the localised ring 
\begin{equation*}
\big(F[x]/(p'(x))\big)_{p(x)}
\end{equation*}
of \(F[x]/(p'(x))\) with respect to the multiplicative set
 \(\{1,\,p(x),\,p(x)^2,\,\cdots\}\) (see \cite{matsumura}, Section 4). 
The following calculation is done in the localised ring above. 
Since \(\gcd(p'(x),\,p(x))=1\) in this situation, 
we see \(\big(F[x]/(p'(x))\big)_{p(x)}=F[x]/(p'(x))\). 
Now, we have
{\small 
\begin{align*}
  \frac{p''(x)}{p(x)} &=\sum_{(i,j),i<j}\frac{2}{(x-a_i)(x-a_j)}
  =\sum_{(i,j),i<j}\frac{2}{a_i-a_j}\Big(\frac{1}{x-a_i}
  -\frac{1}{x-a_j}\Big)\\
  &=2\sum_{i=1}^q\Big(\sum_{j\neq i}\frac{1}{a_i-a_j}\Big)\frac{1}{x-a_i}
   =-2\sum_{i=1}^q\Big(\sum_{j\neq
    i}\frac{1}{a_i-a_j}\Big)\frac{1}{p'(a_i)}
  \frac{-p'(a_i)}{x-a_i}\\
  &=-2\sum_{i=1}^q\Big(\sum_{j\neq i}\frac{1}{a_i-a_j}\Big)\frac{1}{p'(a_i)}
  \frac{p'(x)-p'(a_i)}{x-a_i}
   = -2 \sum_{i=1}^q  c_i\,\frac{p'(x)-p'(a_i)}{x-a_i},
\end{align*}
}
where
\begin{equation*}
 c_i =\Big(\sum_{j \neq i}\frac{1}{a_i-a_j}\Big)\frac{1}{p'(a_i)}.
\end{equation*}
Since \(\mathrm{Hess}f(x,y)=-2p''(x)\) and \(p(x)=-f(x,y)\) in the 
localised ring, it suffices to show that
\begin{equation*}
\frac{\partial}{\partial\mu_{2i}}\log\varDelta=-2\sum_{j=1}^q c_j {a_j}^{q-i}
\end{equation*}
up to a non-zero constant multiple.  
Indeed, if we have the formula above, we have
\begin{equation*}
  \begin{aligned}
    \frac{L(x)\varDelta}{\varDelta}
    &=\sum_{i,k}M_i(x,y)V_{ik}\frac{\partial}{\partial\mu_{2i}}\log\varDelta
    =-2\sum_{i,k}M_i(x,y)V_{ik}\sum_{j=1}^qc_j{a_j}^{q-k}\\
    &=-2\sum_{j=1}^qc_j\sum_{i,k}M_i(x,y)V_{ik}{a_j}^{q-k}
    =-2\sum_{j=1}^qc_j\,f(x,y)\,\mathrm{ph}\big((x,y),(a_j,0)\big)\\
    &=-2f(x,y)\,\sum_{j=1}^qc_j\,(-2)\,\frac{p'(x)-p'(a_j)}{x-a_j}
    =-2f(x,y)\,\frac{p''(x)}{p(x)}
    =2p''(x)
    =-\mathrm{Hess}\,f(x,y).
  \end{aligned}
\end{equation*}
Here we have used 
\begin{equation*}
\mathrm{ph}\big((x,y),(a_j,0)\big)
=\frac1{\,2\,}
\left|
\begin{array}{cc}
\frac{-p'(x)-p'(a_j)\,}{x-a_j} & \frac{2y}{\,x-a_j\,}\\[3pt]
\frac{-p'(x)-p'(a_j)\,}{y}     & \frac{2y}{\,y\,}
\end{array}
\right|
=-2\frac{\,p'(x)-p'(a_j)\,}{x-a_j}. 
\end{equation*}
To calculate \(\frac{\partial\log(\varDelta)}{\partial\mu_{2i}}\), 
we remove the assumption \(\mu_2=0\). 
Since \(\varDelta\) is some non-zero constant multiple of
\begin{equation*}
\prod_{i<j}(a_i-a_j)^2,
\end{equation*}
we easily get the \(q\times q\)-matrix
\(\big[\frac{\partial \mu_{2i}}{\partial a_j}\big]\), and then we get
\(\frac{\partial\log(\varDelta)}{\partial\mu_{2i}}\) by using its
inverse matrix.  For \((e,q)=(3,4)\), \((3,5)\), we know only a proof
by direct calculation with Maple by using the explicit form of
\(\varDelta\) and the operators \(L_{v_j}\)s.
\end{proof}
\vskip 5pt
\begin{remark}
{\rm%
There is another proof of \ref{hessian_formula_} in \cite{oss2024}). 
}
\end{remark}
The following lemma is required by the proof of the next proposition. 
\begin{lemma}\label{liftable}
The derivations \(\ell_j\)s and \(L_j\)s are liftable 
to the vector fields on \(\mathscr{C}\). 
Namely, there exist vector fields \,\(\widetilde{\ell_j}\)s 
and \,\(\widetilde{\ell_j}\)s on \(\mathscr{C}\) 
such that their induced vector fields with respect to
\,\(\pi\,:\,\mathscr{C}\rightarrow \mathrm{Spec}\,\mathbb{Q}[\mu]\) 
coincide with \,\(\ell_j\)s and \,\(L_j\)s, respectively. 
\end{lemma}
\begin{proof}%
Recall the definition of \(T\) that is (\ref{def_T_tilde}), namely, 
\begin{equation*}
-eq\cdot f(X,Y)\cdot M_i=\sum_{b\in\wt{M}}^{2g}T_{a,b}\,M_a+A_a\,f_1(X,Y)+B_a\,f_2(X,Y). 
\end{equation*}
Then the vector field
\begin{equation*}
\ell_i=\sum_{j=1}^{2g}T_{ij}\,\frac{\partial}{\,\partial\mu_{w_j}}\in\Der(\mathbb{Q}[\mu])
\end{equation*}
is liftable with respect to  \,\(\pi\) : \(\mathscr{C}\rightarrow \mathrm{Spec}\,\mathbb{Q}[\mu]\). 
Indeed, if we define
\begin{equation*}
\widetilde{\ell}_i
=\sum_{j=1}^{2g}T_{ij}\,\frac{\partial}{\,\partial\mu_{w_j}}
+A_i\,\frac{\partial}{\,\partial x}
+B_i\,\frac{\partial}{\,\partial y},
\end{equation*}
namely, for any \,\(G(\mu,X,Y)\in\mathbb{Q}[\mu,X,Y]\),
\begin{equation*}
\widetilde{\ell}_i\,:\,
G(\mu,X,Y) \ \longmapsto\ 
\bigg(\,\sum_{j=1}^{2g}T_{ij}\,\frac{\partial}{\,\partial\mu_{w_j}}
+A_i\,\frac{\partial}{\,\partial X}
+B_i\,\frac{\partial}{\,\partial Y}\bigg)G(\mu,X,Y)\bigg|_{(X,Y)=(x,y)},
\end{equation*}
we see at \((x,y,\mu)\in\mathscr{C}\) that
\begin{equation*}
\widetilde{\ell}_if
=\sum_{j=1}^{2g}T_{ij}\,\frac{\partial f}{\,\partial\mu_{w_j}}
+A_i\,\frac{\partial f}{\,\partial x}
+B_i\,\frac{\partial f}{\,\partial y}
=\sum_{j=1}^{2g}T_{ij}\,M_j+A_i\,f_1(x,y)+B_i\,f_2(x,y)
=f(x,y)\cdot M_i=0.
\end{equation*}
Therefore \(\widetilde{\ell}_i\) is a vector field on \(\mathscr{C}\) and 
its restriction to \(\mathrm{Spec}\,\mathbb{Q}[\mu]\) is \(\ell_i\). 
We see similarly that \(L_i\) is also liftable to \(\mathscr{C}\). 
\end{proof}
\vskip 5pt
On the operators in \(\vec{L}\) and the discriminant \(\varDelta\), 
we have the following.  
\begin{proposition}\label{tangent_to_disc}
For a derivation \(D\in\Der(\mathbb{Q}[{\mu}])\), 
\(D\) is tangent to the discriminant \(\varDelta\) 
if and only if \(D\in\vec{L}\). 
\end{proposition}
\begin{proof}
This follows from Kyoji Saito's theorem 
(see Theorem A4 in \cite{bruce}).  
See also Corollary 3 on p.2716 to the theorem on 
the previous page in \cite{zakalyukin_1984} 
and Corollary 3.4 in \cite{arnold_1976}.  
However, the \lq\lq\,if\,\rq\rq-part of the statement for the cases 
\((e,q)=(2,q)\), \((3,4)\) is contained in \ref{hessian_formula_}. 
\end{proof}
\vskip 8pt
Obviously, we see
\begin{equation*}
\bigg[\bigoplus_i\mathbb{Q}[\mu]\frac{\partial}{\,\partial\mu_i\,},\ \bigoplus_i\mathbb{Q}[\mu]\frac{\partial}{\,\partial\mu_i\,}\bigg]
\subset\,\bigoplus_i\mathbb{Q}[\mu]\frac{\partial}{\,\partial\mu_i\,},
\end{equation*}
but on \(\{\ell_i\}\), we have the following. 
\begin{corollary}\label{commutator_is_lin_comb}
We have  \,\([\ell_i,\,\ell_j]\in\bigoplus_i\mathbb{Q}[\mu]\,\ell_i\) 
in \,\(\Der(\mathbb{Q}[\mu])\). 
Moreover, 
\begin{equation*}
\vec{L}=\bigoplus_i\mathbb{Q}[\mu]\,\ell_i
\subset\Der_{\mathbb{Q}[\mu]}\left(
\mathbb{Q}[\mu,\,x,\,y]\frac{dx}{\,f_y\,}
+d\,\bigg(\mathbb{Q}[\,\mu,\,x,\,y]\frac{1}{\,f_y\,}\bigg)
\right). 
\end{equation*}
\end{corollary}
\begin{proof}
Since \,\([\ell_i,\,\ell_j]\)\, is tangent to \(\varDelta\) by \ref{hessian_formula_}, 
\ref{tangent_to_disc} shows the first assertion. 
Therefore, we have the second assertion from \ref{coroll_chevalley}. 
\end{proof}
\vskip 5pt
\noindent
Because of \ref{V_and_T}, \ref{commutator_is_lin_comb}, 
and \(\det([H_{ij}])=(eq)^{2g}\in\mathbb{Q}\), 
we see the Lie algebra generated by \(L_k\)s in (\ref{ell_operators}) 
is no other than \(\vec{L}\), namely, 
\begin{equation}\label{Lie_L_is_span_Lj}
\vec{L}=\bigoplus_k\mathbb{Q}[\mu]\,L_k.
\end{equation}
The structure constants of the algebra \(\vec{L}\) with respect to 
\(\{L_{v_j}\}\) belong to \(\mathbb{Q}[\mu]\), 
so it is a polynomial Lie algebra, as discussed in \cite{bl_2002}.  
The corresponding fundamental relations 
of \(\{L_{v_j}\}\) for \((e,q)=(2,3)\),
\((2,5)\), \((2,7)\), and \((3,4)\) are available on request.
\par
\newpage
\subsection{The sigma function as a solution of the heat equations} 
\label{sigma_is_hat_sigma}
Before showing that the function (\ref{classical_sigma}) 
is exactly the sigma function \(\sigma(u)\) 
(see Lemma 4.17 in \cite{bl_2008}), 
we shall first describe some heuristic arguments supporting this result. 
\par
From the definition of \(L_0\) and \ref{V_1j}, we have
\begin{equation}\label{wt_0_operator}
L_0=\sum_j(eq-v_j)\mu_{eq-v_j}\frac{\partial}{\partial\mu_{eq-v_j}}, \ \text{and} \ \
L_0\big(F({\mu})\big)=-\wt\big(F({\mu})\big)F({\mu})
\end{equation}
for any homogeneous form \(F({\mu})\in\mathbb{Q}[{\mu}]\). 
The operator \(L_0\) is called \textit{Euler vector field}. 
\begin{lemma}\label{Gamma_of_L_0}
On \(H_{\mathrm{dR}}^1(\mathscr{C}/\mathbb{Q}[\mu])\), 
we have
\begin{equation*}
{L}_0\,\tp{\vec{\omega}} = 
\left[
  \begin{array}{ccc|ccc}
 -w_g &        &       \\[-5pt]
      & \ddots &       \\[-5pt]
      &        & -w_1 &     &        &      \\
\hline
      &        &      & w_g &        &      \\[-5pt]
      &        &      &     & \ddots &      \\[-5pt]
      &        &      &     &        & w_1
  \end{array}
\right]\tp{\vec{\omega}}
=\Gamma_0\tp{\vec{\omega}}.
\end{equation*}
\end{lemma}
\begin{proof}
For a power series expansion at \(\infty\) of any 1-form \(\omega\) of homogeneous weight \(w\) 
\begin{equation*}
\omega=\sum_j c_j\,t^{j+w}\,dt, 
\end{equation*}
where \(t\) is a local parameter at \(\infty\) of weight \(1\), 
\(c_j\in\mathbb{Q}[{\mu}]\) is of homogeneous weight 
\(-j-1\), 
\begin{equation*}
\begin{aligned}
L_0\Big(\sum_j c_j t^{j+w}dt\Big)&+w\sum_j c_j t^{j+w}dt
 =\sum_j (j+1)c_j t^{j+w}dt+\sum_j wc_j t^{j+w}dt\\
&=\sum_j (j+w+1)c_j t^{j+w}dt
 =d\Big(\sum_j c_j\,t^{j+w+1}\Big).
\end{aligned}
\end{equation*}
If \(\omega\) is any one of \,\(\omega_{w_i}\), which is of the form
\smash{\(\dfrac{h(x,y)}{f_2(x,y)}\,dx\)} with
\(h(x,y)\in\mathbb{Q}[{\mu}][x,y]\), we see the last above is
\begin{equation}\label{exact_form}
d\bigg(\frac{h(x,y)}{f_2(x,y)}\cdot{t}\,\frac{dx}{dt}\bigg).
\end{equation}
Since we can choose \(t\) as a quotient of monomials of \(x\) and \(y\)
(see \cite{onishi_2018}, Section 3), (\ref{exact_form}) is an exact form. 
So that
\begin{equation*}
  L_0(\omega)=-\wt(\omega)\,\omega \ \ \mbox{in \,
    \(H_{\mathrm{dR}}^1(\mathscr{C}/\mathbb{Q}[{\mu}])\)}. 
\end{equation*}
As \,\(\wt(\omega_i)=i\) \,and \,\(\wt(\eta_{-i})=-i\), the statement
is now obvious.  
\end{proof} 
\vskip 5pt The function \(\sigma(u)\), characterized in
\ref{char_sigma}, is a power series of homogeneous weight, which must
be written as\footnote{%
  A different notation here, the \(\ell_{w_j}\)s here are positive
  integers, not the operators in (\ref{small_ell_operators}). }
\begin{equation}\label{homogeneous_sigma}
\begin{aligned}
\sigma&(u)={u_1}^{(e^2-1)(q^2-1)/24}\\
&\cdot\hskip -10pt\sum_{\{n_{eq-v_j},\,\ell_{w_i}\}}
\hskip -10pt\strutf{0pt}{10pt}
\smash{a(\ell_{w_2},\cdots,\ell_{w_g},n_{eq-v_1},\cdots,n_{eq-v_{2g}})
 \hskip 0pt\prod_{j=1}^{2g}(\mu_{eq-v_j}{u_1}^{eq-v_j})^{n_{eq-v_j}}
  \prod_{i=2}^g{\bigg(\dfrac{u_{w_i}}{{u_1}^{w_i}}\bigg)^{\ell_{w_i}}},}
  \end{aligned}
\end{equation}
where the
\(a(\ell_{w_2},\cdots,\ell_{w_g},n_{eq-v_1},\cdots,n_{eq-v_{2g}})\)'s
are absolute constants and the set of \(3g-1\) variables
\(\{n_{eq-v_j},\,\ell_{w_i}\}\) runs through the non-negative integers
such that
\begin{equation*}
\frac{(e^2-1)(q^2-1)}{24}-\sum_{i=2}^g\ell_{w_i}+\sum_{j=1}^{2g}n_{eq-v_j}\geq0.
\end{equation*}
\vskip 8pt
\par
Here, we shall mention that in Lemma 4.17 of \cite{bl_2005} and its proof here 
along the lines of our reconstruction of BL theory, 
that is if \(\sigma(u)\) can be written as 
\(\sigma(u)=\varDelta^{-M}\tilde{\sigma}(u)\) with a numerical 
constant \(M\), then \(M=\frac18\).  
Before doing so, we point out the following lemma. 
\begin{lemma}\label{weight_of_top_constant} 
The constant
\begin{equation}\label{top_constant}
\big({(2\pi)^g}/{(\mathrm{det}\,{\omega'})}\big)^{\frac12}\det(V)^{-\frac18}
=\big({(2\pi)^g}/{(\mathrm{det}\,{\omega'})}\big)^{\frac12}\det(T)^{-\frac18}, 
\end{equation}
is of weight \({(e^2-1)(q^2-1)}/{24}\).
Hence, 
if \((e,q)=(2,3)\), \((2,5)\), \((2,7)\), or \((3,4)\), 
{\rm (}see {\rm \ref{lemma_discri})} or if the conjecture \ref{conj_D} is true, 
\(\big({(2\pi)^g}/{(\mathrm{det}\,{\omega'})}\big)^{\frac12}\varDelta^{-\frac18}\) 
is also of weight \({(e^2-1)(q^2-1)}/{24}\).
\end{lemma}
\vskip 5pt
\begin{proof}
We know the weight of \(\det(V)^{\frac18}\) is \(-eq(e-1)(q-1)/8\) by \ref{weight_detV_detT}. 
The weight of \(\mathrm{det}(\omega')\) is \(\sum_{j=1}^gw_j\), 
which equals 
\begin{equation*}
\sum_{j=1}^gw_j=\frac{eq(e-1)(q-1)}4-\frac{(e^2-1)(q^2-1)}{12}
\end{equation*}
by p.97 of \cite{bl_2004}. 
Hence the weight of the  constant (\ref{top_constant}) is \(\frac{(e^2-1)(q^2-1)}{24}\). 
\end{proof}
\vskip 5pt
\begin{remark}{\rm 
(1) \,Note that the weight of the constant above is exactly that of \,\(\sigma(u)\). \\
(2) \,From (\ref{pre_alg_heq}) and (\ref{top_constant}), 
we see the series (\ref{homogeneous_sigma}) and both of 
\begin{equation*}
\det(V)^{-\frac{1}{\,8\,}}\tilde{\sigma}(u) \ \ \text{and} \ \ 
\varDelta^{-\frac{1}{\,8\,}}\tilde{\sigma}(u)
\end{equation*}
are killed by 
\begin{equation}\label{homo_weight}
\begin{aligned}
L_0-&H^{L_0}-\frac{1}{\,8\,}L_0\log\varDelta
=L_0-H^{L_0}+\mathrm{wt}(\sigma(u))\\
&=\sum_{j=1}^g(eq-v_j)\mu_{eq-v_j}\frac{\partial}{\partial \mu_{eq-v_j}}
  -\sum_{j=1}^g w_j\,u_{w_j}\frac{\partial}{\partial u_{w_j}}
  +\dfrac{(e^2-1)(q^2-1)}{24}. 
\end{aligned}
\end{equation}
}
\end{remark}
\par
In the rest of the paper, we use the notation
\begin{equation*}
H_{v_i}=H^{L_{v_i}}-\tfrac18{L}_{v_i}\log{\varDelta},
\end{equation*}
where \(H^{L_{v_i}}\) is defined by (\ref{H_attached_to_L_j}). 
Then \ref{heq_tilde_sigma} and (\ref{pre_alg_heq}) imply the following.
\begin{theorem}\label{heq_for_hat_sigma}\ 
We have \,
\(({L}_{v_j}-H_{v_j})\,\hat{\sigma}(u)=0\) \,
for \(j=1\), \(\cdots\), \(2g\).
\end{theorem}
The following theorem is one of the important consequences of the BL-theory. 
\begin{theorem}\label{Main_Thm}
The function \,\(\sigma(u)\)\, is equal to \,\(\hat{\sigma}(u)\) 
up to a non-zero absolute constant. 
\end{theorem} 
\begin{proof}
For 
\((e,q)=(2,3)\), \((2,5)\), \((2,7)\), and \((3,4)\), 
we will solve the system of equations
\begin{equation}\label{heat_eq_for_sigma}
(L_{v_i}-H_{v_i})\,\varphi(u)=0 \ \ \ 
(\,i=1, \ \cdots, \ 2g\,)
\end{equation}
for an unknown entire function \(\varphi(u)\), 
and show that the solution space of this system 
is of dimension \(1\), in Section \ref{section3}. 
Any solution eventually satisfies the properties of \(\sigma(u)\) 
in (\ref{char_sigma}). 
Hence we have proved that \(\sigma(u)\) is equal to \(\hat{\sigma}(u)\) 
up to a non-zero absolute constant. 
\end{proof}
\vskip 8pt
From now on, throughout this paper, we denote 
\begin{equation*}
\Gamma_{v_j}=\Gamma^{L_{v_j}}.
\end{equation*}
Especially, \(\Gamma_0=\Gamma^{L_0}=\Gamma^{L_{v_1}}\). 
\vskip 5pt
\newpage
\begin{remark} %
{\rm  
\label{dictionary}
As noted above in {\rm\ref{dictionary_notation}}, our notation differs from
that of Buchstaber and Leykin; we denote the matrix \(\Gamma_j\) in
{\rm p.274} of {\rm\cite{bl_2008}} by \(\Gamma^{\mathrm{BL}}_j\) and we define
the sub-matrices of \(-J\Gamma^{\mathrm{BL}}_j\) 
and \(\Gamma_{v_j}J\) by
\begin{equation*}
-J\Gamma^{\mathrm{BL}}_j
=\Bigg[\,
\begin{matrix}
\alpha^{\mathrm{BL}}_j & \tp{(\beta^{\mathrm{BL}}_j)} \\[5pt]
\beta^{\mathrm{BL}}_j & \gamma^{\mathrm{BL}}_j 
\end{matrix}\,
\Bigg],  \ \ 
\Gamma_{v_j}J
=\Bigg[\,
\begin{matrix}
\alpha_{v_j}       & \beta_{v_j} \\[5pt]
\tp{(\beta_{v_j})} & \gamma_{v_j} 
\end{matrix}\,
\Bigg]
\end{equation*}
by following the notation of {\rm\cite{bl_2008}} and the present paper. 
Then we have for any \(j\) that
\begin{equation*}
    \alpha^{\mathrm{BL}}_j=\alpha_{v_j}, \ \  
    \beta^{\mathrm{BL}}_j=\tp{(\beta_{v_j})}, \ \  
    \gamma^{\mathrm{BL}}_j=\gamma_{v_j}, \ \ 
    \Gamma^{\mathrm{BL}}_j=\tp{(\Gamma_{v_j})}.
\end{equation*}
}
\end{remark}
\newpage
\section{Solving the heat equations}\label{section3}
\subsection{The initial conditions}
For the rest of the paper, 
we shall solve the system of equations (\ref{heat_eq_for_sigma}) 
for the \((2,3)\)-, \((2,5)\)-, \((2,7)\)-, and \((3,4)\)-curves.  
We frequently switch from regarding the \(\mu_j\)s as indeterminates 
to regarding them as elements in \(\mathbb{C}\).
We suppose the following two initial conditions for any solution 
\(\varphi(u)\) solving (\ref{heat_eq_for_sigma}):\\
\textbf{IC1}. \,\(\varphi(u)\in\mathbb{Q}[{\mu}][[u_{w_g},\,\cdots,\,u_{w_1}]]\), and\\
\textbf{IC2}. \(\varphi(u)\) is of homogeneous weight 
\,\(\frac{1}{\,24\,}(e^2-1)(q^2-1)\) 
with respect to \(u_i\)s and \(\mu_j\)s.
\par
Since the property (4) in \ref{char_sigma} is stronger than 
these conditions, there may be a possibility to reduce the
characterization in \ref{char_sigma} of the sigma function, 
in general.
\par
It is not clear to the authors which part of \cite{bl_2008} shows that
the space of the solutions
\(\varphi(u)=\varphi({\mu},u_{w_g},\cdots, u_{w_2},u_{w_1})
\in\mathbb{Q}[{\mu}][[u_{w_g},\cdots, u_{w_2},u_{w_1}]]\) of
(\ref{heat_eq_for_sigma}) is one dimensional.  The main part of the
present paper, that is from \ref{section_23-curve} to the end of the
paper, is a partial answer to this question.
\subsection{General results for the \texorpdfstring{\((2,3)\)}{Lg}-curve}
\label{on_2q_curves} 
In this subsection, we discuss the hyperelliptic case, 
that is the case \(e=2\).  
Firstly, we give the explicit expression for the entries of 
the matrix \(V\) of (\ref{V}).  
The authors know that the issue in this subsection is 
described in p.566 in V.I.Arnol'd's \cite{arnold_1976}
and p.65 in \cite{arnold_1990}.  
Since they do not know any source which contains a proof, 
we shall give a detailed proof here. 
\begin{lemma}\label{hyp_T-matrix}
We have
\begin{equation*}
\begin{aligned}
V_{2i-2,\,2j+2}&=-\frac{2i(q-j)}{q}\,\mu_{2i}\mu_{2j}
    +\sum_{m=1}^{m_0}2(j-i+2m)\,\mu_{2(i-m)}\,\mu_{2(j+m)}\\
    &=-\frac{2i(q-j)}{q}\,\mu_{2i}\mu_{2j}
    +\sum_{\ell=\ell_0}^{i-1\ \mathrm{or}\ j}2(i+j-2\ell)\,\mu_{2\ell}\,
    \mu_{2(i+j-\ell)},
  \end{aligned}
\end{equation*}
where \(\mu_0=1\), \(\mu_2=0\), \(m_0=\mathrm{min}\{i,\,q-j\}\), 
and  \(\ell_0=\mathrm{max}\{0,i+j-q\}\). 
\end{lemma}
\begin{proof} 
First of all, assuming the first equality, we show the second equality. 
To change the first expression to the second with summation to \(i-1\), 
we use the substitution \(\ell=i-m\). 
It is obvious that the second equality with summation to \(i-1\) is equal to 
one with summation to \(j\) for \(i=j\), \(j+1\). 
For the case of \(i<j\), the difference of the two is expressed as 
\begin{equation*}
\sum_{\ell=i}^{j}2(i+j-2\ell)\,\mu_{2\ell}\,\mu_{2(i+j-\ell)}
=-\sum_{\ell'=i}^{j}2(i+j-2\ell')\,\mu_{2(i+j-\ell')}\,\mu_{2\ell'}
\end{equation*}
setting \(\ell'=i+j-\ell\), it is clear that this vanishes.  We see
the case \(j<i\) in a similar way.  The matrix \(V=[V_{2i-2,\,2j+2}]\) is
symmetric by definition.  However, if the Lemma is proved, we see this
directly, by subtracting the term for \(\ell=j\) from the first term.
Now, noting that in the hyperelliptic case, 
the \(M_{2j-2}(X,Y)\) are independent of \(Y\), 
we define \(M^{(2i-2)}=M^{(2i-2)}(X)\in\mathbb{Z}[\mu][X]\) 
by using \([H_{2i-q,\,2j-q}]\) of
(\ref{hyperell_H_ij}):
\begin{equation*}
M^{(2i-2)}(X)=\sum_{j=1}^{q-1}H_{2i-q,\,2j-q}M_{2q-2-2j}(X,Y)
=\sum_{m=0}^{i-1}2(q+1+m-i)\,\mu_{2(i-m-1)}\,X^m.
\end{equation*}
While we are treating \(f(X,Y)=Y^2-p_2(X)\), 
we denote \(p_2(X)\) by \(p(X)\) in this proof, 
for a less cumbersome notation.
\par 
Since \(f_1=p'(X)\), \(f_2=2Y\), we see \(\mathbb{Q}[\mu][X,Y]/(f_1,f_2)\) 
is isomorphic to \(\mathbb{Q}[\mu][X]/(p'(X))\). 
So that, it suffices to know explicitly the
residue \(V^{(i)}=V^{(i)}(X)\) of degree less than \(q-1\) of 
the division of \(p(X)\,M^{(2i-2)}(X)\) by \(p'(X)\) for \(1\leq i\leq q-1\).
The key to this proof is that we actually know the quotient
\(Q^{(2i)}=Q^{(2i)}(X)\in\mathbb{Q}[\mu][X]\), 
as well as \(V^{(2q+2i-2)}=V^{(2q+2i-2)}(X)\) defined below,
of this division!  
Namely, we will show that, if we define functions
\begin{equation*}
Q^{(2i)}(X)=\sum_{m=1}^{i}2\mu_{2(i-m)}\,X^m+\frac{2i}{q}\,\mu_{2i},
\end{equation*}
then the expression
\begin{equation}\label{V_residue}
V^{(2q+2i-2)}(X)=p(X)\,M^{(2i-2)}(X)-p'(x)\,Q^{(2i)}(X)
\end{equation}
is of degree less than \(q-1\).  
Moreover, we can calculate all the terms of \(V^{(2q+2i-2)}\) explicitly, 
which are no other than the \(V_{2i-2,\,2j-2}\)'s. 
\par
Let us start to calculate each term of \(X^k\) of the right hand side of 
(\ref{V_residue})
for any \(k\geq0\). We divide the calculation into four cases. 
\\
(i) The case \(k\geq q\). 
In this case, \(M^{(2i-2)}\) has terms only up to \(X^{i-1}\) 
(\(i-1\leq q-2<q\leq k\)), 
and \(p(X)\) has terms up to \(X^{q}\) (\(q\leq k\)). 
Therefore, we find that the coefficient \(C_{2k}\) 
of \(X^k\) in \(M^{(2i-2)}(X)\,p(X)\) is given by
\begin{equation*}
  \begin{aligned}
    &C_{2k} =\sum_{m=k-q}^{i-1}2(q+1-i+m)\,\mu_{2(i-1-m)}\,\mu_{2(q-k+m)}\\
    &=\sum_{m'=q-1-k}^{i}2(k-m'+1)\,\mu_{2(q-1-k+m')}\,\mu_{2(i-m')},
  \end{aligned}
\end{equation*}
where we have changed the summation index by \(q-k+m=i-m'\). 
On the other hand, \(Q^{(2i)}\) has terms up to \(X^{i}\) 
(\(i\leq q-1<q\leq k\)), 
and \(p'(X)\) has terms up to \(X^{q-1}\) (\(q-1<k\)), so we see
\begin{equation*}
\mbox{\lq\lq\,Coeff. of \,\(X^k\)\, in \(Q^{(i)}(X)\,p'(X)\)\,''} 
    =\!\!\!\sum_{m=k-q+1}^{i}\!\!2\,\mu_{2(i-m)}\,\mu_{2(q-1-k+m)}(k-m+1).
\end{equation*}
So the right hand side of (\ref{V_residue}) has no term
in \(X^k\) for \(k\geq q\). \\
(ii) The case \(k=q-1\). 
Since \(M^{(2i-2)}\) has terms only up to \(X^{i-1}\)
(\(i-1\leq q-2<q-1=k\)), we see that
\begin{equation*}
  \begin{aligned}
    &\mbox{\lq\lq\,Coeff. of \,\(X^k\)\, in \(M^{(2i-2)}(X)\,p(X)\)\,''}
     =\sum_{m=0}^{i-1}2(q+1-i+m)\,\mu_{2(i-1-m)}\,\mu_{2(1+m)}\\
    &=\sum_{m'=0}^{i-1}2(q-m')\,\mu_{2m'}\,\mu_{2(i-m')}
     =\sum_{m'=1}^{i-1}2(q-m')\,\mu_{2m'}\,\mu_{2(i-m')}+2q\,\mu_0\mu_{2i},
  \end{aligned}
\end{equation*}
where we have changed the index of summation by \(m+1=i-m'\). 
In this case  \(Q^{(2i-2)}\) has terms up to \(X^{i}\) (\(i\leq q-1=k\)), 
and \(p'(X)\) has terms up to \(X^{q-1}\) (\(q=k\)), 
we have that the coefficient \(C_{2k}\) of \(X^k\) 
in \(Q^{(2i)}(X)\,p'(X)\) is given by
\begin{equation*}
  \begin{aligned}
    C_{2k}&=\sum_{m=1}^i2\,\mu_{2(i-m)}\,\mu_{2m}(q-m)+\frac{2i}{q}
    \mu_{2i}\,\mu_0\,q\\
    &=\sum_{m=1}^{i-1}2\,\mu_{2(i-m)}\,\mu_{2m}(q-m)+2(q-i)\,\mu_0\,\mu_{2i}
    +2i\,\mu_{2i}\,\mu_0.
  \end{aligned}
\end{equation*}
So the right hand side of (\ref{V_residue}) has no term in \(X^{q-1}\). \\
(iii) The case \(i-1<k<q-1\). 
Since \,\(M^{(2i-2)}\)\, has terms only up to \,\(X^{i-1}\), 
we see that the coefficient \,\(D_{2k}\)\, of \,\(X^k\)\, in
\,\(M^{(2i-2)}(X)\,p(X)\)\, is given by
\begin{equation*}
\begin{aligned}
D_{2k}&=\sum_{m=0}^{i-1}2(q+1+m-i)\,\mu_{2(i-m-1)}\,\mu_{2(q-k+m)}\\
   &=\sum_{m=1}^i2(q+m-i)\,\mu_{2(i-m)}\,\mu_{2(q-1-k+m)}
\end{aligned}
\end{equation*}
by rewriting \(m\) as \(m-1\).  
On the other hand, the coefficient
\(C_{2k}\) of \(X^k\) in \(Q^{(2i)}(X)\,p'(X)\) is
\begin{equation*}
C_{2k}=\frac{2i}{q}\mu_{2i}\,(k+1)\mu_{q-1-k}+\sum_{m=1}^{i}2\,\mu_{2(i-m)}
\,(k-m-1)\mu_{2(q-1-k+m)}. 
\end{equation*}
So the coefficient of \,\(X^k\)\, in the right hand side of
(\ref{V_residue}) is
\begin{equation*}
\sum_{m=1}^{i}2(q-1-k+2m-i)\,\mu_{2(i-m)}\,\mu_{2(q-1-k+m)}
    +\frac{2i}{q}(k+1)\mu_{2i}\,\mu_{2(q-1-k)}
\end{equation*}
and \(V_{2i-2,\,2j-2}\), which is no other than the value of this at
\(k=q-1-j\), is given by
\begin{equation*}
\sum_{m=1}^{i}2(j+2m-i)\,\mu_{2(i-m)}\,\mu_{2(j+m)}
    +\frac{2i}{q}(q-j)\mu_{2i}\,\mu_{2j}
\end{equation*}
as desired since \(i<k+1=q-j\).
\\
(iv) The case \(k<i-1\). 
Since \(M^{(2i-2)}\) has terms up to \(X^{i-1}\), of higher degree than
\(X^k\), we see that the coefficient \(D_{2k}\) of \,\(X^k\)\, in
\(M^{(2i-2)}(X)\,p(X)\) is given by
\begin{equation*}
\begin{aligned}
D_{2k}&=\sum_{m=0}^k2(q+1+m-i)\,\mu_{2(i-m-1)}\,\mu_{2(q-k+m)}\\
   &=\sum_{m=1}^{k+1}2(q+m-i)\,\mu_{2(i-m)}\,\mu_{2(q-1-k+m)}
\end{aligned}
\end{equation*}
on replacing the summation index \(m\) by \(m+1\). 
Similarly, \(Q^{(2i)}\) has terms up to \(X^i\) exceeding \(X^k\) again, and
\begin{equation*}
  \begin{aligned}
    &\mbox{\lq\lq\,Coeff. of \,\(X^k\)\, in \(p'(X)\,Q^{(2i)}(X)\)\,''}\\
    &=\frac{2i}{q}\mu_{2i}\,\mu_{q-1-k}+\sum_{m=1}^{k+1}2\,\mu_{2(i-m)}
    \,(k-m-1)\mu_{2(q-1-k+m)}
  \end{aligned}
\end{equation*}
with an extra term for \(m=k+1\) which is zero.  
So the coefficient of \,\(X^k\)\, in the right hand side 
of (\ref{V_residue}) is
\begin{equation*}
\sum_{m=1}^{k+1}2(q-1-k+2m-i)\,\mu_{2(i-m)}\,\mu_{2(q-1-k+m)}
    +\frac{2i}{q}(k+1)\mu_{2i}\,\mu_{2(q-1-k)}
\end{equation*}
and then \(V_{ij}\), 
which is no other than the value of this at \(k=q-1-j\), is given by
\begin{equation*}
V_{2i-2,\,2j-2}=\sum_{m=1}^{q-j}2(j+2m-i)\,\mu_{2(i-m)}\,\mu_{2(j+m)}
    +\frac{2i}{q}(q-j)\mu_{2i}\,\mu_{2j}
\end{equation*}
as desired since \(q-j=k+1\leq i\).
\end{proof}
\vskip 8pt Secondly, we give the values
\(L_{v_j}(\log\varDelta)=\frac{\,L_{v_j}\varDelta\,}{\varDelta}\) for
the case \((e,q)=(2,q)\).
\begin{lemma}\label{hyp_disc_tangent} 
If \(e=2\), then we have \ 
\(L_{2j}(\varDelta)=-2(q-j)(q-1-j)\mu_{2j}\,\varDelta\) \ 
for \(0\leq j\leq q-2\). 
\end{lemma}
\begin{proof} %
Since the Hessian of \(f(X,Y)=Y^2-p_2(X)\) is 
\begin{equation*}
\begin{aligned}
\bigg|\,
\begin{matrix}
-{p_2}''(X) & 0 \\
0          & 2 
\end{matrix}\,
\bigg|
&=-2\,{p_2}''(X)\\
&=-2\big(q(q-1)X^{q-2}+(q-2)(q-3)\mu_4X^{q-4}+\cdots+2\cdot1\,\mu_{q-3}\big),
\end{aligned}
\end{equation*}
this lemma follows from (\ref{hessian_formula}). 
\end{proof}
\vskip 10pt
\subsection{Heat equations for the \texorpdfstring{\((2,3)\)}{Lg}-curve}
\label{section_23-curve}
In this section we recall Weierstrass' result which gives 
a recursive relation for the coefficients of 
the power series expansion of his sigma function at the origin.  
We refer the reader to (12) and (13) in p.\ 314 of \cite{fs_1882} also. 
Here we derive Weierstrass' result by following the
method of \cite{bl_2008}, namely, following the theory described in
previous sections, but without using the general results \ref{hyp_T-matrix} 
and \ref{hyp_disc_tangent}, in order to demonstrate the ideas of the theory. 

Weierstrass' original method is explained in \cite{weierstrass_1882}
and some explanation of it is available in \cite{onishi_2015b}.  
It is easy to get \(L_0\) and \(L_2\):
\begin{equation}\label{23L_operators}
V=\Bigg[\ \begin{matrix}
4\mu_4 & 6\mu_6\\[3pt]
6\mu_6 & -\frac43{\mu_4}^2
\end{matrix}\ \Bigg], \ \ \mbox{and}\ \ 
\left\{\ 
\begin{aligned}
L_0&=4{\mu_4}\frac{\partial}{\partial\mu_4}+6{\mu_6}\frac{\partial}
{\partial\mu_6}\\
L_2&=6{\mu_6}\frac{\partial}{\partial\mu_4}-\frac43{\mu_4}^2
\frac{\partial}{\partial\mu_6}.
\end{aligned}
\right.
\end{equation}
In this case, we see \,\(V=T\)\, since
\begin{equation*}
\mathrm{ph}\big((x,y),(z,w)\big)=-6x-6z=[\,x\ \ 1\,]
\bigg[\,
  \begin{matrix}
     & -6 \\
  -6 &  
  \end{matrix}\,
\bigg]\,
\bigg[
\begin{matrix}
z \\  1
\end{matrix}
\bigg].
\end{equation*}
Then \(\varDelta=2^2\cdot3\cdot\det(V)\) (see \ref{def_disc_23-curve}). 
The differential forms
\begin{equation*}
  \omega_1=\frac{dx}{2y}, \ \ 
  \eta_{-1}=\frac{xdx}{2y} 
\end{equation*}
form a symplectic basis of 
\(H_{\mathrm{dR}}^1(\mathscr{C}/\mathbb{Q}[\mu])\).  
We have \(\vec{\omega}=(\,\omega_1, \ \eta_{-1}\,)\). 
Bearing in mind Lemma \ref{chevalley}, we proceed by using 
\(x^{-\frac12}\) as the local parameter satisfying 
\(\frac{\partial}{\partial\mu_j}x=0\) for \(j=4\), \(6\), 
and we compute the matrix \(\Gamma\) as follows. 
Using \(f(x,y)=0\), we see \(2y\frac{\partial}{\partial\mu_4}y=x\) and 
\(2y\frac{\partial}{\partial\mu_6}y=1\), so that
\begin{equation*}
  \frac{\partial}{\partial\mu_4}y=\frac{x}{2y}, \ \ \ 
  \frac{\partial}{\partial\mu_6}y=\frac{1}{2y}.
\end{equation*}
Therefore, we have
{\small 
\begin{equation}\label{exact00}
  \frac{\partial}{\partial\mu_6}\omega_1=-\frac{1}{4y^3}dx, \ \ \ 
  \frac{\partial}{\partial\mu_4}\omega_1
  =\frac{\partial}{\partial\mu_6}
  \eta_{-1}=-\frac{x}{4y^3}dx, \ \ \ 
  \frac{\partial}{\partial\mu_4}\eta_{-1}=-\frac{x^2}{4y^3}dx,
\end{equation}
}
and
{\small 
\begin{align}
  d\Big(\frac{1}{y}\Big)\label{exact0}
  &=-\frac1{y^2}dy
    =-\frac1{y^2}\frac{dy}{dx}dx
    =-\frac1{y^2}\,\frac{3x^2+\mu_4}{2y}\,dx
   =6\frac{\partial}{\partial\mu_4}\,\eta_{-1}+2\mu_4\frac{\partial}
    {\partial\mu_6}\,\omega_1,\notag\\
  d\Big(\frac{x}{y}\Big)
  &=\frac{ydx-xdy}{y^2}
    =\frac{y-x\frac{dy}{dx}}{y^2}\,dx
    =\frac{y-x\frac{3x^2+\mu_4}{2y}}{y^2}\,dx\\
  &=\frac{y-\frac{3y^2-2\mu_4x-3\mu_6}{2y}}{y^2}\,dx 
   =-\omega_1-4\mu_4\frac{\partial}{\partial\mu_4}\omega_1
    -6\mu_6\frac{\partial}{\partial\mu_6}\omega_1 
   =-\omega_1-L_0\omega_1,\notag \\ 
  d\Big(\frac{x^2}{y}\Big)
  &=\frac{2xydx-x^2dy}{y^2}
    =\frac{2xy-x^2\frac{dy}{dx}}{y^2}\,dx
    =\frac{2xy-x^2\frac{3x^2+\mu_4}{2y}}{y^2}\,dx  \\
  &=\frac{2xy-\frac{3x(y^2-\mu_4x-\mu_6)+\mu_4\,x^2}{2y}}{y^2}\,dx
    =\frac{xdx}{2y}+\frac{\mu_4x^2}{y^3}dx+\frac{3}{2}\frac{\mu_6x}{y^3}dx \notag\\
  &=\eta_{-1}-L_0\eta_{-1} \notag \\ 
  &=\eta_{-1}+\frac43{\mu_4}^2\frac{\partial}{\partial\mu_6}\omega_1
    -6\mu_6\frac{\partial}{\partial\mu_4}\omega_1-\frac23\mu_4d\Big(\frac{1}{y}\Big)  \ \ \ \ 
    \mbox{(\ by (\ref{exact0})\ )} \notag \\ 
  &= \eta_{-1}-L_2\omega_1-\frac23\mu_4d\Big(\frac{1}{y}\Big). \notag
\end{align}
}\relax 
Accordingly, we see
\begin{align*}
  L_2\eta_{-1} &=6\mu_6\frac{\partial}{\partial\mu_4}\eta_{-1}
  -\frac43{\mu_4}^2\frac{\partial}{\partial\mu_6}\eta_{-1}\\
  &=-2\mu_6\mu_4\frac{\partial}{\partial\mu_6}\omega_1
  -\frac43{\mu_4}^2\frac{\partial}{\partial\mu_4}\omega_1
  +\mu_6\,d\Big(\frac1{y}\Big)\
  \ \ \ \
  \mbox{(\ by (\ref{exact00}) and (\ref{exact0})\ )}\\
  &=-\frac{\mu_4}{3}\Big(6\mu_6\frac{\partial}{\partial\mu_6}\omega_1
  +4\mu_4\frac{\partial}{\partial\mu_4}\omega_1\Big)
  +\mu_6\,d\Big(\frac1{y}\Big)\\
  &=-\frac{\mu_4}{3}L_0\omega_1+\mu_6\,d\Big(\frac1{y}\Big)
  =\frac{\mu_4}{3}\omega_1+\mu_6\,d\,\Big(\frac1{y}\Big).
\end{align*}
Summarising these results, we have on
\(H_{\mathrm{dR}}^1(\mathscr{C}/\mathbb{Q}[\mu])\) 
that
\begin{equation*}
  L_0\tp{\vec{\omega}}=\Gamma_0\tp{\vec{\omega}}, \ \ 
  L_2\tp{\vec{\omega}}=\Gamma_2\tp{\vec{\omega}}, \ \ 
  \mbox{where}\ \ 
  \Gamma_0
  =\bigg[\,
\begin{matrix}
-1 &  \\  & 1
\end{matrix}\, \bigg], \ \ \Gamma_2 =\bigg[\,
\begin{matrix}
& 1 \\ \frac{\mu_4}3 & 
\end{matrix}
\,\bigg].
\end{equation*}
Note that, by these equations, we have 
\(L_j\varOmega=\Gamma_j\varOmega\) with 
\(\varOmega=\bigg[\begin{matrix}\omega' & \omega''\\\eta' & \eta''
\end{matrix}\bigg]\) as in (\ref{ell_and_Gamma}). 
Since \((L_0,\,L_2)(\det(T))=(12,\,0)\det(T)\) 
(by \ref{hyp_disc_tangent}), 
\(\varDelta=2^2\cdot3\cdot\det(T)\), 
and (\ref{Gamma_of_L_0}), we have arrived at
{\small 
\begin{equation}\label{heat23}
\begin{aligned}
  (L_0-H_0)\sigma(u) &=\left(4{\mu_4}\frac{\partial}{\partial\mu_4}
    +6{\mu_6}\frac{\partial}{\partial\mu_6}
    -u\frac{\partial}{\partial u}+1\,\right)\varphi(u)=0,\\
  (L_2-H_2)\sigma(u) &=\left(6{\mu_6}\frac{\partial}{\partial\mu_4}
    -\frac43{\mu_4}^2\frac{\partial}{\partial\mu_6}
    -\frac12\frac{\partial^2}{\partial u^2}
    +\frac16{\mu_4}u^2\,\right)\varphi(u)=0,
\end{aligned}
\end{equation}
}
where \ \(H_j=H^{L_j}+\tfrac18L_j\log\varDelta\) \ for \ \(j=0\) and
\(2\).  From the first of (\ref{heat23}) and the conditions
\textbf{IC1}, \textbf{IC2}, the solution function is of the form
\begin{equation*}
\varphi(u)=u\sum_{n_4,n_6\geqq0}b(n_4,n_6)
\frac{(\mu_4u^4)^{n_4}(\mu_6u^6)^{n_6}}{(1+4n_4+6n_6)!}.
\end{equation*}
Using the second equation we then have a recurrence relation
\begin{equation}
\begin{aligned}\label{recurrence}
b(&n_4,n_6)=\tfrac23(4n_4+6n_6-1)
(2n_4+3n_6-1)b(n_4-1,n_6)\\
&\qquad  -\tfrac{8}{3}(n_6+1)b(n_4-2,n_6+1) +12(n_4+1)b(n_4+1,n_6-1)
\end{aligned}
\end{equation}
with \ \(b(n_4,n_6)=0\) \ if \ \(n_4<0\) \ or \ \(n_6<0\).  Since the
term \(b(n_4,n_6)\) on the left hand side has weight \(4n_4+6n_6\),
and the terms \(b(i,j)\) on the right hand side have weight smaller
than this, all terms may be found from (\ref{recurrence}).  Therefore,
any solution of (\ref{heat23}) is a constant times the function
\begin{equation*}
  \varphi(u)=\sigma(u)=u 
  +2\mu_4\tfrac{u^5}{5!}
  +24\mu_6\tfrac{u^7}{7!}
  -36\mu_4^2\tfrac{u^9}{9!}
  -288\mu_4\mu_6\tfrac{u^{11}}{11!}
  +\cdots.
\end{equation*}
\newpage
\subsection{Heat equations for the \texorpdfstring{\((2,5)\)}{Lg}-curve}
\label{section_25-curve}
In this section, we list the analogous results for the heat equations
for the curve
\begin{equation*}
\mathscr{C}_{{\mu}}\, :\, y^2=x^5+\mu_4x^3+\mu_6x^2+\mu_8x+\mu_{10}.
\end{equation*}
We note here that our results correct a sign in \cite{bl_2005}; 
the overall constant \(\frac1{80}\) at the \(4\)th line from bottom in
page 68 of \cite{bl_2005} should be \(-\frac1{80}\).  
Here we give the Hurwitz series version of the algorithm. 
Now, we take a usual symplectic basis of differentials
\begin{equation*}
  \omega_3=\frac{1}{2y}dx,\ \ 
  \omega_1=\frac{x}{2y}dx,\ \ 
  \eta_{-3}=\frac{3x^3+\mu_4x}{2y}dx,\ \ 
  \eta_{-1}=\frac{x^2}{2y}dx
\end{equation*}
of \(H_{\mathrm{dR}}^1(\mathscr{C}/\mathbb{Q}[\mu])\).  
The matrix \(V\) for this case is given by
\begin{equation*}
  V=
  \left[\,\begin{matrix}
      4\mu_{4} &
      6\mu_{6} &
      8\mu_{8} &
      10\mu_{10} \\
      6\mu_{6} &
      -\tfrac{12}{5}\mu_{4}^2+8\mu_{8} &
      -\tfrac85\mu_{4}\mu_{6}+10\mu_{10} &
      -\tfrac45\mu_{4}\mu_{8} \\
      8\mu_{8} &
      -\tfrac85\mu_{4}\mu_{6}+10\mu_{10} &
      -\tfrac{12}5\mu_{6}^2+4\mu_{4}\mu_{8} &
      6\mu_{4}\mu_{10}-\tfrac65\mu_{6}\mu_{8} \\
      10\mu_{10} &
      -\tfrac45\mu_{4}\mu_{8} &
      6\mu_{4}\mu_{10}-\tfrac65\mu_{6}\mu_{8} &
      4\mu_{10}\mu_{6}-\tfrac85\mu_{8}^2
\end{matrix}\,\right]. 
\end{equation*}
A calculation by \texttt{Maple} along \ref{def_discriminant} gives \(\varDelta=2^4\cdot5\cdot\det(V)\). 
The operators \(L_j\) are given by
\begin{equation*} \tp{[L_0 \ L_2 \ L_4 \ L_6]} =V\,
\tp{\bigg[
\frac{\partial}{\partial\mu_4}\ \
\frac{\partial}{\partial\mu_6}\ \ 
\frac{\partial}{\partial\mu_8}\ \
\frac{\partial}{\partial\mu_{10}}
\bigg]}.
\end{equation*}
While the authors have the explicit commutation relations of these \(L_i\), 
we shall not include these here because their 
explicit forms are not needed in this paper.
However, these commutators are all in the span of the \(L_i\). 
By \ref{hyp_disc_tangent}, we see that 
these \(L_j\)'s operate on the discriminant \(\varDelta\) as follows:
\begin{equation*}
[\,L_0\ \ L_2\ \ L_4\ \ L_6\,]\varDelta=[\,40\ \ \ 0\ \ 
12\,\mu_4\ \ 4\,\mu_6\,]\varDelta.
\end{equation*}
The representation matrices \(\Gamma_j\) for the \(L_j\) acting on 
\(H_{\mathrm{dR}}^1(\mathscr{C}/\mathbb{Q}[\mu])\) are 
{\small 
\begin{equation*}
\begin{aligned}
\Gamma_0&=
\left[\,
\begin{matrix}
-3 &    &   &   \\
   & -1 &   &   \\ 
   &    & 3 &   \\ 
   &    &   & 1
\end{matrix}\,
\right], \ \ \ 
\Gamma_2
=\left[\,
\begin{matrix}
                               & -1              &   &                  \\[1pt]
\tfrac45\,\mu_4                &                 &   & 1                \\[1pt]
\tfrac45\,{\mu_4}^{2}-3\mu_8 &                 &   & -\tfrac45\mu_4     \\[1pt]
                               & \tfrac35\,\mu_4 & 1 &   
\end{matrix}\,
\right],\\
\Gamma_4
&=\left[\,
\begin{matrix}
-\mu_4                           &                 &       & 1             \\[1pt]
\tfrac65\,\mu_6                  &                 & 1     &               \\[1pt] 
\tfrac65\,\mu_4\mu_6-6\mu_{10} & -\mu_8          & \mu_4 & -\tfrac65\mu_6  \\[1pt]
-\mu_8                           & \tfrac25\,\mu_6 &       & 
\end{matrix}\,
\right], \,
\Gamma_6
=\left[\,
\begin{matrix}
                     &                 & 1 &                  \\[1pt]
\tfrac35\,\mu_8      &                 &   &                  \\[1pt]
\tfrac35\,\mu_4\mu_8 & -2\,\mu_{10}    &   & -\tfrac35\,\mu_8 \\[1pt]
-2\,\mu_{10}         & \tfrac15\,\mu_8 &   &    
\end{matrix}\,
\right]. \\
\end{aligned}
\end{equation*}
}
Therefore, we find the following operators \(H_j\):
\begin{align*}
  H_0&=3{u_3}\frac{\partial}{\partial{u_3}}+{u_1}\frac{\partial}
  {\partial{u_1}}-3,\\
  H_2&=\frac12\frac{\partial^2}{\partial{u_1}^2}
  +{u_1}\frac{\partial}{\partial{u_3}}
  -\frac45{\mu_4}{u_3}\frac{\partial}{\partial{u_1}}
  -\frac3{10}{\mu_4}{u_1}^2
  -\Big(\frac32{\mu_8}-\frac25{\mu_4}^2\Big){u_3}^2,\\
  H_4&=\frac{\partial^2}{\partial{u_1}\partial{u_3}}
  -\frac65{\mu_6}{u_3}\frac{\partial}{\partial{u_1}}
  +{\mu_4}{u_3}\frac{\partial}{\partial{u_3}}
  -\frac15{\mu_6}{u_1}^2+{\mu_8}{u_1}{u_3}
  +\Big(3{\mu_{10}}-\frac35{\mu_4}{\mu_6}\Big){u_3}^2\mathrlap{-\,\mu_4,}\\
  H_6&=\frac12\frac{\partial^2}{\partial{u_3}^2}
  -\frac35{\mu_8}{u_3}\frac{\partial}{\partial{u_1}}
  -\frac1{10}{\mu_8}{u_1}^2 +2{\mu_{10}}{u_3}{u_1}
  -\frac3{10}{\mu_8}{\mu_4}{u_3}^2 -\frac12{\mu_6}.
\end{align*}
By the equation \((L_0-H_0)\,\varphi(u)=0\) 
and the conditions \textbf{IC1} and \textbf{IC2}, 
the solution function must be of the form
\begin{equation*}
\begin{aligned}
  \varphi(u)=\sigma(u_3,u_1)&=\hskip -40pt
  \sum_{\substack{m,n_4,n_6,n_8,n_{10}\geq0\\3-3m+4n_4+6n_6+8n_8+10n_{10}\geq0}}
  \hskip -20pt
  \Big[\,b(m,n_4,n_6,n_8,n_{10})\\
  &\cdot\frac{{u_1}^3\Big(\dfrac{u_3}{{u_1}^3}\Big)^m
    \Big(\mu_4{u_1}^4\Big)^{n_4}
    \Big(\mu_6{u_1}^6\Big)^{n_6}
    \Big(\mu_8{u_1}^8\Big)^{n_8}
    \Big(\mu_{10}{u_1}^{10}\Big)^{n_{10}}}
  {m!\,(3-3m+4n_4+6n_6+8n_8+10n_{10})!}\ \text{\raise 5pt\hbox{\(\Bigg]\)}}.
\end{aligned}
\end{equation*}
Let \ \(k=3-3m+4n_4+6n_6+8n_8+10n_{10}\). 
Then the other heat equations \((L_j-H_j)\,\varphi(u)=0\)
imply the following recursion scheme: 
\vspace{-3pt}
\begin{equation*}
  b(m,n_4,n_6,n_8,n_{10})
  =\left\{
  \begin{aligned}
    \ B_2 & \ \ \ \mbox{(\,if \(k>1\) and \(m\geq0\)\,)} \\
    \ B_1 & \ \ \ \mbox{(\,if \(k=1\) and \(m>0\)\,)}     \\
    \ B_0 & \ \ \ \mbox{(\,if \(k=0\) and \(m>1\)\,)},
  \end{aligned}
\right.
\end{equation*}
\vspace{-3pt}
where the \(B_i\) are given by
\begin{align*}
     B_2=             20(n_8+1)\,&b(m  ,n_4  ,n_6  ,n_8+1,n_{10}-1)\\
                      +16(n_6+1)\,&b(m  ,n_4  ,n_6+1,n_8-1,n_{10}  )\\
                      +12(n_4+1)\,&b(m  ,n_4+1,n_6-1,n_8  ,n_{10}  )\\
             -\tfrac{24}5(n_6+1)\,&b(m  ,n_4-2,n_6+1,n_8  ,n_{10}  )\\
             +\tfrac35(k-3)(k-2)\,&b(m  ,n_4-1,n_6  ,n_8  ,n_{10}  )\\
             -\tfrac85(n_{10}+1)\,&b(m  ,n_4-1,n_6  ,n_8-1,n_{10}+1)\\
             -\tfrac{16}5(n_8+1)\,&b(m  ,n_4-1,n_6-1,n_8+1,n_{10}  )\\
                         -2(k-2)\,&b(m+1,n_4  ,n_6  ,n_8  ,n_{10}  )\\
                        -3m(m-1)\,&b(m-2,n_4  ,n_6  ,n_8-1,n_{10}  )\\
                 +\tfrac45m(m-1)\,&b(m-2,n_4-2,n_6  ,n_8  ,n_{10}  )\\
                      +\tfrac85m\,&b(m-1,n_4-1,n_6  ,n_8  ,n_{10}  ),\\
   B_1=               +10(n_6+1)\,&b(m-1,n_4  ,n_6+1,n_8  ,n_{10}-1)\\
             -\tfrac{12}5(n_8+1)\,&b(m-1,n_4  ,n_6-2,n_8+1,n_{10}  )\\
             -\tfrac65(n_{10}+1)\,&b(m-1,n_4  ,n_6-1,n_8-1,n_{10}+1)\\
                       +8(n_4+1)\,&b(m-1,n_4+1,n_6  ,n_8-1,n_{10}  )\\
-\tfrac15(5m{-}10{+}8n_6{-}20n_8{-}30n_{10})\,&b(m-1,n_4-1,n_6  ,n_8  ,n_{10}  )\\
                    -3(m-1)(m-2)\,&b(m-3,n_4  ,n_6  ,n_8  ,n_{10}-1)\\
             +\tfrac35(m-1)(m-2)\,&b(m-3,n_4-1,n_6-1,n_8  ,n_{10}  )\\
                  +\tfrac65(m-1)\,&b(m-2,n_4  ,n_6-1,n_8  ,n_{10}), \\
  B_0=\ \      -\tfrac{16}5(1+n_{10})\,&b(m-2,n_4  ,n_6  ,n_8-2,n_{10}+1)\\
          -\tfrac15(12n_8-40n_{10}-5)\,&b(m-2,n_4  ,n_6-1,n_8  ,n_{10}  )\\
                           +20(n_4+1)\,&b(m-2,n_4+1,n_6  ,n_8  ,n_{10}-1)\\
                           +12(n_8+1)\,&b(m-2,n_4-1,n_6  ,n_8+1,n_{10}-1)\\
                     -\tfrac85(n_6+1)\,&b(m-2,n_4-1,n_6+1,n_8-1,n_{10}  )\\
                  +\tfrac35(m-2)(m-3)\,&b(m-4,n_4-1,n_6  ,n_8-1,n_{10}  )\\
                       +\tfrac65(m-2)\,&b(m-3,n_4  ,n_6  ,n_8-1,n_{10}).
\end{align*}%
From these, we see that the expansion of \(\sigma(u)\) 
is Hurwitz integral over \(\mathbb{Z}[\tfrac15]\). %
\begin{remark}
{\rm 
Actually the above recurrence scheme is one of 
several possible recurrence relations.  
However, we see any such system gives the same 
solution space by the following argument. 
Here, of course, we suppose that  \(b(m, n_4, \cdots, n_{10})=0\)  
if  \(k\)  or any of the explicit arguments is negative. 
For any finite subset  
\(S\subset\{(m,\,n_4,\,\cdots,\,n_{10})\,|\,k,\,n_4,\,\cdots,\,n_{10}\geq 0\}\), 
we take the set \,\(E_S\)\, of relations \,\(h\)\,  
between  \,\(\{b(m,\,n_4,\,\cdots,\,n_{10})\}\)\,  
such that any  \,\(b(m,\,n_4,\,\cdots,\,n_{10})\)\, appears 
as a term in \(h\) provided 
that \((m,n_4,\cdots,n_{10})\in S\). 
For instance, if we consider
\begin{equation*}
S= \{ (1,0,0,0,0),\ (0,0,0,0,0),\ (0,1,0,0,0),\ 
 (1,1,0,0,0),  \ (2,1,0,0,0)\}, 
\end{equation*}
then  \(E_S\)  consists of the following \(4\) equations: 
\begin{equation*}
\begin{aligned}
  b(0,0,0,0,0)& =  -2b(1,0,0,0,0), \\
  b(0,1,0,0,0)& =  12b(0,0,0,0,0)-10b(1,1,0,0,0),\\
  b(1,1,0,0,0)& = \tfrac65 b(1,0,0,0,0)-4b(2,1,0,0,0)
  -\tfrac{16}5b(0,0,0,0,0),\\
  b(2,1,0,0,0)& = \tfrac15\cdot{0}\cdot{}b(1,0,0,0,0).
\end{aligned}
\end{equation*}
The solution space of such a system of linear equations \(E_S\)
is of dimension 1 or larger because we have at least one iteration
system as above whose solution space is of dimension \(1\).
Since \(E_S\)
is independent of the choice of recursion system, any recursion
system must include the same solution space of dimension \(1\).
} 
\end{remark}
The first few terms of the sigma expansion are given as
follows (up to a constant multiple):
\begin{equation*}
\begin{aligned}
  \sigma({u_3}&,{u_1})={u_3}-2\frac{{u_1}^3}{3!}
  -4\mu_4\frac{{u_1}^7}{7!}  -2\mu_4\frac{{u_3}{u_1}^4}{4!}
  +64\mu_6\frac{{u_1}^9}{9!}  -8\mu_6\frac{{u_3}{u_1}^6}{6!}\\
  &-2\mu_6\frac{{u_3}^2{u_1}^3}{2!3!}
  +\mu_6\frac{{u_3}^3}{3!}
  +(1600\mu_8-408{\mu_4}^2)\frac{{u_1}^{11}}{11!}\\
  &-(4{\mu_4}^2+32\mu_8)\frac{{u_3}{u_1}^8}{8!}
  -8\mu_8\frac{{u_3}^2{u_1}^5}{2!5!}
  -2\mu_8\frac{{u_3}^3{u_1}^2}{3!2!}  +\cdots.
\end{aligned}
\end{equation*}
\newpage
\subsection{The heat equations for the \texorpdfstring{\((2,7)\)}{Lg}-curve}
\label{section_27-curve}
\noindent
We take the hyperelliptic genus three curve \(\mathscr{C}\)
in the Weierstrass form
\[
y^2=f(x)=x^7+\mu_4x^5+\mu_6x^4+\mu_8x^3+\mu_{10}x^2+\mu_{12}x+\mu_{14}.
\]
The discriminant \(\varDelta\) of \(\mathscr{C}\) is the resultant 
of \(f\) and \(f_1\). 
It has \(320\) terms and is of weight \(84\). 
The matrix \(V\) is given by
{\footnotesize 
\begin{align*}
V&=\left[
\begin{array}{ccccccccc}
  4\,\mu_4     & 6\,\mu_6
  & 8\,\mu_8                                                            \\
  6\,\mu_6     & -\tfrac47\,\left(5\,{\mu_4}^2-14\,\mu_8 \right)
  & -\tfrac27\,\left(8\,\mu_6\mu_4 -35\,\mu_{10}\right)                   \\
  8\,\mu_8     & -\tfrac27\,\left(8\,\mu_6\mu_{{4}}-35\,\mu_{10}\right)
  & \tfrac47\,\left(21\,\mu_{12}-6\,{\mu_6}^{2}+7\,\mu_4\mu_8\right)    \\
  10\,\mu_{10} & -\frac{12}{7}\,\left(\mu_4\mu_8-7\,\mu_{12} \right)
  & \tfrac27\,\left(49\,\mu_{14}-9\,\mu_6\mu_8+21\,\mu_4\mu_{10}\right) \\
  12\,\mu_{12} & -\tfrac27\,\left(4\,\mu_4\mu_{10} -49\,\mu_{14}\right)
  & \tfrac47\,\left(14\,\mu_4\mu_{12}-3\,\mu_6\mu_{10}\right)           \\
  14\,\mu_{14} & -\tfrac47\,\mu_4\mu_{12}
  & \tfrac27\,\left(35\,\mu_4 \mu_{14} -3\,\mu_6\mu_{12} \right)          
\end{array}\right.\\
&\left.
\begin{array}{ccccccccc}
  & 10\,\mu_{10}                                                                   & 12\,\mu_{12}                                                                 \\ 
  & -\tfrac{12}{7}\,\left(\mu_4\mu_8-7\,\mu_{12}\right)                            &  -\tfrac27\,\left(4\,\mu_4\mu_{10}-49\,\mu_{14}\right)                       \\
  & \tfrac27\,\left(21\,\mu_4\mu_{10}-9\,\mu_6\mu_8+49\,\mu_{14}\right)            & \tfrac47\,\left(14\,\mu_4\mu_{12}-3\,\mu_6\mu_{10}\right)                    \\  
  & \tfrac47\,\left(7\,\mu_6\mu_{10} -6\,{\mu_8}^{2}+14\,\mu_4\mu_{12} \right)     & \tfrac27\,\left(21\,\mu_6\mu_{12}-8\,\mu_{10}\mu_8+35\,\mu_4\mu_{14}\right)  \\ 
  & \tfrac27\,\left(21\,\mu_6\mu_{12} -8\,\mu_{10}\mu_8 +35\,\mu_4\mu_{14}\right)  & \tfrac47\,\left(7\,\mu_{12}\mu_8- 5\,{\mu_{10}}^2+14\,\mu_6\mu_{14}\right)   \\ 
  & \tfrac{8}{7}\,\left(7\,\mu_6\mu_{14} - \mu_{12}\mu_8\right)                    & \tfrac27\, \left(21\,\mu_{14}\mu_8 -5\,\mu_{12}\mu_{10}\right)                 
\end{array}\right.\\
&\hskip 100pt\left.
\begin{array}{ccccccccc}
& 14\,\mu_{14}                                                   \\
&  -\tfrac47\,\mu_4\mu_{12}                                      \\
& \tfrac27\,\left(35\,\mu_4\mu_{14} -3\,\mu_6\mu_{12}\right)     \\
& \frac{8}{7}\,\left(7\,\mu_6\mu_{14} -\mu_{12}\mu_8 \right)     \\
& \tfrac27\,\left(21\,\mu_{14}\mu_8 -5\,\mu_{12}\mu_{10}\right)  \\
& \tfrac47\,\left(7 \,\mu_{14}\mu_{10} -3\,{\mu_{12}}^{2}\right) 
\end{array}\right].
\end{align*}
}
Here a calculation by \texttt{Maple} 
along \ref{def_discriminant} shows that \(\varDelta=2^6\cdot7\cdot\det(V)\). 
Then we have
\begin{equation*} \tp{[L_0 \ L_2 \ L_4 \ L_6 \ L_8 \ L_{10}]}=V
\tp{\bigg[\frac{\partial}{\partial\mu_4}\ \
  \frac{\partial}{\partial\mu_6}\ \ \frac{\partial}{\partial\mu_8}\ \
  \frac{\partial}{\partial\mu_{10}}\ \
  \frac{\partial}{\partial\mu_{12}}\ \
  \frac{\partial}{\partial\mu_{14}}\bigg]}.
\end{equation*}
Using \ref{hyp_disc_tangent}, 
their operation on \(\varDelta\) are given by
\begin{equation*}
  [L_0 \ L_2 \ L_4 \ L_6 \ L_8 \ L_{10}](\varDelta)=[84 \ \ 0 \ \
  40\mu_4 \ \ 24\mu_6 \ \ 12\mu_8 \ \  4\mu_{10}]\varDelta.
\end{equation*}
As for the \((2,5)\)-case, we have fundamental relations for these
\({L}_i\) as a set of generators of certain Lie algebra, which we do
not include here.  The symplectic basis of
\(H_{\mathrm{dR}}^1(\mathscr{C}/\mathbb{Q}[{\mu}])\) is
{\small 
\begin{equation*}
  \begin{aligned}
   &\omega_5=\frac{dx}{2y}, \ \ \ 
    \omega_3=\frac{xdx}{2y}, \ \ \ 
    \omega_1=\frac{x^2dx}{2y}, \ \
   \eta_{-5}=\frac{(5x^5+3\mu_4x^3+2\mu_6x^3+\mu_8x^2)dx}{2y}, \\
  &\eta_{-3}=\frac{(3x^4+\mu_4x^2)dx}{2y}, \ \ 
   \eta_{-1}=\frac{x^3dx}{2y}. 
  \end{aligned}
\end{equation*}
}
With respect to these, the matrices 
\(\Gamma_j
=\smash{\bigg[
\begin{array}{cc}
  -\beta_j & \alpha_j \\ 
  -\gamma_j & \tp{\beta_j}\end{array}
\bigg]}\) 
are given as follows:
{\small
 \begin{align*}
  \alpha_0 & = O, \ \ 
            \beta_0 = \left[
            \begin{array}{ccc}
              5 &   &   \\[-4pt]
                & 3 &   \\[-4pt]
                &   & 1 
            \end{array}
                      \right], \ \ 
                      \gamma_0=O, \\ 
  \alpha_2&=\left[
            \begin{array}{ccc}
            &     &   \\
            &     &   \\
       \ \  & \ \ & 1
            \end{array}
            \right], \ 
 \beta_2=\left[
   \begin{array}{ccc}
                &        3       &      \\
 -\tfrac47\mu_4 &                &   1  \\
                & -\tfrac87\mu_4 &    
   \end{array}\right], 
\ \gamma_2 =\left[
 \begin{array}{ccc}
   -\tfrac17(4\mu_4\mu_8-35\mu_{12})
   &        &      \\
   & -\tfrac17(8{\mu_4}^2-21\mu_8) &  \\
            &          & -\tfrac57\mu_4 
\end{array}
\right],\\
\alpha_4 
&=\left[%
\begin{array}{ccc}
       &   &   \\
       &   & 1 \\
 \ \   & 1 &   
\end{array}
\right], \ \ 
\beta_4
=\left[
\begin{array}{ccc}
  3\mu_4        &                  &   1  \\
 -\tfrac67\mu_6 & \mu_4            &      \\
                & -\frac{12}7\mu_6 &                   
\end{array}
\right], 
\ \ \gamma_4
=\left[
\begin{array}{ccc}
-\tfrac27(3\mu_6\mu_8-35\mu_{14}) & 3\mu_{12}                        &                \\
3\mu_{12}                         & -\tfrac67(2\mu_4\mu_6-7\mu_{10}) &  \mu_8         \\
                                  & \mu_8                            & -\tfrac47\mu_6   
\end{array}
\right], \\
\alpha_6 
&=\left[
\begin{array}{ccc}
  &   & 1 \\
  & 1 &   \\
1 &   &   
\end{array}
\right], \ 
\beta_6
=\left[
\begin{array}{ccc}
     2\mu_6     &       \mu_4    & \ \ \ \ \ \ \\
-\frac87\,\mu_8 &                &             \\
                & -\tfrac97\mu_8 &    
\end{array}
\right], \ 
\gamma_6
=\left[
\begin{array}{ccc}
3\mu_4\mu_{12}-\tfrac87{\mu_8}^2 & 6\mu_{14}                     & \mu_{12}      \\
6\mu_{14}                        & -\tfrac97\mu_4\mu_8+9\mu_{12} & 2\mu_{10}     \\
\mu_{12}                         & 2\mu_{10}                     & -\tfrac37\mu_8  
\end{array}
\right], \\
\alpha_{8} 
&=\left[
\begin{array}{ccc}
  & 1 &  \\
1 &   &  \\
  &   &
\end{array}
\right], \ \ 
\beta_8
=\left[
\begin{array}{ccc}
             \mu_8     &                   & \ \ \ \ \ \ \ \ \\
  -\tfrac{10}7\mu_{10} &                   &                 \\
            \mu_{12}   & -\tfrac67\mu_{10} &                  
\end{array}
\right], \\ 
&\gamma_8
=\left[
\begin{array}{ccc}
6\mu_4\mu_{14}+2\mu_6\mu_{12}-\tfrac{10}7\mu_8\mu_{10}
                    & \mu_4\mu_{12}                     & 2\mu_{14}         \\
    \mu_4\mu_{12}   & -\tfrac67\mu_4\mu_{10}+12\mu_{14} & 3\mu_{12}         \\
        2\mu_{14}   & 3\mu_{12}                         & -\tfrac27\mu_{10}
\end{array}
\right], \\
\alpha_{10}
& =
\left[
\begin{array}{ccc} 
1 & \ \ \ & \ \ \ \\
  &       &       \\
  &       & 
\end {array}
\right], \quad
\beta_{10} =  
\left[
\begin{array}{ccc} 
 \ \ \ \ \ \        &                    & \ \ \ \ \ \ \ \ \ \\
 -\frac57\,\mu_{12} &                    &                   \\
    2\,\mu_{14}     & -\frac37\,\mu_{12} & 
\end{array}
\right], \ 
\gamma_{10} =
\left[
\begin{array}{ccc}
  4\,\mu_6\mu_{14}-\frac57\,\mu_{8}\mu_{12} & 2\,\mu_4\mu_{14}        &
  \\
  2\,\mu_4\mu_{14} & -\frac37\,\mu_4\mu_{12} & 4\,\mu_{14}
  \\ & 4\,\mu_{14}             & -\frac17\,\mu_{12}
\end{array}
 \right]. 
 \end{align*}
}
These give a set of heat equations \((L_j-H_j)\,\sigma(u)=0\) as before. 
\subsection{The sigma function for the \texorpdfstring{\((2,7)\)}{Lg}-curve }\ 
\label{section_27-sigma} 
We now solve (\ref{heat_eq_for_sigma}) in the \((2,7)\) case. 
The initial conditions \textbf{IC1}, \textbf{IC2} of (\ref{heat_eq_for_sigma}) 
in this case are as follows:
\begin{equation*}
\varphi(u)\in\mathbb{Q}[{\mu}][[u_5,\,u_3,\,u_1]], \ 
\mbox{and \,\(\varphi(u)\)\, is of homogeneous weight \(6\).}
\end{equation*}
Following \cite{bl_2005} but in the Hurwitz series form as
\cite{weierstrass_1882}, 
we write any solution \,\(\varphi(u)\)\, as
\begin{equation*}
\begin{aligned}
  \varphi&(u_5,u_3,u_1)=\sum_{\substack{{\ell},m,n_4,n_6,n_8,\\ n_{10},n_{12},n_{14}}}
  \!\Bigg[\,b({\ell},m,n_4,n_6,n_8,n_{10},n_{12},n_{14})\,\\
  &\cdot\left(\mu_4{u_1}^4\right)^{n_4}
  \left(\mu_6{u_1}^6\right)^{n_6} \left(\mu_8{u_1}^8\right)^{n_8}
  \left(\mu_{10}{u_1}^{10}\right)^{n_{10}}
  \left(\mu_{12}{u_1}^{12}\right)^{n_{12}}
  \left(\mu_{14}{u_1}^{14}\right)^{n_{14}}\\
  &\cdot\frac{{u_1}^6
    \Big(\dfrac{u_5}{{u_1}^5}\Big)^{\ell}
    \Big(\dfrac{u_3}{{u_1}^3}\Big)^m}
{\,(6-5{{\ell}}-3m+4n_4+6n_6+8n_8+10n_{10}+12n_{12}+14n_{14})!\,{{\ell}}!\,m!\,}\Bigg],
\end{aligned}
\end{equation*}
giving a solution of \((L_0-H_0)\,\varphi(u)=0\).  
If we define
\begin{equation*}
  k = 6-5{\ell}-3m+4n_4+6n_6+8n_8+10n_{10}+12n_{12}+14n_{14},
\end{equation*}
the above expression is rewritten as 
\begin{equation}\label{27_sigma_expansion}
\begin{aligned}
  \varphi(u_5,u_2,u_1)&=\sum
  b({\ell},m,n_4,n_6,n_8,n_{10},n_{12},n_{14})\\
  &\ \ \ \ \
  \cdot{\mu_4}^{n_4}\,{\mu_6}^{n_6}\,{\mu_8}^{n_8}\,{\mu_{10}}^{n_{10}}\,
  {\mu_{12}}^{n_{12}}\,{\mu_{14}}^{n_{14}}\,
  \smash{
        \frac{{u_5}^{{\ell}}}{{\ell}!}\, 
        \frac{{u_3}^m}{m!}\,
        \frac{{u_1}^k}{k!}
        },
\end{aligned}
\end{equation}
where we require all the integer indices \(k\), \({\ell}\), \(m\),
\(n_4\), \(n_6\), \(n_8\), \(n_{10}\), \(n_{12}\), \(n_{14}\) to be
non-negative.
\par
Note that the $u$-weight of this expression is
\(k_0=6+4n_4+6n_6+8n_8+10n_{10}+12n_{12}+14n_{14}\), 
which does not depend on \(\ell\) or \(m\). 
(Note also that \(k=k_0-5{\ell}-3m\)). 
For fixed \(n_4\), \(n_6\), \(n_8\), \(n_{10}\), \(n_{12}\),
\(n_{14}\geq0\), \(k_0\geq0\) is fixed, and for non-negative \(k\), we
require \({\ell}=0\), \(\dots\) , \(\lfloor (k_0+6)/5\rfloor\),
\(m=0\), \(\dots\)\,, \(\lfloor(6+k_0-5{\ell})/3 \rfloor\).
As noted above, if we insert this ansatz into the equation for
\((L_0-H_0)\varphi=0\), we get an expression which is identically zero,
for any set of \(b({\ell},m,n_4,n_6,n_8,n_{10},n_{12},n_{14})\).
\par
If we insert this ansatz into the expression for \((L_2-H_2)\varphi=0\), 
we get (after some algebra, and providing \(k>0\)) 
the recurrence relation shown below, involving \(20\) terms 
(compare the equations on p.68 of \cite{bl_2005} for the genus 2 case). 
We can structure this relation by the {\em weight} of 
each \(b\) coefficient of (\ref{27_sigma_expansion}) 
(more precisely by the weight of the corresponding 
term in the expansion).  
We will call this  \(P_2\) :
\begin{align*}
  -7\,&b({{\ell}},m,n_4,n_6,n_8,n_{10},n_{12},n_{14})\\
  +14(2-k)\,&b({{\ell}},m{+}1,n_4,n_6,n_8,n_{10},n_{12},n_{14})\\
  -42m\,&b({{\ell}}{+}1,m{-}1,n_4,n_6,n_8,n_{10},n_{12},n_{14})\\
  +196(n_{12}+1)\,&b({{\ell}},m,n_4,n_6,n_8,n_{10},n_{12}{+}1,n_{14}{-}1)\\
  +168(n_{10}+1)\,&b({{\ell}},m,n_4,n_6,n_8,n_{10}{+}1,n_{12}{-}1,n_{14})\\
  +140(n_8+1)\,&b({{\ell}},m,n_4,n_6,n_8{+}1,n_{10}{-}1,n_{12},n_{14})\\
  +112(n_6+1)\,&b({{\ell}},m,n_4,n_6{+}1,n_8{-}1,n_{10},n_{12},n_{14})\\
  -40(n_6+1)\,&b({{\ell}},m,n_4{-}2,n_6{+}1,n_8,n_{10},n_{12},n_{14})\\
  +5(3-k)(2-k)\,&b({{\ell}},m,n_4{-}1,n_6,n_8,n_{10},n_{12},n_{14})\\
  -8(n_{14}+1)\,&b({{\ell}},m,n_4{-}1,n_6,n_8,n_{10},n_{12}-1,n_{14}+1)\\
  -16(n_{12}+1)\,&b({{\ell}},m,n_4{-}1,n_6,n_8,n_{10}{-}1,n_{12}{+}1,n_{14})\\
  -24(n_{10}+1)\,&b({{\ell}},m,n_4{-}1,n_6,n_8{-}1,n_{10}{+}1,n_{12},n_{14})\\
  -32(n_8+1)\,&b({{\ell}},m,n_4{-}1,n_6{-}1,n_8{+}1,n_{10},n_{12},n_{14})\\
  +84(n_4+1)\,&b({{\ell}},m,n_4{+}1,n_6{-}1,n_8,n_{10},n_{12},n_{14})\\
  -21m(m-1)\,&b({{\ell}},m{-}2,n_4,n_6,n_8{-}1,n_{10},n_{12},n_{14})\\
  +8m(m-1)\,&b({{\ell}},m{-}2,n_4{-}2,n_6,n_8,n_{10},n_{12},n_{14})\\
  +16m\,&b({{\ell}},m{-}1,n_4{-}1,n_6,n_8,n_{10},n_{12},n_{14})\\
  -35{{\ell}}({{\ell}}{-}1)\,&b({{\ell}}{-}2,m,n_4,n_6,n_8,n_{10},
                               n_{12}{-}1,n_{14})\\
  +4{{\ell}}({{\ell}}{-}1)\,&b({{\ell}}{-}2,m,n_4{-}1,n_6,n_8{-}1,n_{10},
                              n_{12},n_{14})\\
  +8{{\ell}}\,&b({{\ell}}{-}1,m{+}1,n_4{-}1,n_6,n_8,n_{10},n_{12},n_{14})=0.
\end{align*}
This relation applies  only for \(k>1\), and can be written as 
\begin{align*}
  P_2:\,&b(\ell,m,n_4,n_6,n_8,n_{10},n_{12},n_{14})\\
&=2\,(2-k)\,b(\ell,1{+}m,n_4,n_6,n_8,n_{10},n_{12},n_{14}) \\
&\ \ \ -6m\,b(1{+}\ell,m{-}1,n_4,n_6,n_8,n_{10},n_{12},n_{14})
        +\mbox{\lq\lq lower weight terms''},
\end{align*}
where the lower weight terms have coefficients which are quadratic 
or linear in 
\({\ell}\), \(m\), \(n_4\), \(n_6\), \(n_8\), \(n_{10}\), \(n_{12}\), \(n_{14}\),
times integers or rational numbers with denominators \(7\). 
Here the number \(4n_4+6n_6+8n_8+10n_{10}+12n_{12}+14n_{14}\) for
\begin{equation*}
b({\ell},m,n_4,n_6,n_8,n_{10},n_{12},n_{14})
\end{equation*}
is the \(\mu\)-{\it weight} of the term. 
For  \(P_2\), the left hand side and the first two terms on the right
hand side all have the \(\mu\)-weight 
\(W=4n_4+6n_6+8n_8+10n_{10}+12n_{12}+14n_{14}\).  
The next highest \(\mu\)-weight terms of  the \lq\lq lower weight terms''
are of \(\mu\)-weight \(W-2\), 
and the lowest weight terms are of \(\mu\)-weight \(W-12\). 
\par
Putting the same ansatz into \((L_4-H_4)\varphi=0\) 
we get another recurrence \(P_4\) with \(20\) terms, providing \(m>0\) and \(k>0\).  
We can write this as 
\begin{align*}
P_4 : \,&b(\ell,1+m,n_4,n_6,n_8,n_{10},n_{12},n_{14})\\
&= -7(k-1)\,b(1\,{+}\,\ell,m{-}1,n_4,n_6,n_8,n_{10},n_{12},n_{14})+\mbox{\lq\lq lower weight terms''}.
\end{align*}
Here the lower weight terms have the same property as \(P_2\). 
We have another relation from the equation \((L_6-H_6)\varphi=0\)
\begin{align*}
P_6 :\,b(&\ell,m+2,n_4,n_6,n_8,n_{10},n_{12},n_{14})\\
  &+2\,b(1+\ell,m,n_4,n_6,n_8,n_{10},n_{12},n_{14})
=\mbox{\lq\lq lower weight terms''}.
\end{align*}
We can write this in two different ways which will each come in useful
\begin{align*}
P_{6\mathrm{a}}&:\,b(l,m,n_4,n_6,n_8,n_{10},n_{12},n_{14})\\
&=-2\,b(l{+}1,m{-}2,n_4,n_6,n_8,n_{10},n_{12},n_{14})
+\mbox{\lq\lq lower weight terms''},\\
P_{6\mathrm{b}}&: 2\,b(l,m,n_4,n_6,n_8,n_{10},n_{12},n_{14})\\
&=-b(l{-}1,m{+}2,n_4,n_6,n_8,n_{10},n_{12},n_{14})
+\mbox{\lq\lq lower weight terms''}.
\end{align*}
Continuing, we have two further relations, \,from the equations 
\,\((L_8-H_8)\varphi=0\) \\
and \((L_{10}-H_{10})\varphi=0\), 
\begin{equation*}
\begin{aligned}
  P_8\ &: \ b({\ell}{+}1,m{+}1,n_4,n_6,n_8,n_{10},n_{12},n_{14})
  =\mbox{\lq\lq lower \(\mu\)-weight terms''}, \\
  P_{10}\ &: \ b({\ell}{+}2,m,n_4,n_6,n_8,n_{10},n_{12},n_{14})
  =\mbox{\lq\lq lower \(\mu\)-weight terms''}, \\
\end{aligned}
\end{equation*}
where the lower \(\mu\)-weight terms have the same properties as
\(P_2\) and \(P_4\).  The relations \(P_6\), \(P_8\), \(P_{10}\) have
a total of \(24\), \(24\), \(19\) terms respectively. 
As before, we need to normalise the expansion, so we choose \,
\(b(1,0,0,0,0,0,0,0)\) \(=1\).  
We need to find relations which either express coefficients in terms of ones with 
lower or equal \(\mu\)-weight.
\par
Clearly we must take care with our recurrence relation to avoid
infinite looping.  We find that the following choice of recurrence
scheme results in a sequence which decreases the \(\mu\)-weight
after no more than one extra step at any point in the recurrence :
\begin{equation*}
\begin{aligned}
b(&{\ell},m,n_4,n_6,n_8,n_{10},n_{12},n_{14})\\
&=\begin{cases}
\ 0   &\mbox{if \
  \(\mathrm{min}\{k,{\ell},m,n_4,n_6,n_8,n_{10},n_{12},n_{14}\}<0\)}, \\
\ 1   &\mbox{if \ \({\ell}=1\), \(m=n_4=n_6=n_8=n_{10}=
n_{12}=n_{14}=0\)}, \\
\ \mathrm{rhs}(P_2)             &\mbox{if \ \(k>1\)}, \\
\ \mathrm{rhs}(P_4)             &\mbox{if \ \(k=1\), \(m>0\)}, \\
\ \mathrm{rhs}(P_{\mathrm{6a}}) &\mbox{if \ \(k=1\), \(m=0\) \ ( and \({\ell}>0\))}, \\
\ \mathrm{rhs}(P_{\mathrm{6b}}) &\mbox{if \ \(k=0\), \(m>1\)},\\
\ \mathrm{rhs}(P_8)             &\mbox{if \ \(k=0\), \(m=1\) \ and \({\ell}>0\)},\\
\ \mathrm{rhs}(P_{10})          &\mbox{if \ \(k=0\), \(m=0\) \ and \({\ell}>1\)}.
\end{cases}
\end{aligned}
\end{equation*}
Note that the structure of this complicated linear recurrence relation
does {\it not} depend on the moduli \(\mu_i\).  
We have used this to calculate the terms on the Hurwitz series for the solution
up to weight 40 in \(u_i\) (weight \(34\) in the \(\mu_i\)).
As for the \((2,5)\)-curve, there is another possible recursion scheme:
\begin{equation*}
\begin{aligned}
b(&{\ell},m,n_4,n_6,n_8,n_{10},n_{12},n_{14})\\
&=\begin{cases}
\ 0  & \mbox{if \ \(\mathrm{min}\{k,{\ell},m,n_4,n_8,n_{10},n_{12},n_{14}\}<0\)}, \\
\ 1  & \mbox{if \ \ \({\ell}=1\), \(m=n_4=n_8=n_{10}=n_{12}=n_{14}=0\)}, \\
\ \mathrm{rhs}(P_{10})          & \mbox{if \ \({\ell}>1\)},\\
\ \mathrm{rhs}(P_{8})           & \mbox{if \ \({\ell}=1\), \(m>0\)},\\
\ \mathrm{rhs}(P_{\mathrm{6a}}) & \mbox{if \ \({\ell}=1\), \(m=0\) and \(k>0\)},\\
\ \mathrm{rhs}(P_4)             & \mbox{if \ \({\ell}=0\), \(m>0\) and \(k>0\)},\\
\ \mathrm{rhs}(P_{\mathrm{6b}}) & \mbox{if \ \({\ell}=0\), \(m>0\) and \(k=0\)},\\
\ \mathrm{rhs}(P_2)             & \mbox{if \ \({\ell}=0\), and  \(m= 0\)}.
\end{cases}
\end{aligned}
\end{equation*}
We have used this to calculate the
terms in the  series up to weight 40 in \(\{u_j\}\), 
or equivalently, weight 35 in the \(\{\mu_i\}\).
The first few terms of the  expansion are given as
follows (up to a constant multiple):
\begin{align*}
  \varphi(u)=\sigma(u_5,& u_3,u_1) =  16\frac{{u_1}^6}{6!}
                - 2\frac{{u_1}^3u_3}{3!} 
             - 2\frac{{u_3}^2}{2!} + {u_5 u_1} 
             + 64\mu_4\frac{{u_1}^{10}}{10!} 
             + 36 \mu_4 \frac{{u_1}^7u_3}{7!} \\
           & - 4 \mu_4\frac{{u_1}^4{u_3}^2}{4!2!} 
             - 2 \mu_4 \frac{u_1{u_3}^3}{3!} 
             + 2 \mu_4\frac{u_5{u_1}^5}{5!} 
             - 512 \mu_6\frac{{u_1}^{12}}{12!}
             + 64 \mu_6\frac{{u_1}^9u_3}{9!} \\
           & + 16 \mu_6\frac{{u_1}^6{u_3}^2}{6!2!}
             - 8 \mu_6\frac{{u_1}^3{u_3}^3}{3!^2} 
             - 8 \mu_6\frac{{u_3}^4}{4!}
             + 24 \mu_6\frac{u_5{u_1}^7}{7!}+\cdots. 
\end{align*}%
Further studies are required to establish whether there are other 
recursion schemes which can be used to generate the series,
and which recursions could be considered the most efficient in some sense. 
\newpage
\subsection{The heat equations for the \texorpdfstring{\((3,4)\)}{Lg}-curve}\ 
\label{section_34-curve}
We take the trigonal genus three curve
\(\mathscr{C}=\mathscr{C}_{{\mu}}^{3,4}\) in the Weierstrass form
\[
y^3=(\mu_8+\mu_5 x+\mu_2 x^2)y + x^4+\mu_6 x^2+\mu_9 x +\mu_{12}.
\]
The matrix \(V\) is given by
\begin{equation*}
V=\underset{\substack{a\,\in\{0,3,4,6,7,10\}\\ b\,\in\{2,5,6,8,9,12\}}}{[V_{a,b}]}
=\left[ 
\begin {array}{cc}
    V^{(1,1)}   & V^{(1,2)}\\[3pt]
\tp{V^{(1,2)}}  & V^{(2,2)}
\end {array}
\right], 
\end{equation*}
where
\begin{equation*}
\begin{aligned}
V^{(1,1)}&
=\left[
\begin {array}{ccc} 
  2\,\mu_2  & 5\,\mu_5     & 6\,\mu_6 \\[3pt]
  5\,\mu_5  & \frac16\,{\mu_2}^4-4\,\mu_2\mu_6+8\,\mu_8
                    & -\frac12\,{\mu_2}^2\mu_5+9\,\mu_9 \\[3pt]
  6\,\mu_6  & -\frac12\,{\mu_2}^2\mu_5+9\,\mu_9
           & \frac23\,{\mu_2}^2\mu_6+\frac{10}3\,\mu_2\mu_8+\frac53\,{\mu_5}^2
\end {array}
\right] 
\end{aligned}
\end{equation*}
and
{\footnotesize
\begin{equation*}
  \begin{aligned}
  V^{(1,2)}&=
  \left[
  \begin {array}{cc}
  8\,\mu_8 & 9\,\mu_9 \\[5pt]
    \frac{1}{12}\,{\mu_2}^3\mu_5-\frac12\,\mu_2\mu_9-\frac32\,\mu_5\mu_6
           & \frac16\,{\mu_2}^3\mu_6{-}\frac13\,{\mu_2}^2\mu_8{-}
             \frac16\,\mu_2{\mu_5}^2{-}3\,{\mu_6}^2{+}12\,\mu_{12} \\[5pt]
    \frac43\,{\mu_2}^2\mu_8-\frac7{12}\,\mu_2{\mu_5}^2 +12\,\mu_{12}
   & \frac43\,{\mu_2}^2\mu_{9} -\frac76\,\mu_2\mu_5\mu_6+\frac{13}3\,\mu_5\mu_8
  \end {array}\right. \\
&\hskip 100pt\left.
   \begin {array}{c}
   12\,\mu_{12}                                                           \\[5pt]  
   \frac1{12}\,{\mu_2}^3\mu_9-\frac16\,\mu_2\mu_5\mu_8-\frac32\,\mu_6\mu_9\\[5pt]  
    2\,{\mu_2}^2\mu_{12}-\frac7{12}\,\mu_2\mu_5\mu_9+\frac83\,{\mu_8}^2            
        \end {array}\right], 
  \end{aligned}
    \end{equation*}}
and the remaining elements are
{\small 
\begin{align*}
V_{6,8} 
  &=\frac1{24}\,{\mu_2}^2{\mu_5}^2+6\,\mu_2\mu_{12}-\frac72\,\mu_5\mu_9
    +4\,\mu_6\mu_8,\\
V_{6,9} &= V_{5,4} = \frac1{12}\,{\mu_2}^2\mu_5\mu_6
            +\frac76\,\mu_2\mu_5\mu_8-\frac{5}{12}\,{\mu_5}^3
          -\frac32\,\mu_6\mu_9,\\
V_{6,10} &= V_{6,4}=
          \frac1{24}\,{\mu_2}^2\mu_5\mu_9 +\frac43\,\mu_2{\mu_8}^2
          -\frac5{12}\,{\mu_5}^2\mu_8+6\,\mu_6\mu_{12}-\frac94\,{\mu_9}^2,\\
V_{7,9} &=\frac16\,{\mu_2}^2{\mu_6}^2+2\,{\mu_2}^{2}\mu_{12} 
          +\frac53\,\mu_2\mu_5\mu_9 -\frac83\,\mu_2\mu_6\mu_8-\frac43\,{\mu_5}^2\mu_6
          +\frac83\,{\mu_8}^2,\\
V_{7,12} &=V_{6,5}=\frac1{12}\,{\mu_2}^2\mu_6\mu_9+3\,\mu_2\mu_5\mu_{12}
-\frac56\,\mu_2\mu_8\mu_9-\frac5{12}\,{\mu_5}^2\mu_9 -\frac12\,\mu_5\mu_6\mu_8,\\
V_{10,12} &=\frac1{24}\,{\mu_2}^2{\mu_9}^2+2\,\mu_2\mu_8\mu_{12} +{\mu_5}^2\mu_{12} 
-\frac{11}{6}\,\mu_5\mu_8\mu_9 +\frac43\,\mu_6{\mu_8}^2.
\end{align*}
}
The discriminant \(\varDelta\) of \(\mathscr{C}\)
is calculated by \texttt{Maple} along \ref{def_discriminant},  
has \(670\) terms and is of weight \(72\). 
The result shows \(\varDelta=3^3\,{\cdot}\,4^2\,{\cdot}\,\det(V)\). 
Now, after a calculation by \texttt{Maple} (which shows \ref{hessian_formula_}), 
we have 
\begin{equation*}
{[\,{L}_0\,\ {L}_3\,\ {L}_4\,\ {L}_6\,\ {L}_7\,\ {L}_{10}]}\,\varDelta 
=[\,72\ \ 0\,\ -8{\mu_2}^2\,\ 12\mu_6\,\ -8\mu_2\mu_5\,\ -{\mu_5}^2-4\mu_2\mu_8\,]\varDelta.
\end{equation*}
As in the \((2,7)\)-case, 
we have fundamental relations for these \({L}_i\)
as a set of generators of certain Lie algebra. 
The symplectic basis of \(H_{\mathrm{dR}}^1(\mathscr{C}/\mathbb{Q}[{\mu}])\) in
this case is given by
\begin{equation*}
    \omega_{5}=\frac{dx}{3y^2}, \ 
    \omega_{2}=\frac{xdx}{3y^2}, \  
    \omega_{1}=\frac{ydx}{3y^2}, \ 
    \eta_{-5}=\frac{5x^2ydx}{3y^2}, \ 
    \eta_{-2}=\frac{2xydx}{3y^2}, \ 
    \eta_{-1}=\frac{x^2dx}{3y^2}.
\end{equation*}
\newpage
\noindent
The matrices \(\Gamma_j=\bigg[
\begin{array}{cc}
  -\beta_j &\ \alpha_j \\ 
  -\gamma_j &\ \tp{\beta_j}\end{array}
\bigg]\) 
are given as follows\footnote{These should not be confused
  with the symplectic basis of cycles \(\alpha_j\) and
  \(\beta_j\) in (\ref{per_matrix})} :
{\footnotesize 
\begin{align*}
\alpha_0&=O, \ \ 
\beta_0=\left[
\begin{array}{ccc}
5 &   &   \\
  & 2 &   \\
  &   & 1 
\end{array}
\right], \ \
\gamma_0=O, \\
\alpha_3&=
\left[
\begin{array}{ccc}
\ \ &   &   \\[3pt]
    &   & 1 \\[3pt]
    & 1 & 
\end{array}
\right], \quad 
\beta_3 =
\left[
\begin{array}{ccc}
                                      &            2        &                 \\[3pt]
 \frac1{12}\,{\mu_2}^3-\frac32\,\mu_6 &                     & -\frac12\,\mu_2 \\[3pt]
 -\frac16\,\mu_5\mu_2                 & -\frac13\,{\mu_2}^2 &                  
\end{array}
\right], \\
\gamma_3&=
\left[
\begin{array}{ccc}
 5\,\mu_5\mu_8-\frac13\,{\mu_2}^2 \mu_9+\frac1{12}\,{\mu_2}^4\mu_5
 -\frac53\,\mu_2\mu_5\mu_6    & -\mu_2\mu_8         & \mu_9  \\[3pt]
  -\mu_2\mu_8                  & \frac23\,\mu_5\mu_2  &
   \frac16\,{\mu_2}^3 -\mu_6 \\[3pt]
 \mu_9 & \frac16\,{\mu_2}^3-\mu_6 & \mu_5 
\end{array}
\right], \\
\alpha_4&
=\left[
\begin {array}{ccc}
\ \  &   &               \\[3pt] 
     & 1 &               \\[3pt]
     &   &\frac13\,\mu_2
\end{array}
\right], \ \ 
\beta_4=
\left[
\begin{array}{ccc}
\frac23\,{\mu_2}^2      &                & 1 \\[3pt]
-\frac7{12}\,\mu_5\mu_2 &                &   \\[3pt]
\frac53\,\mu_8          & \frac23\,\mu_5 & 
\end{array}
\right], \\
\gamma_4&=
\left[
\begin{array}{ccc}
 -\frac7{12}\,{\mu_2}^2{\mu_5}^2+9\,\mu_2\mu_{12}
  +\frac53\,\mu_5\mu_9+\frac{11}3\,\mu_6\mu_8 & \frac23\,\mu_2\mu_9 &
 -\frac13\,\mu_2\mu_8 \\[3pt]
  \frac23\,\mu_2\mu_9 &        \frac43\,\mu_2\mu_6 +\frac43\,\mu_8
  & -\frac12\,\mu_5\mu_2 \\[3pt]
-\frac13\,\mu_2\mu_8 & -\frac12\,\mu_5\mu_2 & \mu_6
\end{array}
\right], \\
\alpha_6 &=
\left[
\begin{array}{ccc}
    & \ \ & 1 \\[3pt]
    &     &   \\[3pt]
  1 &     &  
\end{array}
\right], \quad
\beta_6=
\left[
\begin{array}{ccc}
      \mu_6                                 &                      &                  \\[3pt]
 \frac1{24}\,{\mu_2}^2\mu_5 -\frac54\,\mu_9 &                      & -\frac14\,\mu_5  \\[3pt]
\frac13\,\mu_2\mu_8-\frac5{12}\,{\mu_5}^2   & -\frac16\,\mu_5\mu_2 &    
\end{array}
\right] , \\
\gamma_6 &=
\left[
\begin{array}{ccc}
  \left\{\ \substack{
  \frac1{24}\,{\mu_2}^3{\mu_5}^2 +4\,{\mu_2}^2\mu_{12}
  -\frac{29}{12}\,\mu_2\mu_5\mu_9\\
  +\frac43\,\mu_2\mu_6\mu_8-\frac5{12}\,{\mu_5}^2\mu_6 +6\,{\mu_8}^2}
 \  \right\}
& \frac12\,\mu_5\mu_8  &  3\,\mu_{12} \\[8pt]
  \frac12\,\mu_5\mu_8 & \frac83\,\mu_2\mu_8 -\frac13\,{\mu_5}^2 
& \frac1{12}\,{\mu_2}^2\mu_5 -\frac12\,\mu_9 \\[3pt]
    3\,\mu_{12} & \frac1{12}\,{\mu_2}^2\mu_5-\frac12 \,\mu_9 & 2\,\mu_8 
\end{array}
\right] ,\\
\alpha_7 &=
\left[
\begin{array}{ccc}
     & 1 &                \\[3pt]
   1 &   &                \\[3pt]
     &   & \frac13\,\mu_5
\end{array}
\right], \ \ 
\beta_7=
\left[
\begin{array}{ccc}
  \frac23\,\mu_5\mu_2                                                    &                                      &                 \\[3pt]
 \frac1{12}\,{\mu_2}^2\mu_6-\frac76\,\mu_2\mu_8 -\frac{5\,{\mu_5}^2}{12} &                                      & -\frac12\,\mu_6 \\[3pt]
 \frac13\,\mu_2\mu_9 -\frac12\,\mu_5\mu_6                                & -\frac13\,\mu_2\mu_6 +\frac23\,\mu_8 &   
\end{array}
\right] ,\\
\gamma_7 &=
\left[
\begin{array}{ccc}
\gamma_7^{[1,1]}                                 &  2\,\mu_2\mu_{12}+\frac23\,\mu_5\mu_9-\mu_6\mu_8                 &   \frac23\,\mu_5\mu_8                                              \\[3pt]
2\,\mu_2\mu_{12} +\frac23\,\mu_5\mu_9-\mu_6\mu_8 & \frac83\,\mu_2\mu_9 -\frac23\,\mu_5\mu_6                         &  -\frac16\,{\mu_5}^2  +\frac16\,{\mu_2}^2\mu_6-\frac13\,\mu_2\mu_8 \\[3pt]
  \frac23\,\mu_5\mu_8                            & -\frac16\,{\mu_5}^2 +\frac16\,{\mu_2}^2\mu_6-\frac13\,\mu_2\mu_8 &  2\,\mu_9                                                          
\end{array}
\right],\\
& \big(\,\gamma_7^{[1,1]} = -\tfrac5{12}\,\mu_2{\mu_5}^3
-\tfrac12\,\mu_5{\mu_6}^2 +5\,\mu_{12}\mu_5 +\tfrac{14}3\,\mu_8
\mu_9 +\tfrac1{12}\,{\mu_2}^3\mu_5\mu_6
-\tfrac16\,{\mu_2}^2\mu_5\mu_8 -\mu_2\mu_6\mu_9\,\big), \\
\alpha_{10}& =
\left[
\begin{array}{ccc}
  1 & \ \ &                 \\[3pt]
    &     &                 \\[3pt]
    &     & \frac13\,\mu_8
\end{array}
\right], \ \
\beta_{10}=
\left[
\begin{array}{ccc}
       -\frac13\,\mu_2\mu_8                                &                      &                 \\[3pt]
 \frac1{24}\,{\mu_2}^2\mu_9 -\frac5{12}\,\mu_5\mu_8        &                      & -\frac14\,\mu_9 \\[3pt]
 \mu_2\mu_{12}-\frac5{12}\,\mu_5\mu_9 +\frac13\,\mu_6\mu_8 & -\frac16\,\mu_2\mu_9 &   
\end{array}
\right] , \\
\gamma_{10}&=
\left[
\begin{array}{ccc}
  \gamma_{10}^{[1,1]}  & 2\,\mu_{12}\mu_5 -\frac56\,\mu_8\mu_9
  & \frac23\,{\mu_8}^2 \\[3pt]
2\,\mu_{12}\mu_5-\frac56\,\mu_8\mu_9 & 4\,\mu_2\mu_{12}-\frac13\,\mu_5\mu_9 
& -\frac16\,\mu_5\mu_8+\frac1{12}\,{\mu_2}^2\mu_9   \\[3pt]
  \frac23\,{\mu_8}^2  &  -\frac16\,\mu_5\mu_8+\frac1{12}\,{\mu_2}^2\mu_9
  & 3\,\mu_{12}
\end{array}
\right] ,
\\
& \Big(\,\gamma_{10}^{[1,1]}
=\tfrac1{24}\,{\mu_2}^3\mu_5\mu_9+{\mu_2}^2{\mu_8}^2 
-\tfrac5{12}\,\mu_2{\mu_5}^2\mu_8 +3\,\mu_2\mu_6\mu_{12} -\tfrac76\,\mu_2{\mu_9}^2\\ 
        &\hskip 150pt -\tfrac{5}{12}\,\mu_5\mu_6\mu_9 +\tfrac13\,{\mu_6}^2\mu_8
          +9\,\mu_8\mu_{12}\,\Big).
\end{align*}
}\relax 
\subsection{The sigma function for the \texorpdfstring{\((3,4)\)}{Lg}-curve }
\label{section_34-sigma} 
\noindent
Following \cite{bl_2005} and from the conditions 
\textbf{IC1}, \textbf{IC2} of (\ref{heat_eq_for_sigma}), 
but using the Hurwitz series form, the sigma function is of the form
{\small 
\begin{align*}
\sigma&(u_5,u_2,u_1)
=\sum_{\ell,m,n_2,n_5,n_6,n_8,n_9,n_{12}}\Big[\,b(\ell,m,n_2,n_5,n_6,n_8,n_9,n_{12})\,
{u_1}^5\left(\frac {u_5}{{u_1}^5} \right)^{\ell}\\
&\cdot\left(\frac{u_2}{{u_1}^2}\right)^{m}
\left(\mu_2{u_1}^2\right)^{n_2}
\left(\mu_5{u_1}^5\right)^{n_5}
\left(\mu_6{u_1}^6\right)^{n_6}
\left(\mu_8{u_1}^8\right)^{n_8}
\left(\mu_9{u_1}^9\right)^{n_9}
\left(\mu_{12}{u_1}^{12}\right)^{n_{12}}\\
&\hskip 40pt\big/\big(\,\ell!\,m!\,(5-5\ell-2m+2n_2+5n_5+6n_6+8n_8
+9n_9+12n_{12})!\,\big)\Big].
\end{align*}
}
for the \((3,4)\)-curve. 
If we define 
\[
k=5-5\ell-2m+2n_2+5n_5+6n_6+8n_8+9n_9+12n_{12},
\]
we can rewrite the above expression as
\begin{equation*}
\begin{aligned}
\sigma(u_5,u_2,u_1)
=\sum\,b(\ell,&m,n_2,n_5,n_6,n_8,n_9,n_{12})\\
&\cdot\frac{{\mu_2}^{n_2}
      {\mu_5}^{n_5}
      {\mu_6}^{n_6}
      {\mu_8}^{n_8}
      {\mu_9}^{n_9}
      {\mu_{12}}^{n_{12}}
    \,{u_5}^{\ell}
      {u_2}^m
      {u_1}^k}
      {\ell!\,m!\,k!},
\end{aligned}
\end{equation*}
where we require all the integer indices 
\(k\), \(\ell\), \(m\), \(n_2\), \(n_5\), \(n_6\), \(n_8\), \(n_9\), 
\(n_{12}\) to be
non-negative.  Note that the \(u\)-weight of this expression is
\(k_0=5+2n_2+5n_5+6n_6+8n_8+9n_9+10n_{12}\), 
which does not depend on \(\ell\) or \(m\). 
(Note also that \(k=k_0-5\ell-2m\).)   
For fixed \(n_2\), \(n_5\), \(n_8\), \(n_6\), \(n_9\), \(n_{12}\ge0\), 
\(k_0\ge0\) is fixed, and for non-negative $k$, we require 
\(\ell=0\), \(\dots\)\,, \(\lfloor k_0/5\rfloor\), 
\(m=0\), \(\dots\)\,, \(\lfloor (k_0-5\ell)/2 \rfloor\).  
In addition, we can use the condition that $\sigma$ is an odd function,
$\sigma(-u)=-\sigma(u)$; this tells us that if $k_0$ is even(odd) then
we should restrict ourselves to $m$ even(odd) respectively.
\par
If we insert this ansatz into the equation for \((L_0-H_0)\sigma=0\),
we get an expression which is identically zero, whatever the values
for the $b(\ell,m,n_2,n_5,n_6,n_8,n_9,n_{12})$.  If we insert the
ansatz into the equation for \((L_3-H_3)\sigma=0\), we get (after some
algebra) the recurrence relation shown below, involving 34 terms
(compare the equations on p.68 of \cite{bl_2005} for the genus 2
case). 
We can structure the relation by the {\em weight} of each $b$ 
coefficient (more precisely by the weight of the corresponding term in
the sigma expansion).
\par
Contrarily to the \((2,3)\)-, \((2,5)\)-, \((2,7)\)-curves, we could
not find any approach for the \((3,4)\)-curve to prove Hurwitz integrality
of the expansion of \(\sigma(u)\).
\par
We call the recurrence relation, generated from
\((L_3\,{-}\,H_3)\sigma=0\), \(R_3\):
\begin{align*}
24\,b({\ell},m+1,n_2,&n_5,n_6,n_8,n_9,n_{12})\\
&+48m\,b({\ell}+1,m-1,
n_2,n_5,n_6,n_8,n_9,n_{12}) \\
=-12(4-k)(3-k)&\,b({\ell},m,n_2,n_5-1,n_6,n_8,n_9,n_{12}) \\
-36(n_8+1)&\,b({\ell},m,n_2,n_5-1,n_6-1,n_8+1,n_9,n_{12}) \\
+4(n_9+1)&\,b({\ell},m,n_2-3,n_5,n_6-1,n_8,n_9+1,n_{12}) \\
+2(n_{12}+1)&\,b({\ell},m,n_2-3,n_5,n_6,n_8,n_9-1,n_{12}+1)\\
+4(n_5+1)&\,b({\ell},m,n_2-4,n_5+1,n_6,n_8,n_9,n_{12}) \\
-8(n_9+1)&\,b({\ell},m,n_2-2,n_5,n_6,n_8-1,n_9+1,n_{12}) \\
+2(n_8+1)&\,b({\ell},m,n_2-3,n_5-1,n_6,n_8+1,n_9,n_{12}) \\
-96(n_5+1)&\,b({\ell},m,n_2-1,n_5+1,n_6-1,n_8,n_9,n_{12}) \\
-12(n_8+1)&\,b({\ell},m,n_2-1,n_5,n_6,n_8+1,n_9-1,n_{12}) \\
-12(n_6+1)&\,b({\ell},m,n_2-2,n_5-1,n_6+1,n_8,n_9,n_{12}) \\
-4(n_{12}+1)&\,b({\ell},m,n_2-1,n_5-1,n_6,n_8-1,n_9,n_{12}+1) \\
-4(n_9+1)&\,b({\ell},m,n_2-1,n_5-2,n_6,n_8,n_9+1,n_{12}) \\
+120(n_2+1)&\,b({\ell},m,n_2+1,n_5-1,n_6,n_8,n_9,n_{12}) \\
-12(3-k)&\,\underline{b({\ell},m+1,n_2-1,n_5,n_6,n_8,n_9,n_{12})} \\
-24m(3-k)&\,b({\ell},m-1,n_2,n_5,n_6-1,n_8,n_9,n_{12}) \\
-8m(m-1)&\,b({\ell},m-2,n_2-1,n_5-1,n_6,n_8,n_9,n_{12}) \\
+24{\ell}(3-k)&\,b({\ell}-1,m,n_2,n_5,n_6,n_8,n_9-1,n_{12}) \\
-60{\ell}({\ell}-1)&\,b({\ell}-2,m,n_2,n_5-1,n_6,n_8-1,n_9,n_{12}) \\
+4m(3-k)&\,b({\ell},m-1,n_2-3,n_5,n_6,n_8,n_9,n_{12}) \\
-{\ell}({\ell}-1)&\,b({\ell}-2,m,n_2-4,n_5-1,n_6,n_8,n_9,n_{12}) \\
+4{\ell}({\ell}-1)&\,b({\ell}-2,m,n_2-2,n_5,n_6,n_8,n_9-1,n_{12}) \\
+20{\ell}({\ell}-1)&\,b({\ell}-2,m,n_2-1,n_5-1,n_6-1,n_8,n_9,n_{12}) \\
+24{\ell}m&\,b({\ell}-1,m-1,n_2-1,n_5,n_6,n_8-1,n_9,n_{12}) \\
+8m&\,b({\ell},m-1,n_2-2,n_5,n_6,n_8,n_9,n_{12}) \\
+4{\ell}&\,b({\ell}-1,m,n_2-1,n_5-1,n_6,n_8,n_9,n_{12}) \\
-2{\ell}&\,b({\ell}-1,m+1,n_2-3,n_5,n_6,n_8,n_9,n_{12}) \\
+36{\ell}&\,b({\ell}-1,m+1,n_2,n_5,n_6-1,n_8,n_9,n_{12}) \\
-36(n_{12}+1)&\,b({\ell},m,n_2,n_5,n_6-1,n_8,n_9-1,n_{12}+1) \\
-72(n_9+1)&\,b({\ell},m,n_2,n_5,n_6-2,n_8,n_9+1,n_{12}) \\
+288(n_9+1)&\,b({\ell},m,n_2,n_5,n_6,n_8,n_9+1,n_{12}-1)\\
+192(n_5+1)&\,b({\ell},m,n_2,n_5+1,n_6,n_8-1,n_9,n_{12}) \\
+216(n_6+1)&\,b({\ell},m,n_2,n_5,n_6+1,n_8,n_9-1,n_{12}). 
\end{align*}
Note the two expressions on the left hand side, which are the highest weight
terms, at weight $W=2n_2+5n_5+6n_6+8n_8+9n_9+12n_{12}$.  
The next highest weight term (underlined) is of weight $W-2$, 
and the lowest weight terms are of weight \(W-13\). 
\relax
\indent
Putting the ansatz into the equation for \((L_4-H_4)\sigma\) 
we get another recurrence with 27 terms, which we call \(R_4\)\,: 
\begin{align*}
-12\,b({\ell},m+2,n_2,&n_5,n_8,n_6,n_9,n_{12})\\
&+24(4-k)\,b({\ell}+1,m,n_2,n_5,n_6,n_8,n_9,n_{12}) \\
=4&\,\underline{b({\ell},m,n_2-1,n_5,n_6,n_8,n_9,n_{12})} \\
+16{\ell}m&\,b({\ell}-1,m-1,n_2-1,n_5,n_6,n_8,n_9-1,n_{12}) \\
-14{\ell}&\,b({\ell}-1,m+1,n_2-1,n_5-1,n_6,n_8,n_9,n_{12}) \\
+40{\ell}&\,b({\ell}-1,m,n_2,n_5,n_6,n_8-1,n_9,n_{12}) \\
+16m&\,b({\ell},m-1,n_2,n_5-1,n_6,n_8,n_9,n_{12}) \\
-7{\ell}({\ell}-1)&\,b({\ell}-2,m,n_2-2,n_5-2,n_6,n_8,n_9,n_{12}) \\
+108{\ell}({\ell}-1)&\,b({\ell}-2,m,n_2-1,n_5,n_6,n_8,n_9,n_{12}-1) \\
+44{\ell}({\ell}-1)&\,b({\ell}-2,m,n_2,n_5,n_8-1,n_6-1,n_9,n_{12}) \\
+12(5-k)(4-k)&\,b({\ell},m,n_2,n_5,n_6-1,n_8,n_9,n_{12}) \\
+8{\ell}(4-k)&\,b({\ell}-1,m,n_2-1,n_5,n_6,n_8-1,n_9,n_{12}) \\
+20{\ell}({\ell}-1)&\,b({\ell}-2,m,n_2,n_5-1,n_6,n_8,n_9-1,n_{12}) \\
-288(n_8+1)&\,b({\ell},m,n_2,n_5,n_6,n_8+1,n_9,n_{12}-1) \\
-64(n_{12}+1)&\,b({\ell},m,n_2,n_5,n_6,n_8-2,n_9,n_{12}+1) \\
-40(n_6+1)&\,b({\ell},m,n_2,n_5-2,n_6+1,n_8,n_9,n_{12}) \\
-216(1+n_5)&\,b({\ell},m,n_2,n_5+1,n_6,n_8,n_9-1,n_{12}) \\
-80(n_6+1)&\,b({\ell},m,n_2-1,n_5,n_6+1,n_8-1,n_9,n_{12}) \\
  -4(\ell-1-n_9-8&n_5-2n_6+k+2m-2n_2)\\
  \cdot&\,b({\ell},m,n_2-2,n_5,n_6,n_8,n_9,n_{12}) \\
-104(n_9+1)&\,b({\ell},m,n_2,n_5-1,n_6,n_8-1,n_9+1,n_{12}) \\
+14(n_8+1)&\,b({\ell},m,n_2-1,n_5-2,n_6,n_8+1,n_9,n_{12}) \\
+14(n_{12}+1)&\,b({\ell},m,n_2-1,n_5-1,n_6,n_8,n_9-1,n_{12}+1) \\
+28(n_9+1)&\,b({\ell},m,n_2-1,n_5-1,n_6-1,n_8,n_9+1,n_{12}) \\
-144(n_2+1)&\,b({\ell},m,n_2+1,n_5,n_6-1,n_8,n_9,n_{12}) \\
+16m(m-1)&\,b({\ell},m-2,n_2,n_5,n_6,n_8-1,n_9,n_{12}) \\
+16m(m-1)&\,b({\ell},m-2,n_2-1,n_5,n_6-1,n_8,n_9,n_{12}) \\
+12m(4-k)&\,b({\ell},m-1,n_2-1,n_5-1,n_6,n_8,n_9,n_{12}).
\end{align*}
As for \(R_3\), the two expressions on the left hand side, are the
highest weight terms, at weight \(W=2n_2+5n_5+6n_6+8n_8+9n_9+12n_{12}\).
The next highest weight terms are of weight \(W-2\), and the lowest
weight terms are of weight \(W-14\).
\par
We see that the two recurrence relations have the same terms in \(b\).
Hence we can take linear combinations to get two relations, 
each with only one leading term at weight \(W\)
 \begin{align*}
   S_{3,4}: \,b(\ell, m, n_2, n_5, n_8, n_6, n_9, n_{12}) 
& =\frac1{k-m+1} 
   \left(\text{lower weight terms}\right)\ \,(m\neq0),\\
  T_{3,4}: \,b(\ell, m, n_2, n_5, n_8, n_6, n_9, n_{12}) 
& =\frac1{k-m}\left(\text{lower weight terms}\right) \ \ (l\neq0). 
\end{align*}
These \(S_{3,4}\) and \(T_{3,4}\) connect the left hand side 
with terms of relative weight \(-2\)
and lower, down to \(-14\).  
In addition we have other relations from the equations
\((L_6-H_6)\sigma=0\), 
\((L_7-H_7)\sigma=0\), 
and \((L_{10}-H_{10})\sigma=0\) that
\begin{align*}
R_6: \,b(\ell, m, n_2, n_5, n_8, n_6, n_9, n_{12})  & = 
    \mbox{\lq\lq lower weight terms''},\\
  R_7: \ b(\ell, m, n_2, n_5, n_8, n_6, n_9, n_{12}) & = 
    \mbox{\lq\lq lower weight terms''},\\
  R_{10}: \ b(\ell, m, n_2, n_5, n_8, n_6, n_9, n_{12}) & =
    \mbox{\lq\lq lower weight terms''}, 
\end{align*}
respectively. 
Here the right hand sides are linear in the coefficients \(b\) 
with coefficients at most quadratic in \(k\), \(\ell\), \(m\), 
\(n_2\), \(n_5\), \(n_8\), \(n_6\), \(n_9\), \(n_{12}\) over the rationals
but each denominator is a divisor of \(24\). 
\par
\(R_6\), \(R_7\), \(R_{10}\) have a total of 37, 47, 42 terms
respectively and connect the left hand side with terms of relative
weight \(-5\), \(-5\), \(-8\) and lower, down to \(-16\), \(-17\),
\(-20\) respectively.
\par
Ideally we would like to proceed as follows.  Suppose we have already
calculated the \(b\) coefficients at weight \(W-2\).  Then we would
like to use one of the above to calculate each coefficient at weight
\(W\).  We could proceed in this manner to calculate coefficients at
successive weight levels to the required number of terms.  However
this approach needs some modification.  Recall that the weight does
not depend on \(\ell\) or \(m\).  Clearly if \(\ell>1\) we can use
\(R_6\), and if \(\ell=1\), \(m>0\), we can use \(R_7\).  
Similarly if \(\ell=1\), \(k>0\), we can use \(R_6\).  
A short calculation shows
that if \(\ell=1\), one of these two possibilities holds except in the
special case \(\ell=1\), \(m=n_2=n_5=n_8=n_6=n_9=n_{12}=0\) which is
covered later.  For the case \(\ell=0\) we cannot use \(R_6\),
\(R_7\), \(R_{12}\).  
If \,\(m\ne 0\) \,{and}\, \(m\ne (k+1)\)\, we can use \(S_{3,4}\).  
All the possibilities considered so far will reduce the weight by 2.  
There remain the cases \(m=0\) and \(m=(k+1)\) to deal with. \relax
\newline
\indent
The case \(\ell=0\), \(m=0\) is handled as follows.  
Take \(48(k-m)\,T_{3,4}\) : 
\begin{align*}
48&(m-k)\,b({\ell},m,n_2,n_5,n_6,n_8,n_9,n_{12})\\
&=8\,b({\ell}-1,m,n_2-1,n_5,n_6,n_8,n_9,n_{12})\\
&\ \ \ \ +12k\,b({\ell}-1,m+2,n_2,n_5,n_6,n_8,n_9,n_{12})+\cdots.
\end{align*}
Shifting by \(n_2\,{\rightarrow}\,n_2\,{+}\,1\), 
the first term in the right hand side is expressed as
\begin{align*}
b({\ell}-1,&m,n_2,n_5,n_6,n_8,n_9,n_{12})\\
&=-6(m-k-2)\,b({\ell},m,n_2+1,n_5,n_6,n_8,n_9,n_{12})\\
&\quad -\tfrac32(k+2)\,b({\ell}-1,m+2,n_2,n_5,n_6,n_8,n_9,n_{12})+\cdots.
\end{align*}
On the right hand side\ we now have two terms of non-negative relative
weight as above of relative weight \(+2\) which comes from 
the underlined term in \(R_3\) and of relative weight \(0\) 
which comes from the underlined term in \(R_4\).  
Putting \(\ell=1\), \(m=0\), we have, 
say \,\(T_{3,4}^{(0)}\), \,that
\begin{align*}
b(&0,0,n_2,n_5,n_6,n_8,n_9,n_{12})\\
&=-6\,(k-3+5\,\ell+2\,m)\,
b(1,0,n_2+1,n_5,n_6,n_8,n_9,n_{12})\\ 
&\qquad -\tfrac32\,(k-3+5\,\ell+2\,m)\,b(0,2,n_2,n_5,n_6,
n_8,n_9,n_{12})+\cdots.
\end{align*}
The first term in the right hand side has \(\ell=1\), \(m=0\), 
and   \(k=5+2n_2+\dots>0\).  
Hence we can apply  \(R_7\) to this term to give 
a term with maximum relative weight \(+2-5=-3\).  
The second term has \(\ell=0\), \(m=2\), 
and  \(k=1+2n_2+5n_5+\cdots\) \ so \(k+1>2\) and hence \(k+1\ne m\).  
For this term we can apply  \(S_{3,4}\) to produce a term 
of maximum  relative weight \(0-2=-2\).  
Hence both terms of weight  \(\geq 0\)  can be expressed 
as terms of relative weight \(\leq -2\), so our chain  eventually 
decreases in weight. \relax
\indent
The case \(\ell=0\), \(m=(k+1)\) is treated as follows.  
Take \(R_3\), shift by \(m\rightarrow m-1\), 
and set \(\ell=0\) to get
\begin{align*}
  R_3^{(0)}: \ &b(0,m,n_2,n_5,n_6,n_8,n_9,n_{12})\\
  &=-2(m-1)\,b(1,m-2,n_2,n_5,n_6,n_8,n_9,n_{12})
  +\mbox{\lq\lq lower weight terms''}.
\end{align*}
Now the first term on the right, \(b(1,m-2,n_2,n_5,n_6,n_8,n_9,n_{12})\), 
is of the same weight as the term on the left.  
Write this as
\begin{equation*}
b(1,m',n_2,n_5,n_6,n_8,n_9,n_{12}), 
\end{equation*}
with corresponding \(k\)-value  \(k'\).  
If \(\mathrm{min}(k', m')<0\) then this term is zero as discussed above.  
If \(k'=0\), \(m'=0\), it is easy to show that \(m'=n_2=n_5=n_8=n_6=n_9=n_{12}=0\), 
and this term \(b(1,0,0,0,0,0,0,0)\) cannot be reduced further.  
Otherwise one or both of \(k'=0\), \(m'\) is positive, so we can apply \(R_7\) or
\(R_6\) to reduce the term to terms of relative weight \(\leq -2\), so
our chain terminates or decreases in weight.  These choices, plus the
requirement discussed above that
\(b(\ell,m,n_2,n_5,n_6,n_8,n_9,n_{12})=0\) if any of the
\(\{k,\ell,m,n_2,n_5,n_6,n_8,n_9,n_{12}\}\) are negative, define all
the \(b(\ell,m,n_2,n_5,n_6,n_8,n_9,n_{12})\) in terms of the so-far undefined 
\,\(b(1,0,0,0,0,0,0,0)\).  
Therefore the solution of the system
\begin{equation*}
(L_j-H_j)\sigma(u)=0\ \ (j=0,\ 3,\ 4,\ 6,\ 7,\ 10)
\end{equation*}
is of dimension one. 
Choosing \(b(1,0,0,0,0,0,0,0)=1\), we summarise with 
\(k=5-5\ell-2m+2n_2+5n_5+8n_8+6n_6+9n_9+12n_{12}\) \  
as defined as above as follows:
\begin{equation*}
\begin{aligned}
b(\ell, &m, n_2,n_5, n_8, n_6, n_9, n_{12}) = \\
&=\begin{cases}
  \ 0      & \quad \mbox{if \ \(\mathrm{min}(k,\ell,m,n_2,n_5,n_6,n_8,n_9,n_{12})<0\)},  \\
\ 1  & \quad \mbox{if \ \(\ell=1\), \(m=n_2=n_5=n_8=n_6=n_9=n_{12}=0\)}, \\
\ \mathrm{rhs}(R_{10})         & \quad \mbox{if \ \(\ell>1\)},\\
\ \mathrm{rhs}(R_7)            & \quad \mbox{if \ \(\ell>0\), \(m>0\)},\\
\ \mathrm{rhs}(R_6)            & \quad \mbox{if \ \(\ell>0\), \(k>0\)},\\
\ \mathrm{rhs}(S_{3,4})        & \quad \mbox{if \ \(\ell=0\), \(m\ne 0\)
  and \(m\ne(k+1)\)},\\
\ \mathrm{rhs}(T_{3,4}^{(0)})  & \quad \mbox{if \ \(l=m=0\)},\\
\ \mathrm{rhs}(R_3^{(0)})      & \quad \mbox{if \ \(\ell=0\) and \(m=(k+1)\)}.
\end{cases}
\end{aligned}
\end{equation*}
We have used this to calculate the
terms in the sigma series up to weight 40 in \(\{u_j\}\), 
or equivalently, weight 35 in the \(\{\mu_i\}\).
The first few terms of the sigma expansion are given as
follows (up to a constant multiple):
\begin{align*}
\sigma({u_5}&,{u_2},{u_1}) =
{u_5}+\frac{6\,u_1^5}{5!}
-2\frac{u_1{u_2}^2}{2!} 
-2\mu_2\frac{u_1^3{u_2}^2}{2!\,3!}
+30\mu_2\frac{{u_1}^7}{7!}
-2{\mu_2}^2\frac{u_1{u_2}^4}{4!}\\
&-2{\mu_2}^2\frac{{u_1}^5{u_2}^2}{2!\,5!}
+126\,{\mu_2}^2\frac {{u_1}^9}{9!}
+24\mu_5\frac{{u_1}^8u_2}{8!}
-\mu_5\frac{u_5u_2{u_1}^3}{3!}
+8\mu_5\frac{{u_2}^5}{5!}\\
&-2\mu_6\frac{2\,u_5{u_2}^2{u_1}^2}{2!\,2!}
+6\mu_6\frac{u_5{u_1}^6}{6!}
 +24\mu_6\frac{{u_1}^7{u_2}^2}{2!\,7!}
-2{\mu_3}^3\frac {{u_1}^7{u_2}^2}{2!\,7!}\\
&-2{\mu_2}^3\frac{{u_1}^3{u_2}^4}{4!\,3!}
+432\mu_6\frac{{u_1}^{11}}{11!}
+510{\mu_2}^3\frac{{u_1}^{11}}{11!}\\
&-\mu_2\mu_5{u_1}^5\frac {u_5u_2}{5!}
-\mu_2\mu_5\frac{u_5{u_2}^3u_1}{3}
+288\mu_2\mu_5\frac{{u_1}^{10}u_2}{10!}+\cdots.
\end{align*} 
\newpage
\bibliographystyle{amsplain}
\vskip -20pt
\bibliography{refers}
\end{document}